\documentclass[11pt]{article}
\usepackage[utf8]{inputenc}
\usepackage{lipsum}
\usepackage{fullpage}
\usepackage[numbers,sort&compress]{natbib}
\usepackage{float}
\usepackage[normalem]{ulem}
\usepackage{geometry}
\usepackage{amsmath,mathrsfs,amsbsy,color,graphicx,bm,amsthm,amsfonts,dsfont}
\usepackage{xcolor}
\usepackage{qcircuit}
\usepackage{booktabs}
\usepackage{array}
\usepackage{physics}
\usepackage{amsmath,amssymb, amsthm}
\usepackage{thmtools}
\usepackage{bm}
\usepackage{graphicx}
\usepackage{braket,dsfont}
\usepackage{color}
\usepackage[10pt]{moresize}
\usepackage{mathrsfs}
\usepackage{multirow}
\usepackage{hyperref}
\usepackage{changes}
\usepackage{mathdots}
\usepackage{tikz}
\usetikzlibrary{arrows.meta,positioning,decorations.pathreplacing,calc,fit,shapes.geometric}

\newtheorem{theorem}{Theorem}
\newtheorem{lemma}{Lemma}

\newtheorem{observation}{Observation}
\newtheorem{definition}{Definition}

\newcommand{\Wg}{\mathrm{Wg}}

\newcommand{\id}{\mathbb{I}}
\newcommand{\Cl}{\mathrm{Cl}}
\newcommand{\Sp}{\mathrm{Sp}}
\newcommand{\SL}{\mathrm{SL}}

\newcommand{\EL}{\mathrm{EL}}
\newcommand{\Alt}{\mathrm{Alt}}
\newcommand{\Sym}{\mathrm{Sym}}
\newcommand{\Mat}{\mathrm{Mat}}
\newcommand{\E}{\mathbb{E}}
\newcommand{\F}{\mathbb{F}}

\newcommand{\U}{\mathrm{U}}

\newcommand{\diag}{\mathrm{diag}}

\newcommand{\CPZPC}{\mathrm{CPZPC}}

\newcommand{\av}{\mathrm{av}}
\newcommand{\ph}{\mathrm{ph}}
\newcommand{\Span}{\operatorname{span}}

\newcommand{\cE}{\mathcal{E}}
\newcommand{\cO}{\mathcal{O}}
\newcommand{\cK}{\mathcal{K}}

\newcommand{\cket}[1]{\vert #1 )}
\newcommand{\cbra}[1]{( #1 \vert}
\newcommand{\cbraket}[2]{( #1 \vert #2 )}
\newcommand{\cketbra}[2]{\vert #1 )\!( #2 \vert}

\title{(Almost) quadruply optimal unitary designs in 1D}
\author{Guoding Liu \thanks{Center for Quantum Information, Institute for Interdisciplinary Information Sciences, Tsinghua University, Beijing, \texttt{lgd22@mails.tsinghua.edu.cn}} \and Jonas Helsen \thanks{Centrum Wiskunde \& Informatica (CWI) and QuSoft, Amsterdam, \texttt{jonas@cwi.nl}}}
\date{\today}

\begin{document}

\maketitle

\begin{abstract}
We construct $n$-qubit approximate unitary $k$-designs in 1D systems, achieving circuit depth $\cO(\log(n/\varepsilon) + k\log k)$ with relative error $\varepsilon$ and requiring $\cO(nk\log k)$ magic gates. This matches existing lower bounds $\Omega(\log(n/\varepsilon) + k)$ for circuit depth, and $\widetilde{\Omega}(nk)$ for the required number of $T$ gates, up to a $\log k$ factor, achieving simultaneous near-optimality in all parameters. 
Our construction is based on a combination and refinement of two existing results. We reduce the required magic block size for breaking Clifford symmetries in the magic-augmented circuit construction of Zhang et al.~(Ref.~\cite{Zhang2026magic}) from $\cO(k\log k)$ to $\cO(\log k)$. We also improve the breakthrough construction of Chen et al.~(Ref.~\cite{chen2024incompressibilityspectralgapsrandom}), to construct a generating set of 1D local constant-depth circuits for the unitary group with a constant spectral gap, making $\cO(\log k)$-local random unitaries realizable in depth $\cO(k\log k)$. As a by-product, we provide a constant-size 1D-local generating set for the Clifford group, which we expect to be of independent interest. Combining the two results with the gluing lemma, we prove the final result. 

\end{abstract}

\section{Introduction}
Random unitaries are a foundational tool and object of study in quantum information science and quantum physics. Theoretically, they serve as indispensable tools for investigating quantum many-body dynamics~\cite{Nahum2017Random,Nahum2018randomness} and quantum gravity~\cite{Patrick2007Blackhole,Almheiri2015QEC,Pastawski2015Holographic}. Practically, they support a wide range of quantum information processing tasks, including quantum learning~\cite{huang2020shadow}, quantum metrology~\cite{Zhou2026Metrology}, device benchmarking~\cite{Elben2023toolbox,helsen2022general}, quantum supremacy demonstrations~\cite{arute2019supremacy}, and quantum cryptography~\cite{Ji2018Pseudorandom}.

Given their ubiquitous applications, significant effort has been devoted to synthesizing random unitaries with minimal resource overhead. However, generating true Haar-random unitaries requires circuit depth or ancillary qubit counts that scale exponentially with the qubit number $n$. To circumvent this bottleneck, unitary $k$-designs are employed to mimic the Haar measure up to the $k$-th moment, rendering them statistically indistinguishable from Haar-random unitaries under $k$ queries. For most practical protocols, an approximate unitary design suffices. Specifically, in any experiment querying a random unitary at most $k$ times, an approximate unitary $k$-design with relative error $\varepsilon$ guarantees an output accuracy within $\varepsilon$. While alternative error measures exist, such as additive or measurable errors, relative error provides the strongest guarantee when queries are restricted to $U$ without access to its inverse~\cite{schuster2025strongrandomunitariesfast}.

Consequently, optimizing the circuit depth required to generate approximate unitary $k$-designs across various geometric constraints has become a central objective. Efforts have focused in particular on one-dimensional (1D) nearest-neighbor architectures, as 1D geometries represent the most experimentally accessible setup, establish a fundamental baseline for all spatial dimensions, and serve as a conceptual building block for higher-dimensional generalizations.

A prominent line of work has aimed to settle this 1D depth requirement. Initial spectral gap analyses of 1D brickwork circuits proved that a polynomial depth was sufficient~\cite{brandao2016local} and over several years~\cite{hunter2019unitary, haferkamp2022random,haferkamp2023efficient,harrow2023approximate, metger2024simple,haah2025efficient} this was improved to a depth of $\cO(nk\log^7k)$~\cite{chen2024incompressibilityspectralgapsrandom}. A recent breakthrough further refined this bound, establishing an exact $\cO(nk)$ scaling free of logarithmic factors~\cite{baer2026randomunitarycircuitsconstant}. While $\cO(nk)$ represents the optimal depth achievable via standard spectral gap and light-cone techniques, the recently introduced gluing lemma dramatically improved this upper bound to $\widetilde{\cO}(k)\log(n/\varepsilon)$~\cite{Schuster2025Gluing,laracuente2026approximate}. Furthermore, leveraging magic-augmented circuits in Ref.~\cite{Zhang2026magic} achieved a depth bound of $\cO(\log(n/\varepsilon) + 2^{\cO(k\log k)})$ in the $k = o(\sqrt{n})$ regime, successfully decoupling the system size $n$ from the design order $k$ in the depth scaling.

Despite these advances, the fundamental question of constructing depth-optimal unitary designs in 1D remains open. In Ref.~\cite{cui2025unitarydesignsnearlyoptimal}, it was proved that even with $\cO(n)$ ancillas, generating a 1D approximate unitary $k$-design with additive error $\varepsilon$ requires a depth of $\Omega(\log(n/\varepsilon) + k)$. Note that earlier constructions have saturated this lower bound in each parameter separately, but simultaneously achieving optimal depth scaling in both $n$ and $k$ while keeping them completely decoupled has thus far eluded existing techniques.

In this paper, we nearly close this gap for 1D architectures by constructing a 1D-local, $\varepsilon$-relative-error approximate unitary $k$-design of depth $\cO(\log n + \log(1/\varepsilon) + k\log k)$ without ancillary qubits. This result matches the known lower bound across all three parameters simultaneously, up to a single $\log k$ factor in the ``design direction" only. Moreover, evaluating the non-Clifford resources (i.e. the number of $T$ gates) in our construction yields near-optimality in the ``magic" direction: our design requires only $\cO(nk\log k)$ $T$-gates, completely independent of the error parameter $\varepsilon$. This further improves upon the best previous upper bound for $t$-doped Clifford circuits, $\cO(\log^2(k)(nk + \log(1/\varepsilon)))$~\cite{Haferkamp2022randomquantum,Haferkamp2023Designs}, and moves closer to the lower bound of $\Omega(nk / (\log n \log k + \log^2 k))$~\cite{Leone2026NonClifford}. Consequently, our construction offers an exceptionally resource-efficient candidate for minimizing magic overhead, which is an indispensable yet costly resource for universal and fault-tolerant quantum computation. Furthermore, we extend the available regime of the design order $k$ from $\cO(\sqrt{n})$ up to $\Theta(n)$, establishing the feasibility of efficient high-order approximate unitary designs.

Our results rely on reducing the magic block size for the magic-augmented circuits of Ref.~\cite{Zhang2026magic} and constructing a constant-spectral-gap, constant-depth, and 1D-local generating set for the $\ell$-qubit unitary group. Structurally, the magic-augmented circuit consists of a random Clifford operation flanked by two layers of tensor-product $\ell$-local random unitary gates. A crucial parameter in this scheme is the minimal locality $\ell$ required for these random unitaries to eliminate the extra symmetries of the Clifford group relative to the full unitary group. In other words, $\ell$ quantifies the minimal magic needed to break these symmetries. While it was previously known that $\ell = \cO(k\log k)$ is sufficient~\cite{Zhang2026magic}, pushing $\ell$ lower requires new techniques. By analyzing the non-permutation Clifford commutant elements under permutation symmetry, we reduce this required locality from $\ell = \cO(k\log k)$ down to $\ell = \cO(\log k)$, marking a pivotal step in compressing both the depth and magic requirements for approximate unitary designs.

The other key technical step in our construction, which we believe to be of independent interest, is the construction of a generating set for the $\ell$-qubit unitary group with a constant spectral gap, where each generator can be implemented in constant depth on a 1D nearest-neighbor geometry. This generating set is constructed by combining generators of the Clifford group and the permutation group, following the breakthrough work on the so-called ``CPZPC" ensemble~\cite{chen2024incompressibilityspectralgapsrandom}. Because the underlying CPZPC ensemble exhibits a constant spectral gap, establishing the result reduces to finding constant-spectral-gap 1D-local generating sets for both constituent groups. While the permutation group was known to possess a constant spectral gap with respect to Kassabov's generators~\cite{chen2024incompressibilityspectralgapsrandom}, their constant-depth implementation was previously established only for all-to-all architectures or 1D circuits with periodic boundary. Here, by introducing a folded-qubit layout, we demonstrate that all of Kassabov's generators for the permutation group can be synthesized at a constant depth in a generic 1D system. For the Clifford group, we construct a generating set starting from the binary symplectic group and extending it to the full Clifford group via Pauli extension. We explicitly prove that this Clifford expander contains at most $451$ generators, providing a concrete framework for efficiently constructing Clifford circuits.

\subsection{Results and proof overview}
We now present the circuit architecture for our nearly optimal unitary designs, as illustrated in Fig.~\ref{fig:circuitarch}. We begin by partitioning the 1D chain of $n$ qubits into $m$ patches, each comprising $\xi = n/m$ qubits. The construction then consists of a two-layer block-brickwork architecture with blocks of size $2\xi$, where inside each block we apply a $2\xi$-qubit random unitary operation across two adjacent patches. The patch size $\xi$ is chosen to scale as $\xi = \Omega(\log(n/\varepsilon)+k)\geq \log(nk^2/\varepsilon)$, satisfying the conditions required by the gluing lemma~\cite{Schuster2025Gluing}.

Each $2\xi$-qubit block consists of a magic-augmented circuit (as first considered in \cite{Zhang2026magic}), formed by a $2\xi$-qubit random Clifford unitary sandwiched by two layers of tensor-product $\ell$-qubit approximate unitary $k$-designs (with $\ell = c\log(k)$ for a constant $c\geq 8$). Each $\ell$-qubit approximate unitary $k$-design is implemented by repeatedly applying 1D-local generators from the CPZPC ensemble~\cite{chen2024incompressibilityspectralgapsrandom}, as depicted in Fig.~\ref{fig:circuitCPZPC}.

As a remark, throughout this work, we assume for simplicity that the block size $\xi$ is an integer multiple of the locality parameter $\ell$. When $\xi$ is not a multiple of $\ell$, one can partition the $2\xi$-qubit system into three disjoint subsystems $A \sqcup B \sqcup C$, where $|A| = |C| = \xi - \ell \lfloor \xi/\ell \rfloor$ and $|B| = 2\ell \lfloor \xi/\ell \rfloor - \xi$. Under this partition, the combined subsystems $A \sqcup B$ and $B \sqcup C$ each contain a total qubit count that is a multiple of $\ell$. Consequently, magic-augmented circuits can be implemented independently on $A \sqcup B$ and $B \sqcup C$ before using the gluing lemma~\cite{Schuster2025Gluing} to glue them together. This modification incurs at most a constant-factor overhead in the error scaling and does not affect our asymptotic results. By a similar argument, we assume without loss of generality that the total system size $n$ is a multiple of $\xi$.\\

\begin{figure}
\centering
\includegraphics[width=.8\linewidth]{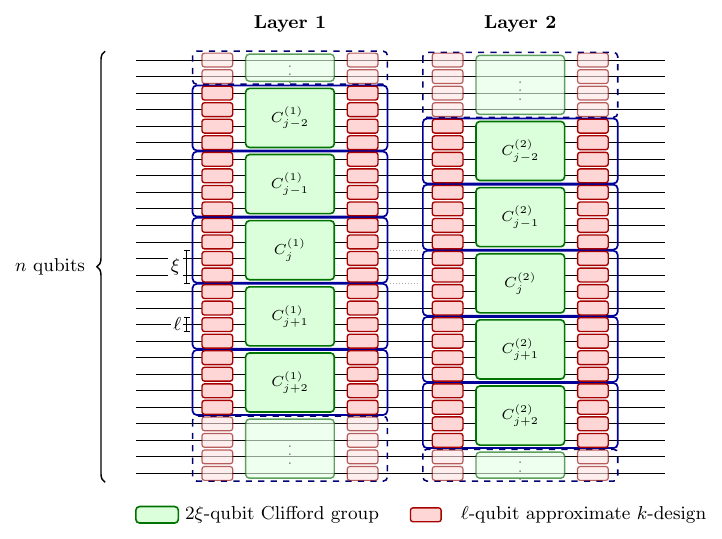}
\caption{Circuit architecture for nearly optimal 1D approximate unitary $k$-designs. The construction employs a two-layer block-brickwork architecture acting on a 1D chain, where each block spans $2\xi$ qubits with an overlap of $\xi$ qubits between adjacent layers. The block size scales as $\xi = \cO(\log(n/\varepsilon)+k)$. Each $2\xi$-qubit block is realized as a magic-augmented circuit, consisting of a random Clifford operation flanked by $\ell$-local $\varepsilon''$-error approximate unitary $k$-designs, where the locality is $\ell = \cO(\log k)$ and the relative error of the local block is set to $\varepsilon'' = 2^{-2\ell k}$.
}
\label{fig:circuitarch}
\end{figure}

\begin{figure}
\centering
\includegraphics[width=.85\linewidth]{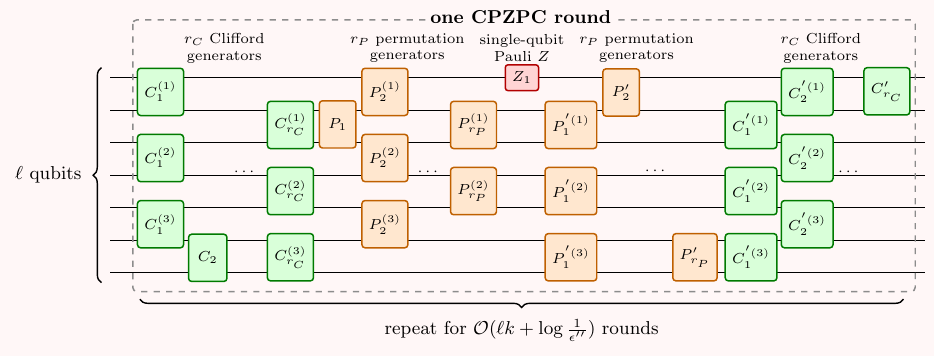}
\caption{Implementation of the $\ell$-local $\varepsilon''$-relative-error approximate unitary $k$-design. An $\ell$-qubit approximate unitary $k$-design with relative error $\varepsilon''$ is constructed by repeating $\cO(\ell k + \log(1/\varepsilon''))$ rounds of CPZPC generators. Each individual round consists of $2r_C$ Clifford group generators, $2r_P$ permutation group generators, and a single Pauli $Z$ gate, where $r_C$ and $r_P$ are fixed constants. All generators are strictly constant-local in a 1D nearest-neighbor layout.}
\label{fig:circuitCPZPC}
\end{figure}

We establish that the proposed architecture yields a nearly optimal approximate unitary $k$-design, as summarized in the following theorem.
\begin{theorem}\label{thm:Mainthm}
For any $n$-qubit system and design order $k = \cO(n)$, an $\varepsilon$-relative-error approximate unitary $k$-design can be explicitly constructed using gates local in a $1D$ geometry and no ancillas, with circuit depth of 
\begin{equation}
    \cO\left( \log \frac{n}{\varepsilon} + k \log k \right).
\end{equation}
Moreover these circuits can be instantiated using $\cO(nk\log k)$ magic $T$ gates.
\end{theorem}

Theorem~\ref{thm:Mainthm} is established by combining the gluing lemma with two core theorems proved in this work. Specifically, we prove that an $\varepsilon'$-relative-error approximate unitary $k$-design on a $\xi$-qubit block can be generated provided the block size satisfies $\xi = \Omega(\log(1/\varepsilon') + k)$, achieving a circuit depth of $\cO(\log(1/\varepsilon') + k + \xi)$. Consequently, setting $\varepsilon'=\varepsilon/n$ and $\xi = \cO(\log(n/\varepsilon)+k)$ and applying the gluing lemma yields our main result. Note that since the patch size satisfies $\xi \le n$, our construction naturally imposes the regime constraint $k = \cO(n)$.

For completeness, we state the gluing lemma below, which guarantees that an $n$-qubit $\varepsilon$-error approximate $k$-design is realized provided each $2\xi$-qubit block forms an $(\varepsilon/n)$-approximate unitary $k$-design with patch size $\xi \geq \log_2(nk^2/\varepsilon)$.

\begin{lemma}[Gluing lemma~\cite{Schuster2025Gluing}]
Given any approximation error $\varepsilon \leq 1$. Suppose each small random unitary in the double-layer blocked ensemble $\mathcal{E}$ is drawn from an $\varepsilon/n$-approximate unitary $k$-design on $2\xi$ qubits with circuit depth $d$. Then $\mathcal{E}$ forms an $\varepsilon$-approximate unitary $k$-design on $n$ qubits with depth $2d$, whenever the local patch size is at least $\xi \geq \log_2(nk^2/\varepsilon)$.
\end{lemma}

Our first main result  (Theorem \ref{thm:MagicAugmented}) demonstrates that the magic-augmented circuit forms an $\varepsilon'$-error approximate unitary $k$-design on $\xi$ qubits provided that $\xi = \Theta(\log(1/\varepsilon') + k)$, the local unitary block locality satisfies $\ell \geq 8\log k$, and the relative error of the $\ell$-local block satisfies $\varepsilon'' \le 2^{-2\ell k}$. Our second main result shows that this $\ell$-qubit $\varepsilon''$-relative-error approximate unitary $k$-design can be synthesized in circuit depth $\cO(\ell k + \log(1/\varepsilon''))$ using CPZPC generators whenever $k = \cO(2^{\ell/6.1})$.

Setting $\ell = \cO(\log k)$, $\varepsilon' = \varepsilon/n$, and $\varepsilon'' = 2^{-2\ell k}$ yields Theorem~\ref{thm:Mainthm}. Specifically, each $\ell$-qubit magic block is implemented in depth $\cO(\ell k) = \cO(k\log k)$ with a $T$-gate count of $\cO(\ell \cdot \ell k) = \cO(k\log^2 k)$, while the random Clifford layer requires depth $\cO(\xi) = \cO(\log(n/\varepsilon) + k)$. Summing over the $m = \cO(n/\ell)$ local patches, the overall circuit achieves a total circuit depth of $\cO(\log(n/\varepsilon) + k\log k)$ with a global $T$-gate count of $\cO((n/\ell) \cdot k\log^2 k) = \cO(nk\log k)$. Below, we formalize these results and detail their proof outlines.

\begin{restatable}{theorem}{ThmMagicAugmented}
\label{thm:MagicAugmented}
For any $\xi$-qubit system and positive integer $k\leq \xi-1$, a magic-augmented circuit composed of a random Clifford gate and two layers of $\ell$-local $\varepsilon''$-error approximate unitary $k$-design will form a $\xi$-qubit approximate unitary $k$-design with error
\begin{equation}
2^{-\Omega(\xi)+\cO(k)},
\end{equation}
provided with $\ell \geq 8\log k$ and $\varepsilon''\leq 2^{-2\ell k}$. Thus, choosing $\xi = \cO(\log 1/\varepsilon'+k)$ will lead to an approximation error $\varepsilon'$.
\end{restatable}

To prove Theorem~\ref{thm:MagicAugmented}, we analyze the magic-augmented quantum channel by expanding it in the basis of Clifford commutant elements. The $k$-fold twirling channel of a random Clifford gate differs from that of a Haar-random unitary operation in two key aspects: the value of the Clifford Weingarten function and the presence of non-permutation Clifford commutant elements. By explicit evaluation, the discrepancy between the Clifford and Haar Weingarten functions is upper-bounded by $2^{-\xi + k}$. To eliminate the non-permutation Clifford commutant elements, we append $\ell$-local random unitaries alongside the Clifford twirl. The prior work established that choosing $\ell = \cO(k\log k)$ suffices to eliminate these non-permutation terms~\cite{Zhang2026magic}. Here, we present a refined analysis focused on Clifford commutant elements nearest to the identity element under permutation symmetry. This allows us to establish that $\ell \geq 8\log k$ suffices for removing all non-permutation commutant elements.\\

The other key ingredient of our construction is a constant spectral gap generating set for the unitary group. We have:

\begin{theorem}\label{thm:CPZPCdesign}
For any $\ell$-qubit system and design order $k$, suppose $k\leq c 2^{\ell/6.1}$ for a sufficiently small constant $c>0$, an $\varepsilon''$-relative-error approximate unitary $k$-design can be explicitly constructed using gates local in a $1D$ geometry and no ancillas, with circuit depth of 
\begin{equation}
\cO\left(\ell k + \log \frac{1}{\varepsilon''} \right).
\end{equation}
\end{theorem}

The proof of Theorem~\ref{thm:CPZPCdesign} relies on constructing a generating set for the CPZPC ensemble that exhibits a constant spectral gap and is implementable in 1D constant depth. The exact CPZPC ensemble, which comprises two random Clifford operations, two random permutations, and a single Pauli $Z$ gate, was previously proved to possess a constant spectral gap~\cite{chen2024incompressibilityspectralgapsrandom} with respect to the Haar measure. To make this construction efficient, we replace the fully random Clifford and permutation operations with the 1D-local generators. Particularly, the random permutation layer can be replaced by a constant number of Kassabov's generators, which can be synthesized in constant depth on a 1D nearest-neighbor geometry, flanked with two layers of bit-flip operators.

Similarly, the random Clifford layer must be replaced by a constant number of 1D-local generators while incurring at most a constant reduction in the spectral gap. To achieve this, we explicitly prove that the Clifford group admits a bounded-size generating set with an $\Omega(1)$ Kazhdan constant, where each generator can be implemented in constant depth in 1D. This structural result, which may be of independent interest, is stated below.

\begin{theorem}[Constant-depth 1D Clifford generators]\label{thm:cliffordgenerator}
For any integer $\ell\geq 1$, there exists an explicit symmetric generating set $S_{\Cl,\ell} \subset \Cl_{\ell}$ such that each element in $S_{\Cl,\ell}$ is a product of single-qubit Pauli, Hadamard ($H$), phase ($S$), $\mathrm{CZ}$, and $\mathrm{CNOT}$ gates with circuit depth $\cO(1)$ in a one-dimensional nearest-neighbor system, $|S_{\Cl,\ell}|\leq 451$, and $\Cl_{\ell}$ has a Kazhdan constant $\kappa_{\Cl} = \Omega(1)$ with respect to $S_{\Cl,\ell}$. Thus, the spectral gap of $\Cl_{\ell}$ with respect to the uniform distribution over $S_{\Cl,\ell}$ is
\begin{equation}
\Delta(\mu(S_{\Cl,\ell}), \rho, \Cl_{\ell}) = \Omega(1),
\end{equation}
for any finite-dimensional unitary representation $\rho$ of $\Cl_{\ell}$ with no trivial subrepresentation.
\end{theorem}

The proof of Theorem~\ref{thm:cliffordgenerator} proceeds by first constructing a generating set for the symplectic group, which is then extended to the full Clifford group. In \cite{Kassabov2006expanders,nikolov2005productdecompositionclassicalquasisimple} the following argument for a generating set for the binary symplectic group is given, which we essentially translate into quantum circuit language:  (1) the special linear group admits a bounded-size 1D-local generating set with an $\Omega(1)$ Kazhdan constant, and (2) the symplectic group can be generated by products of the special linear group and its group isomorphisms. This means one can construct a generating set for the symplectic group by combining special linear group generators with their conjugations. Extending this result via the projective Pauli group yields a generating set for the projective Clifford group.

While the global phase is irrelevant for standard unitary designs as the projective and full Clifford groups yield identical twirling channels, the global phase becomes non-trivial when constructing controlled Clifford operations. To complete the proof for the full Clifford group, we incorporate a group extension with a discrete phase generator into the projective Clifford group. Crucially, each step in this construction preserves both the constant size of the generating set and the $\cO(1)$ 1D circuit depth of individual generators. Because group extensions and products preserve the Kazhdan constant up to constant factors for bounded generating sets, combining these steps establishes Theorem~\ref{thm:cliffordgenerator}.

\subsection{Related work}
\paragraph{1D random circuits with constant spectral gap.}
While we were preparing this manuscript, a concurrent preprint~\cite{baer2026randomunitarycircuitsconstant} demonstrated that 1D brickwork random circuits generate $\varepsilon$-relative-error approximate unitary $k$-designs in depth $\cO(nk + \log(1/\varepsilon))$ via spectral gap analysis. This result offers an alternative approach for synthesizing $\ell$-local random unitary gates, thereby providing an independent path to establishing Theorem~\ref{thm:CPZPCdesign}.


\paragraph{Approximate unitary $k$-designs in an all-to-all system with decoupled $n$ and $k$ in the depth scaling.} In this paper, we establish near-optimality in the 1D geometry. One can ask the same question in other geometries such as the all-to-all geometry. While it is not explicitly discussed in Ref.~\cite{cui2025unitarydesignsnearlyoptimal} and Ref.~\cite{Zhang2026magic}, combining their results, one can straightforwardly realize an approximate unitary design at a near-optimal circuit depth in all-to-all connectivity -- in particular achieving decoupled scaling in $n$ and $k$. In Ref.~\cite{Zhang2026magic}, the authors consider a magic augmented circuit with Clifford block size $\xi = \cO(\log(n/\varepsilon)+k^2)$ and local unitary operation size $\ell = k\log k$. In an all-to-all system, the non-local routing of Clifford operations~\cite{cleve2016nearlinearconstructionsexactunitary} enables a reduction in circuit depth from $\cO(\xi)$ to $\cO(\log \xi)$ with $\cO(\xi\log\xi\log\!\log\xi)$ ancillas. Meanwhile, the $2^{-\Omega(\ell k)}$-error $\cO(\ell)$-local random unitary block can be realized in either depth $\cO(\log(\ell k)\log(k))$ with $\ell k\cdot \widetilde{\cO}(\log \ell k)$ ancillas or depth $\cO(k\log(\ell k))$ with $\ell \cdot \widetilde{\cO}(\log \ell k)$ ancillas~\cite{cui2025unitarydesignsnearlyoptimal}. This gives the construction of an approximate $\varepsilon$-relative-error unitary $k$-design within an all-to-all system using
\begin{enumerate}
\item
depth $\cO(\log\!\log(n/\varepsilon)+\log^2k)$ and ancillas $n\cdot\widetilde{\cO}(\log\!\log(n/\varepsilon))+nk\cdot\widetilde{\cO}(\log k)$;

\item 
depth $\cO(\log\!\log(n/\varepsilon)+k\log k)$ and ancillas $n\cdot\widetilde{\cO}(\log\!\log(n/\varepsilon)+\log k)$.
\end{enumerate}
This result provides the state-of-the-art depth scaling for constructing unitary designs in the all-to-all system while keeping the qubit number $n$ and design order $k$ decoupled.

\paragraph{Depth-1 quantum (tensor product) expanders and applications} While finalizing this manuscript we became aware of the upcoming work~ \cite{anshu2026depth}, which constructs a 1D-local (on a circle) constant-size generating set for the special linear group $\SL(3s;\F_2)$, in a manner essentially identical to our construction in Lemma~ \ref{lem:symplectic_generators}, and applies this to a number of problems in quantum information theory. 
Using the Aaronson-Gottesman normal form of the Clifford group~\cite{aaronson2004improved}, they subsequently obtain a constant-size generating set of the Clifford group with a constant number of generators. Their construction has a smaller number of generators, which is $28$, but only treats the case of $n=3s$ and is local on a circle (both of these problems are easily overcome using Lemma~\ref{lem:sl-overlap} and the folded-qubit layout construction in Fig.~\ref{fig:kassabov-folded-cells}, respectively). Given that Aaronson-Gottesman is essentially a product decomposition in the sense of Nikolov~\cite{nikolov2005productdecompositionclassicalquasisimple}, we expect our constructions to be structurally similar. 

\subsection{Discussion and outlook}
In this work, we constructed nearly quadruply optimal unitary $k$-designs on 1D geometries, matching fundamental lower bounds across circuit depth, qubit count, and error up to a single $\log k$ factor in the design order. To remove this residual logarithmic overhead in the design order, one possible approach is to further reduce the locality $\ell$ from $\cO(\log k)$ to $\cO(1)$ required for magic blocks to suppress all non-permutation symmetries. This requires analyzing the Weingarten function in a singular regime since $2^{\ell} < k$. 

Relative to the lower bound on the magic resources, our construction leaves open a gap of scaling $\cO(\log n\,\log^2 k+\log^3k)$. We believe that one $\log(k)$ factor is due to residual inefficiency in our construction (similar to the excess depth overhead), and that the actual lower bound should be $\Omega(nk)$, with the lower bound derived in Ref.~\cite{Leone2026NonClifford} being lossy by a factor of $\cO(\log n\,\log k+\log^2k)$. Closing this gap in both directions thus remains to be done.

As a key technical subroutine, we constructed a system-size-independent 1D-local generating set for the Clifford group that possesses an $\Omega(1)$ Kazhdan constant. Beyond its role in design generation, this 1D-local Clifford expander may find broader utility in tasks relying on random Clifford gates, such as randomized benchmarking and quantum device verification. From a practical perspective, further compressing the size of this generating set remains an appealing open problem. Furthermore, applying the spectral gap properties of the CPZPC ensemble to analyze other random circuit architectures presents a promising avenue for future research.

Finally, our construction operates in the regime where $k = \cO(n)$, a constraint directly tied to the reversibility of the Clifford Weingarten function, which becomes singular when $k \geq n$. In this high-order case, combining the results of the gluing lemma~\cite{Schuster2025Gluing} and the recent achievement of the constant spectral gap of the 1D brickwork circuits~\cite{baer2026randomunitarycircuitsconstant} gives us a 1D depth scaling $\cO(\log(n/\varepsilon)+\xi k) = \cO(k\log(1/\varepsilon)+ k\log k)$. Decoupling the design order $k$ from the approximation error $\varepsilon$ remains an open question in the high-design-order regime.

\subsection{Structure of this work}
The remainder of this paper is organized as follows. Section~\ref{sc:pre} introduces basic notation and essential preliminaries. Section~\ref{sc:magic-augmented} presents our core results and proofs regarding magic-augmented circuits. In Section~\ref{sc:spectral-gap}, we analyze the generating sets formed by Clifford and permutation group elements, establishing 1D-local generating sets with a constant spectral gap. Finally, Section~\ref{sc:clif} provides the explicit construction of the constant-depth Clifford generators along with their formal proofs.

\subsection*{Acknowledgements}
We thank Jonas Haferkamp for the insightful discussions on the spectral gap of the CPZPC ensemble and the realization of Kassabov's generators and for sharing with us an early draft of \cite{anshu2026depth}. We thank Markus Heinrich for discussions on the number of magic gates in realizing the approximate designs. We also thank Lorenzo Grevink and Soumik Ghosh for early discussions on the magic-augmented circuit construction. JH acknowledges funding from the Dutch Research Council (NWO) through a Veni grant (grant No.VI.Veni.222.331) and the Quantum Software Consortium (NWO Zwaartekracht Grant No.024.003.037). This work was done while GL was visiting Centrum Wiskunde \& Informatica (CWI) and QuSoft, supported by the Tsinghua Scholarship for Overseas Graduate Studies. GL acknowledges funding support from the Turing AI Institute of Nanjing (Grant No.~TR-IIIS 002), the National Natural Science Foundation of China (Grant No.~12575023), the Quantum Science and Technology-National Science and Technology Major Project (Grants No.~2021ZD0300804 and No.~2021ZD0300702), and the CCF-QuantumCtek Superconducting Quantum Computing Special Cooperation Program (Grant No.~CCF-QC2025005).\\

\noindent \textbf{AI use statement}: GPT5.5 Pro derived an initial version of the estimates in lemmas~\ref{lem:permutation-block-wg-discrepancy} and~\ref{lem:sl-overlap}. These were subsequently improved and humanized by the authors.  GPT5.5 Pro and Claude Opus 4.8 were also invaluable in understanding Nikolov's proof and translating its components to quantum circuits, and some edited output made it into the final version of the manuscript.

\section{Preliminaries}\label{sc:pre}
In this section we set notation and introduce necessary preliminaries for our results.
\subsection{Unitary design, spectral gap, and Kazhdan constant}
We use $n$ to denote the total number of qubits in the system, $\xi$ to denote the size of the smallest overlapping regions of the (Clifford) unitaries and $\ell$ to denote the size of the smallest disjoint subsystem that the magic unitary operations act on. The Hilbert space of a system is denoted as $\mathcal{H}$. We use $\mathcal{E}_u$ to represent an ensemble of unitary operations. $\Phi_{\mathcal{E}_u}^{(k)}$ represents the $k$-fold twirling channel for the unitary ensemble, defined as 
\begin{equation}
\Phi_{\mathcal{E}_u}^{(k)}(\cdot) := \mathbb{E}_{U\sim\mathcal{E}_u} U^{\otimes k} (\cdot) U^{\dagger\otimes k}.
\end{equation}
$\Phi_H^{(k)}$ denotes the corresponding $k$-fold channel for the exact Haar random ensemble.

An ensemble $\mathcal{E}_u$ forms an approximate unitary $k$-design if its  $k$-fold twirling channel is close to that of the Haar random ensemble. Depending on the error type, we can define the approximate unitary $k$-design in either an additive-error version or a relative-error version.

\begin{definition}[Additive-error unitary design]
A unitary ensemble $\cE_u$ is an additive-error approximate unitary $k$-design with error $\varepsilon$ if
\begin{equation}
    \left\|\Phi_{\cE_u}^{(k)}-\Phi_H^{(k)}\right\|_\diamond
    \leq
    \varepsilon.
\end{equation}
\end{definition}

\begin{definition}[Relative-error unitary design]
A unitary ensemble $\cE_u$ is a relative-error approximate unitary $k$-design with error $\varepsilon$ if
\begin{equation}
    (1-\varepsilon)\Phi_H^{(k)}
    \preceq
    \Phi_{\cE_u}^{(k)}
    \preceq
    (1+\varepsilon)\Phi_H^{(k)}.
\end{equation}
Here $\Phi\preceq\Psi$ means that $\Psi-\Phi$ is a completely positive (CP) map.
\end{definition}

The relative-error notion is stronger than additive error. Any approximate unitary $k$-design with relative error $\varepsilon$ has additive error $2\varepsilon$, but an approximate unitary $k$-design with additive error $\varepsilon$ might have relative error $2^{2nk}\varepsilon$. The circuit depth lower bound for generating unitary designs was proved for the additive error version~\cite{cui2025unitarydesignsnearlyoptimal}, and hence also applies to the relative error case. In this work, we construct unitary designs in the stronger relative error case.

A sufficient condition for $\mathcal{E}_u$ to form an $n$-qubit $\varepsilon$-relative-error approximate unitary $k$-design is to bound the operator norm of the Choi-state difference~\cite{Schuster2025Gluing},
\begin{equation}\label{eq:epr_norm}
\frac{2^{2nk}}{k!}(1+\frac{k^2}{2^{n+1}}) \Vert[ (\Phi^{(k)}_{\mathcal{E}_u}-\Phi^{(k)}_{H})\otimes \id] (P^{EPR})\Vert_{\infty}\leq \varepsilon,
\end{equation}
where
\begin{equation}
P^{EPR} = (\frac{1}{2}\sum_{i,j\in \{0, 1\}}\ketbra{ii}{jj})^{\otimes nk}.
\end{equation}

One can also view the unitary ensemble $\mathcal{E}_u$ as a probability distribution $\nu$ on the $n$-qubit unitary operation group $\U(2^n)$. The corresponding twirling channel is also denoted as $\Phi_\nu^{(k)}$. Correspondingly, we use $\mu_H$ to denote the Haar measure over $\U(2^n)$. Given a finite operation set $S$, we use $\mu(S)$ to denote the uniform distribution over set $S$.

The twirling channel $\Phi_\nu^{(k)}$ is equivalent, after vectorization, to the moment operator
\begin{equation}\label{eq:prelim-moment-operator}
M_k(\nu):=\E_{U\sim\nu}\left[U^{\otimes k}\otimes \overline U^{\otimes k}\right].
\end{equation}
If $\mu_H$ is Haar measure, we write
\begin{equation}
P_H:=M_k(\mu_H).
\end{equation}
The operator $P_H$ is an orthogonal projection onto the Haar-invariant subspace of the moment representation, which can be characterized by Schur-Weyl duality introduced later.

Below, we give an important definition of spectral gap, characterizing the distance between the distribution $\nu$ and Haar measure $\mu_H$.

\begin{definition}[Essential norm and moment spectral gap]
The $k$-moment essential norm of $\nu$ is
\begin{equation}\label{eq:prelim-essential-norm}
    g(\nu,k)
    :=
    \left\|M_k(\nu)-P_H\right\|_\infty.
\end{equation}
The corresponding $k$-moment spectral gap is
\begin{equation}
    \Delta(\nu,k)
    :=
    1-g(\nu,k).
\end{equation}
\end{definition}

If $\nu$ is supported on a finite group $G$ and $\rho$ is a unitary representation of $G$, the analogous representation-theoretic essential norm is
\begin{equation}
    g(\nu,\rho,G)
    :=
    \left\|\E_{g\sim\nu}\rho(g)-P_{\rho^G}\right\|_\infty,
\end{equation}
where $P_{\rho^G} = \E_{g\sim\mu(G)}\rho(g)$ is the orthogonal projector onto the trivial subspace within $\rho(G)$. The corresponding representation-theoretic spectral gap can be defined accordingly.

The study of spectral gap is important to proving the necessary circuit depth to form relative-error approximate unitary designs, as introduced below.

\begin{lemma}[Gap amplification and designs, Lemma 2.6 in Ref.~\cite{chen2024incompressibilityspectralgapsrandom}]\label{lemma:gap-to-design}
Suppose $\nu$ is a distribution on $n$-qubit unitaries and
\begin{equation}
g(\nu,k)\leq 1-\Delta.
\end{equation}
Then $\nu^{*L} = \nu*\nu*\cdots *\nu$ is an $\varepsilon$-relative-error approximate unitary $k$-design once
\begin{equation}\label{eq:prelim-gap-to-design}
L=\cO\left(\frac{1}{\Delta}\left(nk+\log\frac{1}{\varepsilon}\right)\right).
\end{equation}
In particular, if $\Delta=\Omega(1)$, then
\begin{equation}
L=\cO\left(nk+\log\frac{1}{\varepsilon}\right).
\end{equation}
\end{lemma}
Lemma~\ref{lemma:gap-to-design} is the standard spectral-gap-to-relative-design estimate. Bounding the spectral gap offers a way to prove the depth upper bound to form approximate unitary designs.

To evaluate the spectral gap of a generating set for certain groups, we resort to the tool of the Kazhdan constant. The Kazhdan constant of a large group can be bounded by the Kazhdan constants of subgroups decomposed from the large group. We provide the definition of Kazhdan constants and the related results below.

\begin{definition}[Kazhdan constants]\label{def:kazhdanconstant}
Consider a group generated by a set $S$. The Kazhdan constant for $G$ with respect to $S$, denoted by $\cK(G;S)$ is the largest $\varepsilon\geq 0$ such that for every unitary representation $\rho: G\rightarrow \U(\mathcal{H})$ that contains no trivial subrepresentation, there exists $g\in S$ satisfying that for every unit vector $\ket{\psi}\in \mathcal{H}$,
\begin{equation}
\Vert\rho(g)\ket{\psi}-\ket{\psi}\Vert\geq \varepsilon.
\end{equation}
\end{definition}

There is another definition called the average Kazhdan constant, which is closely related to the above definition.

\begin{definition}[Average Kazhdan constant]
Let $G$ be a finite group and let $S$ be a finite generating set. The average Kazhdan constant is
\begin{equation}
\cK_{\av}(G;S)=
\inf_{\rho\in \widehat G_0}\inf_{\Vert v\Vert=1}
\frac{1}{|S|}\sum_{s\in S}\Vert\rho(s)v-v\Vert^2,
\end{equation}
where $\widehat G_0$ denotes the set of all non-trivial irreducible representations.
\end{definition}

The Kazhdan constant and the average Kazhdan constant are mutually bounded with each other:
\begin{equation}
\cK(G;S)^2\geq \cK_{\av}(G;S)\geq \frac{\cK(G;S)^2}{|S|}.
\end{equation}

Below, we list two useful lemmas related to Kazhdan constants proved in Ref.~\cite{chen2024incompressibilityspectralgapsrandom}. One establishes the Kazhdan constant of a product group, and the other relates the Kazhdan constant to the essential norm and the spectral gap.

\begin{lemma}[Bounded-product Kazhdan comparison, Corollary 2.16 in Ref.~\cite{chen2024incompressibilityspectralgapsrandom}]
\label{lem:bounded-product}
Let $H$ be generated by $S$, and let $\phi_i:H\to G$ be group homomorphisms. Suppose every element of $G$ is a product of at most $M$ elements from $\bigcup_i\phi_i(H)$. Then
\begin{equation}
    \cK\left(G;\bigcup_i\phi_i(S)\right)
    \geq \frac{1}{M\sqrt{2}}\cK(H;S).
\end{equation}
\end{lemma}

Lemma~\ref{lem:bounded-product} manifests that the Kazhdan constants of a product group can be lower bounded by the Kazhdan constant of each product element.

\begin{lemma}[Spectral gap from Kazhdan constant, Lemma 2.17 in Ref.~\cite{chen2024incompressibilityspectralgapsrandom}]\label{lem:kazhdan-constant-to-spectral-gap}
Suppose $G$ is a compact group generated by a finite subset $S$ where $S$ is symmetric or closed under taking inverses. Let $\mu(S)$ be the uniform distribution over $S$. Then for any finite-dimensional unitary representation $\rho$ of $G$, the essential norm
\begin{equation}
g(\mu(S), \rho, G)\leq 1 - \frac{\cK(G;S)^2}{2\abs{S}}.
\end{equation}
\end{lemma}

Lemma~\ref{lem:kazhdan-constant-to-spectral-gap} shows that as long as the Kazhdan constant is constant and the generating set has a bounded size, the spectral gap is also a constant. It is important to note that the generating set must be symmetric to apply Lemma~\ref{lem:kazhdan-constant-to-spectral-gap}.

In this work, we also evaluate how Kazhdan constants behave under group extensions. As preliminary context, we employ the standard homological algebra notation of a short exact sequence to describe group decompositions and extensions. A sequence of groups and group homomorphisms
\begin{equation}
1 \longrightarrow N \longrightarrow G \longrightarrow H \longrightarrow 1
\end{equation}
represents that $N$ is a normal subgroup of $G$ ($N \triangleleft G$), and $G/N\cong H$. In the context of group theory, such a sequence denotes that $G$ is an extension of the group $H$ by the group $N$. For example, when $N$ is an abelian group, it is referred to as an abelian extension. This algebraic framework enables us to systematically construct a generating set for $G$ by bootstrapping Kazhdan constants from its quotient $H$ and normal subgroup $N$. Specifically, we invoke a key result established by Hadad~\cite{hadad2010kazhdanconstantsgroupextensions}, which bound Kazhdan constants under group extensions using the notion of average Kazhdan constants.

\begin{lemma}[Hadad abelian-extension theorem, Theorem 1.6 in Ref.~\cite{hadad2010kazhdanconstantsgroupextensions}]\label{lem:hadad-abelian}
Let
\begin{equation}
    1\longrightarrow A\longrightarrow \Gamma\longrightarrow H\longrightarrow 1
\end{equation}
be an extension with $A$ finite abelian. Let $S\subseteq H$, let $\widetilde S$ be any lift of $S$ to $\Gamma$, and let $B\subseteq A$. Let
\begin{equation}
    \widetilde B=\{a^h:a\in B,\ h\in H\}
\end{equation}
be the union of the $H$-orbits of $B$. If $H$ has average Kazhdan constant $\cK_{\av,H}$ with respect to $S$, and $A$ has average Kazhdan constant $\cK_{\av,A}$ with respect to $\widetilde B$, then $\Gamma$ has average Kazhdan constant at least
\begin{equation}
    \frac{\cK_{\av,H}\cK_{\av,A}}{512\left(1+|S|/|B|+|B|/|S|\right)}
\end{equation}
with respect to $\widetilde S\cup B$.
\end{lemma}

The important point of Lemma~\ref{lem:hadad-abelian} is that the generating set of $\Gamma$ only contains the small set $B$, while the expansion of the abelian kernel is checked with respect to the full orbit $\widetilde B$.

Together, the lemmas above establish the basis of our proof for the Clifford group having a constant Kazhdan constant with respect to a constant-size, 1D-local generating set.

\subsection{The Clifford commutant theory}
In this work, we use magic-augmented circuits to construct unitary designs, which involves $k$-fold twirling channel for the Clifford group and the Clifford commutant theory. We introduce the basic and necessary preliminaries of the Clifford commutant theory below. For more information and proofs, one can refer to Refs.~\cite{Gross2021Clifford,bittel2025completetheorycliffordcommutant}. Before we discuss the Clifford case, we first introduce the Haar Weingarten calculus as a first step.

Based on the Haar Weingarten calculus, the $k$-fold twirling channel for the Haar random ensemble $\Phi_H^{(k)}$ can be explicitly represented with the symmetric group $S_k$.  For $\pi\in S_k$, define the permutation operator $T_\pi$ on $\mathcal{H}^{\otimes k}$ by
\begin{equation}
T_\pi\left(\ket{v_1}\otimes\cdots\otimes\ket{v_k}\right):=\ket{v_{\pi^{-1}(1)}}\otimes\cdots\otimes\ket{v_{\pi^{-1}(k)}}.
\end{equation}
Schur-Weyl duality ensures that the commutant of the $U^{\otimes k}$ action is spanned by a linearly independent set $\{T_\pi:\pi\in S_k\}$ when the dimension $d\geq k$.
\begin{equation}
\mathrm{Com}(\U(d)^{\otimes k}) = \{X, XU^{\otimes k} = U^{\otimes k}X, \forall U\in \U(d)\} = \Span\{T_\pi: \pi\in S_k\}.
\end{equation}
This allows us to expand $\Phi_H^{(k)}$ in the basis of $\{T_\pi: \pi\in S_k\}$, where the coefficient is determined by the Weingarten matrix.

The Haar Gram matrix is
\begin{equation}\label{eq:prelim-haar-gram}
G_H(T_\pi,T_\sigma):=\Tr\left(T_\pi^\dagger T_\sigma\right)=d^{c(\pi^{-1}\sigma)},
\end{equation}
where $d$ is the Hilbert-space dimension under discussion and $c(\tau)$ is the number of cycles of $\tau\in S_k$. The Haar Weingarten matrix is the inverse Gram matrix:
\begin{equation}
\Wg(T_\pi,T_\sigma):=\left(G_H^{-1}\right)(T_\pi,T_\sigma).
\end{equation}
For simplicity, one also writes $G_H(\pi,\sigma) = G_H(T_\pi,T_\sigma)$ and $\Wg(\pi,\sigma) = \Wg(T_\pi,T_\sigma)$. It is easy to verify that $\Wg(\pi,\sigma)$ depends only on $\pi^{-1}\sigma$, so one often writes it as $\Wg(\pi^{-1}\sigma)$.

By Schur-Weyl duality, the Haar twirl $\Phi_H^{(k)}$ has the expansion
\begin{equation}\label{eq:prelim-haar-twirl}
    \Phi_H^{(k)}(X)
    =
    \sum_{\pi,\sigma\in S_k}
    T_\pi\Wg(\pi,\sigma)\Tr\left(T_\sigma^\dagger X\right).
\end{equation}
In this work, we denote $\cket{T}$ to represent $T$ and $\cbra{T}$ to represent $T^{\dagger}$. In this notation, the inner product $\Tr\left(T_\sigma^\dagger X\right)$ can be written as $\cbraket{T_{\sigma}}{X}$. Thus, we have
\begin{equation}\label{eq:prelim-haar-twirl_braket_notation}
\Phi_H^{(k)}(X)=\sum_{\pi,\sigma\in S_k} \Wg(\pi,\sigma) \cketbra{T_{\pi}}{T_{\sigma}}\cket{X}.
\end{equation}
The Haar twirl $\Phi_H^{(k)}$ is represented as
\begin{equation}\label{eq:haar-twirl_braket_notation}
\Phi_H^{(k)} =\sum_{\pi,\sigma\in S_k} \Wg(\pi,\sigma) \cketbra{T_{\pi}}{T_{\sigma}}.
\end{equation}

The Haar commutant is generated only by permutations. When we move to Clifford group and consider the Clifford commutant, we have to consider a larger set than the permutations. The set is labeled by stochastic Lagrangian subspaces.

\begin{definition}[Stochastic Lagrangian subspace]
A linear subspace $T\subseteq \F_2^{2k}$ is a stochastic Lagrangian subspace iff
\begin{enumerate}
\item $\forall (\mathbf{x}, \mathbf{y})\in T, x\cdot x-y\cdot y = 0\bmod 4$.
\item $\dim T = k$.
\item $1^{2k}\in T$.
\end{enumerate}
\end{definition}

The $n$-qubit Clifford commutant is generated by a set $\Sigma_{k,k}$ of stochastic Lagrangian subspaces of $\F_2^{2k}$ (which we will refer to as Clifford commutant elements, in a slight abuse of language). Suppose $n\geq k-1$, then
\begin{equation}
\mathrm{Com}(\Cl_n^{\otimes k}) = \{X, XC^{\otimes k} = C^{\otimes k}X, \forall C\in \Cl_n\} = \Span\{R(T): T\in \Sigma_{k,k}\},
\end{equation}
where the bases below are linearly independent:
\begin{equation}
R(T) = r(T)^{\otimes n},\quad r(T) = \sum_{(\mathbf{x}, \mathbf{y})\in T}\ketbra{\mathbf{x}}{\mathbf{y}}.
\end{equation}
One can see that permutations $S_k$ are subset of $\Sigma_{k,k}$. For any permutation $\pi$, one can define
\begin{equation}
T_{\pi} = \{(x, \pi x), x\in \F_2^k\},
\end{equation}
which is a valid stochastic Lagrangian subspace. In the following, we will use $T$ to also represent $R(T)$ and $r(T)$ for simplicity. Meanwhile, when writing a Clifford commutant element $T\in S_k$, we mean there exists a permutation $\pi\in S_k$, such that $T = T_{\pi}$. In this sense, $S_k\subseteq \Sigma_{k,k}$. We use $\Sigma_{k,k}\setminus S_k$ to represent all non-permutation Clifford commutant elements, and $\Sigma_{k,k}/S_k$ to represent all $\Sigma_{k,k}$ equivalence classes under the right action of permutations.

Similar to the Haar case, the Clifford Gram matrix is defined as
\begin{equation}\label{eq:prelim-clifford-gram}
G_C(T,T'):= \Tr\left(R(T)^\dagger R(T')\right)=d^{\dim(T\cap T')}.
\end{equation}
One can notice that the permutation block of $G_C$ agrees with $G_H$. The Clifford Weingarten matrix is the inverse Clifford Gram matrix:
\begin{equation}
\Wg_C(T,T'):=\left(G_C^{-1}\right)(T,T').
\end{equation}

Correspondingly, by Clifford Schur-Weyl duality, the Clifford twirl $\Phi_C^{(k)}$ has the expansion
\begin{equation}\label{eq:prelim-clif-twirl}
\Phi_C^{(k)} =\sum_{T,T'\in \Sigma_{k,k}} \Wg_C(T,T') \cketbra{T}{T'},
\end{equation}
where $\cket{T}$ represents $R(T)$, and $\cbraket{T}{X} = \Tr\left( R(T)^{\dagger} X\right)$. In particular, for a map of the form $\cketbra{T_1}{T_2}$, we have
\begin{align}
&[\cketbra{T_1}{T_2}\otimes \id] (P^{EPR}) = 2^{-nk} \cket{T_1}\otimes \cket{T_2},\\
&\Vert [\cketbra{T_1}{T_2}\otimes \id] (P^{EPR})\Vert_{\infty}
 = 2^{-nk}\Vert T_1\otimes T_2\Vert_{\infty} = 2^{-nk} \Vert T_1\Vert_{\infty}\Vert T_2\Vert_{\infty}.
\end{align}
For any Clifford commutant element $T$, we have the bound
\begin{equation}
1\leq \Vert T\Vert_{\infty}\leq 2^{nk},
\end{equation}
where for any permutation $\pi$,
\begin{equation}
\Vert T_{\pi}\Vert_{\infty} = 1.
\end{equation}
For the permutation group generating set $S_k$ and Clifford commutant generating set $\Sigma_{k,k}$, we can classify elements depending on their distance to the identity element. For $T,T'\in\Sigma_{k,k}$, we define the distance
\begin{equation}
\delta(T,T') := k-\dim(T\cap T').
\end{equation}
For permutation subspaces $T_\pi,T_\sigma\in S_k$, this distance agrees with the transposition length:
\begin{equation}
\delta(T_\pi,T_\sigma) = k-c(\pi^{-1}\sigma).
\end{equation}
We denote $\delta(T) = \delta(T, T_e)$ for simplicity, where $e$ is the identity. The number of permutations with $\delta(T_\pi)=m$ is the unsigned Stirling number
\begin{equation}
\#\{\pi\in S_k:\delta(T_\pi)=m\}
=
\begin{bmatrix}
k\\ k-m
\end{bmatrix},
\end{equation}
which satisfies
\begin{equation}
\sum_{m=0}^{k-1} \begin{bmatrix}
k\\ k-m
\end{bmatrix} x^m = \prod_{j=0}^{k-1}(1+jx).
\end{equation}
One can verify that
\begin{equation}
|S_k| = \sum_{m=0}^{k-1} \begin{bmatrix}
k\\ k-m
\end{bmatrix} = k!.
\end{equation}

Similarly, for the Clifford commutant, the number of commutant elements with $\delta(T)=m$ is given by the Gaussian binomial coefficient over $\mathbb{F}_2$,
\begin{equation}
\#\{T\in \Sigma_{k,k}:\delta(T)=m\} = 2^{m(m-1)/2}\binom{k-1}{m}_{2},
\end{equation}
which satisfies
\begin{equation}\label{eq:clifweingartensum}
\sum_{m=0}^{k-1} 2^{m(m-1)/2}
\binom{k-1}{m}_{2} x^m = \prod_{j=0}^{k-2}(1+2^jx).
\end{equation}
Correspondingly, the size of $\Sigma_{k,k}$ is given by
\begin{equation}
|\Sigma_{k,k}| = \prod_{j=0}^{k-2}(1+2^j) = 2^{\cO(k^2)}.
\end{equation}

\section{Unitary designs with magic-augmented circuits}\label{sc:magic-augmented}
In this section, we show that the magic-augmented circuit can form an approximate design at a small block size and prove Theorem~\ref{thm:MagicAugmented}. For convenience, we restate Theorem~\ref{thm:MagicAugmented} below.

\ThmMagicAugmented*

We first consider the case where the two layers of tensor-product random unitary gates are sampled from an exact unitary $k$-design. Later in the proof we will replace these with approximate unitary $k$-designs. The unitary gate on a $\xi$-qubit block has the form
\begin{equation}
U = \big(\bigotimes_{i=1}^{\xi/\ell} V^2_i\big) C  \big(\bigotimes_{i=1}^{\xi/\ell} V^1_i\big),
\end{equation}
where $V^j_i$ is an $\ell$-qubit random unitary gate. The corresponding unitary ensemble is denoted as $\mathcal{E}_u$.
The associated $k$-fold twirling operation of $\mathcal{E}_u$ is given by
\begin{equation}
\Phi^{(k)}_{\mathcal{E}_u} = \sum_{\vec{\sigma}, \vec{\pi}, \vec{\mu}, \vec{\nu}}\sum_{T_1, T_2} \Wg(\vec{\sigma}, \vec{\pi})\Wg(\vec{\mu}, \vec{\nu}) \Wg_C(T_1, T_2) \cketbra{\vec{\sigma}}{\vec{\pi}}\cketbra{T_1}{T_2}\cketbra{\vec{\mu}}{\vec{\nu}}.
\end{equation}
Note that we denote $\vec\sigma=(\sigma_1,\ldots,\sigma_{\xi/\ell})\in S_k^{\xi/\ell}$, $\ket{\vec\sigma}:=\bigotimes_j\ket{T_{\sigma_j}}$, and $\Wg(\vec\sigma,\vec\pi):=\prod_j\Wg(\sigma_j,\pi_j)$.

We want to compare this operator to its Haar equivalent.
Following \cite{Zhang2026magic}, we can split up this summation into three different groups:
\begin{align}
&\Phi^{(k), 1}_{\mathcal{E}_u} = \sum_{\vec{\sigma}, \vec{\pi}, \vec{\mu}, \vec{\nu}}\sum_{T_1, T_2\in S_k} \Wg(\vec{\sigma}, \vec{\pi})\Wg(\vec{\mu}, \vec{\nu}) \Wg_C(T_1, T_2) \cketbra{\vec{\sigma}}{\vec{\pi}}\cketbra{T_1}{T_2}\cketbra{\vec{\mu}}{\vec{\nu}};\\
&\Phi^{(k), 2}_{\mathcal{E}_u} = \sum_{\vec{\sigma}, \vec{\pi}, \vec{\mu}, \vec{\nu}}\sum_{T\in \Sigma_{k,k}\setminus S_k} \Wg(\vec{\sigma}, \vec{\pi})\Wg(\vec{\mu}, \vec{\nu}) \Wg_C(T, T) \cketbra{\vec{\sigma}}{\vec{\pi}}\cketbra{T}{T}\cketbra{\vec{\mu}}{\vec{\nu}};\\
&\Phi^{(k), 3}_{\mathcal{E}_u} = \sum_{\vec{\sigma}, \vec{\pi}, \vec{\mu}, \vec{\nu}}\sum_{\substack{T_1\neq T_2\\(T_1,T_2)\notin S_k\times S_k}} \Wg(\vec{\sigma}, \vec{\pi})\Wg(\vec{\mu}, \vec{\nu}) \Wg_C(T_1, T_2) \cketbra{\vec{\sigma}}{\vec{\pi}}\cketbra{T_1}{T_2}\cketbra{\vec{\mu}}{\vec{\nu}}.
\end{align}
For the first part, with $\sum_{\pi}\Wg(\sigma, \pi)\cbraket{\pi}{T} = \delta_{\sigma, T}$ for $T\in S_k$, we have that
\begin{equation}
\Phi^{(k), 1}_{\mathcal{E}_u} = \sum_{T_1, T_2\in S_k} \Wg_C(T_1, T_2) \cketbra{T_1}{T_2}.
\end{equation}
Thus,
\begin{equation}
\begin{split}
\Vert [(\Phi^{(k),1}_{\mathcal{E}_u}-\Phi^{(k)}_{H})\otimes \id] (P^{EPR})\Vert_{\infty} &= \Vert\sum_{T_1, T_2\in S_k} (\Wg_C(T_1, T_2)-\Wg(T_1,T_2)) [\cketbra{T_1}{T_2}\otimes \id](P^{EPR})\Vert_{\infty}\\
&\leq \sum_{T_1, T_2\in S_k} 2^{-\xi k}\abs{\Wg_C(T_1, T_2)-\Wg(T_1,T_2)}.
\end{split}
\end{equation}
The above error is the discrepancy between Clifford and Haar Weingarten functions on all permutation elements. The scaling of the error is shown in the following lemma.

\begin{lemma}[Permutation-block Clifford--Haar Weingarten discrepancy]
\label{lem:permutation-block-wg-discrepancy}
For any $\xi$-qubit system and design order $k$ with $k\leq \xi-1$,
\begin{equation}
\sum_{T_1,T_2\in S_k}2^{\xi k}\abs{\Wg_C(T_1,T_2)-\Wg(T_1,T_2)} \leq
k!2^{k-2\xi+3}.
\end{equation}
\end{lemma}
We postpone the proof of Lemma~\ref{lem:permutation-block-wg-discrepancy} to the end of this section. By Lemma~\ref{lem:permutation-block-wg-discrepancy}, the first term in the $k$-fold twirl is close to its Haar equivalent in the "EPR" distance required for relative error. In particular we have
\begin{equation}
\begin{split}
2^{2\xi k}\Vert [(\Phi^{(k),1}_{\mathcal{E}_u}-\Phi^{(k)}_{H})\otimes \id] (P^{EPR})\Vert_{\infty}\leq k!2^{k-2\xi+3}.
\end{split}
\end{equation}
Furthermore we can show that the second and third terms are small in the same norm. These terms can be bounded as follows:
\begin{equation}
\begin{split}
&\Vert [\Phi^{(k), 2}_{\mathcal{E}_u}\otimes \id](P^{EPR})\Vert_{\infty}\\
\leq&  \sum_{T\in \Sigma_{k,k}\setminus S_k} \abs{\Wg_C(T, T)} \times \Vert [\sum_{\vec{\sigma}, \vec{\pi}, \vec{\mu}, \vec{\nu}} \Wg(\vec{\sigma}, \vec{\pi})\Wg(\vec{\mu}, \vec{\nu})  \cketbra{\vec{\sigma}}{\vec{\pi}}\cketbra{T}{T}\cketbra{\vec{\mu}}{\vec{\nu}}\otimes \id](P^{EPR})\Vert_{\infty}\\
=& \sum_{T\in \Sigma_{k,k}\setminus S_k} 2^{-\xi k}\abs{\Wg_C(T, T)} \times \Vert \sum_{\vec{\sigma}, \vec{\pi}, \vec{\mu}, \vec{\nu}} \Wg(\vec{\sigma}, \vec{\pi})\Wg(\vec{\mu}, \vec{\nu})  \cketbra{\vec{\sigma}}{\vec{\pi}}\cketbra{T}{T}\cket{\vec{\mu}}\cket{\vec{\nu}} \Vert_{\infty}\\
\leq& \sum_{T\in \Sigma_{k,k}\setminus S_k} 2^{-\xi k}\abs{\Wg_C(T, T)} \times (f(T))^{2\xi/\ell},
\end{split}
\end{equation}

\begin{equation}
\begin{split}
&\Vert [\Phi^{(k), 3}_{\mathcal{E}_u}\otimes \id](P^{EPR})\Vert_{\infty}\\
\leq&  \sum_{\substack{T_1\neq T_2\\(T_1,T_2)\notin S_k\times S_k}} \abs{\Wg_C(T_1, T_2)} \times \Vert [\sum_{\vec{\sigma}, \vec{\pi}, \vec{\mu}, \vec{\nu}} \Wg(\vec{\sigma}, \vec{\pi})\Wg(\vec{\mu}, \vec{\nu})  \cketbra{\vec{\sigma}}{\vec{\pi}}\cketbra{T_1}{T_2}\cketbra{\vec{\mu}}{\vec{\nu}}\otimes \id](P^{EPR})\Vert_{\infty}\\
=&  \sum_{\substack{T_1\neq T_2\\(T_1,T_2)\notin S_k\times S_k}} 2^{-\xi k} \abs{\Wg_C(T_1, T_2)} \times \Vert \sum_{\vec{\sigma}, \vec{\pi}, \vec{\mu}, \vec{\nu}} \Wg(\vec{\sigma}, \vec{\pi})\Wg(\vec{\mu}, \vec{\nu})  \cketbra{\vec{\sigma}}{\vec{\pi}}\cketbra{T_1}{T_2}\cket{\vec{\mu}}\cket{\vec{\nu}}\Vert_{\infty}\\
\leq& \sum_{\substack{T_1\neq T_2\\(T_1,T_2)\notin S_k\times S_k}} 2^{-\xi k}\abs{\Wg_C(T_1, T_2)} \times (f(T_1)f(T_2))^{\xi/\ell}.
\end{split}
\end{equation}
The function $f$ is defined as follows:
\begin{equation}
f(T) = \Vert \Phi^{(k)}_H(T) \Vert_{\infty} = \Vert \sum_{\sigma, \pi}\Wg(\sigma, \pi) \cketbra{\sigma}{\pi}\cket{T} \Vert_{\infty}.
\end{equation}
Note that $f(T)$ is evaluated on $\ell$ qubits. Loosely speaking it measures how far away an element of the Clifford commutant generating set is from the unitary commutant. In particular if $T\in S_k$, then $\Phi^{(k)}_H(T) = T$ and $f(T) = 1$.\\

There are two factors determining the values of $\Vert [\Phi^{(k), 2}_{\mathcal{E}_u}\otimes \id](P^{EPR})\Vert_{\infty}$ and $\Vert [\Phi^{(k), 3}_{\mathcal{E}_u}\otimes \id](P^{EPR})\Vert_{\infty}$. One is the Clifford Weingarten function, and the other is the scaling of $f(T)$. We will deal with each of these in turn. For the summation of Clifford Weingarten functions, we have the following lemma.

\begin{lemma}\label{lem:CliffordWeingartenSum}
For any $\xi$-qubit system, order $k\leq \xi-1$, and any Clifford commutant element $T\in \Sigma_{k,k}$,
\begin{align}
2^{\xi k}\abs{\Wg_C(T, T)}&\leq \frac{1}{1-\eta};\\
\sum_{T'\neq T} 2^{\xi k}\abs{\Wg_C(T, T')} &\leq \frac{\eta}{1-\eta},
\end{align}
where
\begin{equation}
\eta = \Vert (2^{\xi k}\Wg_C)^{-1} - \id \Vert_{\infty\rightarrow\infty} \leq 2^{k-\xi}.
\end{equation}
\end{lemma}

The proof of Lemma~\ref{lem:CliffordWeingartenSum} is contained in the proof of Lemma~\ref{lem:permutation-block-wg-discrepancy}, since they require the same proof technique. The derivation of these lemmas is tedious but straightforward (though care must be taken not to introduce factors of $k^2$ in the final approximation). In short, the targets in the two lemmas are the absolute sum of the Weingarten function's elements, and hence can be upper-bounded by the infinite norm $\Vert\cdot\Vert_{\infty\rightarrow\infty}$. To evaluate the infinite norm, like $\eta = \Vert (2^{\xi k}\Wg_C)^{-1} - \id \Vert_{\infty\rightarrow\infty}$, we transform it to a counting procedure, allowing us to upper bound $\eta$ with
\begin{equation}
\eta \leq \sum_{m=1}^{k-1} 2^{m(m-1)/2}
\binom{k-1}{m}_{2} 2^{-\xi m}\leq 2^{k-\xi}.
\end{equation}
The right inequality utilizes the property of Gaussian binomial coefficients in Eq.~\eqref{eq:clifweingartensum}.\\

With these lemmas, we can continue upper bounding the contributions to the $k$-fold twirling operator. By Lemma~\ref{lem:CliffordWeingartenSum}, we have that
\begin{equation}
\begin{split}
2^{2\xi k}\Vert [\Phi^{(k), 2}_{\mathcal{E}_u}\otimes \id](P^{EPR})\Vert_{\infty} \leq \frac{1}{1-\eta}\sum_{T\in \Sigma_{k,k}\setminus S_k} (f(T))^{2\xi/\ell}.
\end{split}
\end{equation}
\begin{equation}
\begin{split}
2^{2\xi k}\Vert [\Phi^{(k), 3}_{\mathcal{E}_u}\otimes \id](P^{EPR})\Vert_{\infty} &\leq \sum_{\substack{T_1\neq T_2\\(T_1,T_2)\notin S_k\times S_k}} 2^{\xi k}\abs{\Wg_C(T_1, T_2)} \times \frac{f(T_1)^{2\xi/\ell}+f(T_2)^{2\xi/\ell}}{2}\\
&\leq \frac{\eta}{1-\eta}\sum_{T\in \Sigma_{k,k}} f(T)^{2\xi/\ell}\\
&= \frac{\eta}{1-\eta}(k!+\sum_{T\in \Sigma_{k,k}\setminus S_k} f(T)^{2\xi/\ell}).
\end{split}
\end{equation}
In summary, we have that
\begin{equation}
2^{2\xi k}\left(\Vert [\Phi^{(k), 2}_{\mathcal{E}_u}\otimes \id](P^{EPR})\Vert_{\infty} + \Vert [\Phi^{(k), 3}_{\mathcal{E}_u}\otimes \id](P^{EPR})\Vert_{\infty}\right) \leq \frac{\eta}{1-\eta}k! + \frac{1+\eta}{1-\eta}\sum_{T\in \Sigma_{k,k}\setminus S_k} f(T)^{2\xi/\ell},
\end{equation}
which reduces to
\begin{equation}
2^{2\xi k}\left(\Vert [\Phi^{(k), 2}_{\mathcal{E}_u}\otimes \id](P^{EPR})\Vert_{\infty} + \Vert [\Phi^{(k), 3}_{\mathcal{E}_u}\otimes \id](P^{EPR})\Vert_{\infty}\right) \leq k!2^{k-\xi+1} + 3\sum_{T\in \Sigma_{k,k}\setminus S_k} f(T)^{2\xi/\ell},
\end{equation}
given $k \leq \xi-1$ and $\eta \leq 2^{k-\xi}\leq 2^{-1}$. The following lemma, which is the core novelty of our argument, provides a bound on $\sum_{T\in \Sigma_{k,k}\setminus S_k} f(T)^{2\xi/\ell}$.

\begin{lemma}\label{lem:twirledCliffordCommutant}
For any $\xi$-qubit system and order $5\leq k\leq \xi-1$, suppose $\ell < \xi$, $2^{\ell}\geq k^8$, and define $f(T) = \Vert \Phi^{(k)}_H(T) \Vert_{\infty} = \Vert \sum_{\sigma, \pi}\Wg(\sigma, \pi) \cketbra{\sigma}{\pi}\cket{T} \Vert_{\infty}$ on an $\ell$-qubit system, we have that
\begin{equation}
\begin{split}
\sum_{T\in \Sigma_{k,k}\setminus S_k} f(T)^{2\xi/\ell}\leq k! 2^{k} (\frac{k^4}{2^\ell})^{2\xi/\ell}.
\end{split}
\end{equation}
\end{lemma}

We leave the proof of Lemma~\ref{lem:twirledCliffordCommutant} to the end of this section. Proving this lemma is substantially more complicated than the preceding arguments. The key is to utilize the invariance of $f(T)$ under permutation actions. Applying Lemma~\ref{lem:twirledCliffordCommutant}, we would get for a $\xi$-qubit magic-augmented circuit with $\ell$-qubit exact unitary $k$-design, the approximation relative error $\varepsilon$ is upper bounded by
\begin{align}
\frac{2^{2\xi k}}{k!}(1+\frac{k^2}{2^{\xi+1}})\Vert [(\Phi^{(k)}_{\mathcal{E}_u}&-\Phi^{(k)}_{H})\otimes \id] (P^{EPR})\Vert_{\infty}\notag\\
=&\frac{2^{2\xi k}}{k!}(1+\frac{k^2}{2^{\xi+1}})\bigg(\Vert [(\Phi^{(k),1}_{\mathcal{E}_u}-\Phi^{(k)}_{H})\otimes \id] (P^{EPR})\Vert_{\infty} \notag\\
&\hspace{5em}+ \Vert [\Phi^{(k), 2}_{\mathcal{E}_u}\otimes \id](P^{EPR})\Vert_{\infty} + \Vert [\Phi^{(k), 3}_{\mathcal{E}_u}\otimes \id](P^{EPR})\Vert_{\infty}\bigg)\\
\leq& (1+\frac{k^2}{2^{\xi+1}})\left(2^{k-2\xi+3} + 2^{k-\xi+1} + 3\cdot 2^{k} (\frac{k^4}{2^\ell} )^{2\xi/\ell} \right)\\
\leq& 2^{k-\xi+3} + 2^{k+3} (\frac{k^4}{2^\ell})^{2\xi/\ell}.
\end{align}
The last inequality utilizes the fact that $k\leq \xi-1$ and $\xi$ is sufficiently large. Hence, $1+\frac{k^2}{2^{\xi+1}}\leq 2$ and $2^{k-2\xi+3}\leq 2^{k-\xi+1}$. Applying Eq.~\eqref{eq:epr_norm} gives us the result we want for Haar random $\ell$-qubit blocks\\

Now we replace the Haar random blocks with approximate unitary $k$-designs with relative error $\varepsilon''$. We will denote the associated ensemble as $\mathcal{E}_{v,\ell}$. We have the following result.

\begin{lemma}\label{lem:lblockapprdesign}
For any $\ell$-qubit system and Clifford commutant element $T\in \Sigma_{k,k}$, we have that
\begin{equation}
\Vert \Phi^{(k)}_{\mathcal{E}_{v,\ell}}(T) \Vert_{\infty}\leq (1+\varepsilon''2^{2\ell k})\Vert \Phi^{(k)}_{H}(T) \Vert_{\infty}.
\end{equation}
\end{lemma}
\begin{proof}[Proof of Lemma~\ref{lem:lblockapprdesign}]
Since $\mathcal{E}_{v,\ell}$ is an $\varepsilon''$-relative-error
approximate unitary $k$-design, we have
\begin{equation}
(1-\varepsilon'')\Phi_H^{(k)}
\preceq
\Phi_{\mathcal{E}_{v,\ell}}^{(k)}
\preceq
(1+\varepsilon'')\Phi_H^{(k)}.
\end{equation}
Define the difference operator
\begin{equation}
\Delta
:=
\Phi_{\mathcal{E}_{v,\ell}}^{(k)}
-
\Phi_H^{(k)},
\end{equation}
and the following two maps,
\begin{equation}
\Gamma_{+}
:=
\varepsilon''\Phi_H^{(k)}+\Delta = \Phi_{\mathcal{E}_{v,\ell}}^{(k)} - (1-\varepsilon'')\Phi_H^{(k)},
\quad
\Gamma_{-}
:=
\varepsilon''\Phi_H^{(k)}-\Delta = (1+\varepsilon'')\Phi_H^{(k)} - \Phi_{\mathcal{E}_{v,\ell}}^{(k)}.
\end{equation}
By the definition of relative-error approximate unitary design, both $\Gamma_+$ and $\Gamma_-$ are completely positive. Both $\Phi_{\mathcal{E}_{v,\ell}}^{(k)}$ and $\Phi_H^{(k)}$ are
unital, and hence
\begin{equation}
\Gamma_{\pm}(\id)=\varepsilon''\id.
\end{equation}
For any completely positive map $\Gamma$, by the Russo-Dye Theorem for positive linear maps~\cite{russo1966note}, which is stated as Theorem~1.1 in Ref.~\cite{bourin2019russodyetheorempositivelinear}, its induced infinity norm satisfies
\begin{equation}
\Vert\Gamma\Vert_{\infty\rightarrow\infty} = \Vert\Gamma(\id)\Vert_{\infty}.
\end{equation}
Consequently,
\begin{equation}
\Vert\Gamma_{\pm}\Vert_{\infty\rightarrow\infty}
=
\varepsilon''.
\end{equation}
Since
\begin{equation}
\Delta=\frac{1}{2}\left(\Gamma_{+}-\Gamma_{-}\right),
\end{equation}
the triangle inequality gives
\begin{equation}
\Vert\Delta\Vert_{\infty\rightarrow\infty}
\leq
\frac{1}{2}
\left(
\Vert\Gamma_{+}\Vert_{\infty\rightarrow\infty}
+
\Vert\Gamma_{-}\Vert_{\infty\rightarrow\infty}
\right)
\leq
\varepsilon''.
\end{equation}
It follows that
\begin{equation}
\begin{split}
\Vert\Phi_{\mathcal{E}_{v,\ell}}^{(k)}(T)\Vert_{\infty}
&\leq
\Vert\Phi_H^{(k)}(T)\Vert_{\infty}
+
\Vert\Delta(T)\Vert_{\infty}\\
&\leq
\Vert\Phi_H^{(k)}(T)\Vert_{\infty}
+
\varepsilon''\Vert T\Vert_{\infty}.
\end{split}
\end{equation}

For an $\ell$-qubit Clifford commutant element acting on the $k$-fold
system, we have the bound
\begin{equation}
\Vert T\Vert_{\infty}\leq 2^{\ell k}.
\end{equation}
Moreover, Haar twirling is trace preserving, and every Clifford
commutant element satisfies $\tr T\geq 1$. Therefore,
\begin{equation}
\begin{split}
2^{\ell k}\Vert\Phi_H^{(k)}(T)\Vert_{\infty}
\geq
\left|\tr\Phi_H^{(k)}(T)\right|=
|\tr T|
\geq 1,
\end{split}
\end{equation}
which implies
\begin{equation}
\Vert\Phi_H^{(k)}(T)\Vert_{\infty}
\geq
2^{-\ell k}.
\end{equation}
Combining the previous two estimates, we obtain
\begin{equation}
\Vert T\Vert_{\infty}
\leq
2^{2\ell k}
\Vert\Phi_H^{(k)}(T)\Vert_{\infty}.
\end{equation}
Hence,
\begin{equation}
\begin{split}
\Vert\Phi_{\mathcal{E}_{v,\ell}}^{(k)}(T)\Vert_{\infty}
&\leq
\Vert\Phi_H^{(k)}(T)\Vert_{\infty}
+
\varepsilon''2^{2\ell k}
\Vert\Phi_H^{(k)}(T)\Vert_{\infty}\\
&=
\left(1+\varepsilon''2^{2\ell k}\right)
\Vert\Phi_H^{(k)}(T)\Vert_{\infty}.
\end{split}
\end{equation}
\end{proof}

Using Lemma~\ref{lem:lblockapprdesign}, one can set $\varepsilon''=2^{-2\ell k}$, and obtain the following inequality
\begin{equation}
\sum_{T\in \Sigma_{k,k}\setminus S_k} \Vert \Phi^{(k)}_{\mathcal{E}_{v,\ell}}(T) \Vert_{\infty}^{2\xi/\ell}\leq k! 2^{k} (\frac{k^4}{2^{\ell-1}})^{2\xi/\ell}.
\end{equation}
Correspondingly, the final approximation error of the magic-augmented circuit (with the approximate unitary designs slotted in) is
\begin{equation}
\begin{split}
&\frac{2^{2\xi k}}{k!}(1+\frac{k^2}{2^{\xi+1}})\Vert [(\Phi^{(k)}_{\mathcal{E}_u}-\Phi^{(k)}_{H})\otimes \id] (P^{EPR})\Vert_{\infty}\\
\leq& 2^{k-\xi+3} + 2^{k+3} (\frac{k^4}{2^{\ell-1}} )^{2\xi/\ell}\\
=& 2^{k-\xi+3} + 2^{k+3 - 2\xi [1 - (4\log k + 1)/\ell]}\\
\leq& 2^{k-\xi+3} + 2^{k+3 - \xi (1 - 2/\ell)}.
\end{split}
\end{equation}
Thus, choosing $\ell \geq 8\log k$ will let the approximation error be upper-bounded by $2^{-\Omega(\xi)+\cO(k)}$, which finishes the proof of Theorem~\ref{thm:MagicAugmented}.\\

We end this section with the proofs of Lemmas \ref{lem:permutation-block-wg-discrepancy}, \ref{lem:CliffordWeingartenSum} and~\ref{lem:twirledCliffordCommutant}.

\begin{proof}[Proof of Lemmas~\ref{lem:permutation-block-wg-discrepancy} and~\ref{lem:CliffordWeingartenSum}]
Here, we set $d = 2^{\xi}$ for the system dimension. Let 
\begin{equation}
A = d^k((\Wg_C)_{P,P} - \Wg),
\end{equation}
where $(\Wg_C)_{P,P}$ is the restriction of $\Wg_C$ on the permutation block. Then our target is the absolute summation of all elements of $A$, which is upper bounded by the product of number of rows, $k!$, and the maximum absolute row sum of $A$,
\begin{align}
\sum_{T_1,T_2\in S_k}d^{k}\abs{\Wg_C(T_1,T_2)-\Wg(T_1,T_2)} \leq k!\Vert A\Vert_{\infty\rightarrow\infty}\\
\Vert A\Vert_{\infty\rightarrow\infty} = \max_{1\leq i\leq k!}\sum_{j=1}^{k!}\abs{A_{ij}}.
\end{align}
The remaining part is to get the upper bound for $\Vert A\Vert_{\infty\rightarrow\infty}$. 

Recall that the Weingarten function is the inverse of the Gram matrix. We have that $\Wg_C = G_C^{-1}$ and $\Wg = G_H^{-1}$. For further elaboration, we define the normalized Gram matrix,
\begin{equation}
\Gamma_C = d^{-k}G_C, \quad \Gamma_H = d^{-k}G_H,
\end{equation}
where $\Gamma_C(T,T') = d^{-\delta(T,T')}$ and $\Gamma_H(\pi,\sigma) = d^{-\delta(T_\pi,T_\sigma)}$. By definition of the Gram matrix, within the permutation block, $(\Gamma_C)_{P,P} = \Gamma_H$. Particularly, let $P=S_k$ and $Q=\Sigma_{k,k}\setminus S_k$. With respect to the decomposition $P\sqcup Q$, we have that
\begin{equation}
\Gamma_C = \begin{pmatrix}
\Gamma_H & B\\
B^T & \Gamma_Q
\end{pmatrix},
\end{equation}
where $B$ is the non-diagonal block, representing the interaction between the permutation elements and the non-permutation elements.

With the notion of the normalized Gram matrix, $A = (\Gamma_C^{-1})_{P,P} - \Gamma_H^{-1}$. By the Schur complement formula,
\begin{equation}
(\Gamma_C^{-1})_{P,P} = \left(\Gamma_H - B\Gamma_Q^{-1}B^T\right)^{-1}.
\end{equation}
Thus,
\begin{equation}
\begin{split}
A &= \left(\Gamma_H - B\Gamma_Q^{-1}B^T\right)^{-1} - \Gamma_H^{-1}\\
&= \left(\Gamma_H - B\Gamma_Q^{-1}B^T\right)^{-1}\left[I - \left(\Gamma_H - B\Gamma_Q^{-1}B^T\right)\Gamma_H^{-1} \right]\\
&=\left(\Gamma_H - B\Gamma_Q^{-1}B^T\right)^{-1}B\Gamma_Q^{-1}B^T\Gamma_H^{-1}\\
&= (\Gamma_C^{-1})_{P,P}B\Gamma_Q^{-1}B^T\Gamma_H^{-1}.
\end{split}
\end{equation}
Since the matrix norm is the sub-multiplicative, we have that
\begin{equation}
\begin{split}
\Vert A\Vert_{\infty\rightarrow\infty}&\leq 
\Vert (\Gamma_C^{-1})_{P,P}\Vert_{\infty\rightarrow\infty}\Vert B\Vert_{\infty\rightarrow\infty}\Vert \Gamma_Q^{-1}\Vert_{\infty\rightarrow\infty}\Vert B^T\Vert_{\infty\rightarrow\infty}\Vert \Gamma_H^{-1}\Vert_{\infty\rightarrow\infty}\\
&\leq \Vert \Gamma_C^{-1}\Vert_{\infty\rightarrow\infty}\Vert \Gamma_Q^{-1}\Vert_{\infty\rightarrow\infty}\Vert \Gamma_H^{-1}\Vert_{\infty\rightarrow\infty}\Vert B\Vert_{\infty\rightarrow\infty}\Vert B^T\Vert_{\infty\rightarrow\infty}.
\end{split}
\end{equation}
We now bound the norm of each term separately. We first focus on $\Gamma_C^{-1}$. Define
\begin{equation}
A_C = \Gamma_C - \id,
\end{equation}
where $\forall T, (A_C)_{T,T} = 0$. Then we have that
\begin{equation}
\Gamma_C^{-1} = (\id+A_C)^{-1} = \sum_{r=0}^{\infty} (-A_C)^r,
\end{equation}
if $\Vert A_C\Vert_{\infty\rightarrow\infty} < 1$. In this case,
\begin{equation}
\Vert \Gamma_C^{-1}\Vert_{\infty\rightarrow\infty} \leq \sum_{r=0}^{\infty} \Vert A_C\Vert_{\infty\rightarrow\infty}^r = \frac{1}{1-\Vert A_C\Vert_{\infty\rightarrow\infty}}.
\end{equation}
By definition, we have that
\begin{equation}
\begin{split}
\Vert A_C\Vert_{\infty\rightarrow\infty}
&=\max_{T\in\Sigma_{k,k}}\sum_{T'\neq T}d^{-\delta(T,T')}\\
&= \sum_{m=1}^{k-1}2^{m(m-1)/2}\binom{k-1}{m}_{2} d^{-m}\\
&=\prod_{i=0}^{k-2}(1+2^id^{-1})-1,
\end{split}
\end{equation}
where the last line utilizes Eq.~\eqref{eq:clifweingartensum}. Since
\begin{equation}
\sum_{i=0}^{k-2}2^id^{-1}\leq 2^{k-1}d^{-1} \leq\frac{1}{2},
\end{equation}
we can use induction method to obtain that
\begin{equation}
\Vert A_C\Vert_{\infty\rightarrow\infty} = \prod_{i=0}^{k-2}(1+2^id^{-1})-1\leq 2\cdot 2^{k-1}d^{-1} = 2^k d^{-1}.
\end{equation}
Since $k\leq \xi-1$ and $d=2^\xi$, we have that $\Vert A_C\Vert_{\infty\rightarrow\infty}\leq 0.5$. To obtain a tighter upper bound, we use the following inequality,
\begin{equation}
\Vert A_C\Vert_{\infty\rightarrow\infty}\leq \prod_{r=3}^{\infty}(1+2^{-r})-1 < 0.272.
\end{equation}
As a result,
\begin{equation}
\Vert \Gamma_C^{-1}\Vert_{\infty\rightarrow\infty}\leq \frac{1}{1-\Vert A_C\Vert_{\infty\rightarrow\infty}} < 1.38.
\end{equation}
By a similar argument,
\begin{equation}
\Vert\Gamma_Q^{-1}\Vert_{\infty\rightarrow\infty}\leq\frac{1}{1-\Vert A_C\Vert_{\infty\rightarrow\infty}} < 1.38.
\end{equation}
One can do the same derivation for the Haar Weingarten function, which has been derived in Ref.~\cite{harrow2023permutation}. We have that
\begin{equation}
\Vert \Gamma_H^{-1}\Vert_{\infty\rightarrow\infty} \leq \frac{1}{1-\Vert A_H\Vert_{\infty\rightarrow\infty}},
\end{equation}
where $A_H = \Gamma_H-\id$ and
\begin{equation}
\Vert A_H\Vert_{\infty\rightarrow\infty}\leq\exp\left(\frac{k(k-1)}{2d}\right)-1=o(1).
\end{equation}
Thus, for a sufficiently large system, $\Vert \Gamma_H^{-1}\Vert_{\infty\rightarrow\infty} < 1.01$.

The remaining terms are $\Vert B \Vert_{\infty\rightarrow\infty}$ and $\Vert B^T\Vert_{\infty\rightarrow\infty}$. For the first term, we have that
\begin{equation}
\begin{split}
\Vert B \Vert_{\infty\rightarrow\infty} = \max_{\pi\in S_k}\sum_{T\in \Sigma_{k,k}\setminus S_k}d^{-\delta(T_\pi,T)}\leq \Vert A_C\Vert_{\infty\rightarrow\infty}\leq 2^k d^{-1}.
\end{split}
\end{equation}
For the second term, we have that
\begin{equation}
\Vert B^T\Vert_{\infty\rightarrow\infty} = \max_{T\in \Sigma_{k,k}\setminus S_k}\sum_{\pi\in S_k} d^{-\delta(T,T_\pi)} = \max_{T\in \Sigma_{k,k}\setminus S_k}d^{-k} \sum_{\pi\in S_k} \cbraket{T}{T_{\pi}}.
\end{equation}
Recall that the $k$-th moment operator of Haar random states is
\begin{equation}
\int_{\psi}d\psi \ketbra{\psi}^{\otimes k} = \frac{1}{(d+k-1)\cdots (d+1)d}\sum_{\pi} T_{\pi}.
\end{equation}
We have that
\begin{equation}
\begin{split}
\Vert B^T\Vert_{\infty\rightarrow\infty} &= \prod_{i=1}^{k-1}(1+id^{-1})\max_{T\in \Sigma_{k,k}\setminus S_k} \tr( T \int_{\psi}d\psi \ketbra{\psi}^{\otimes k})\\
&\leq \exp\left(\frac{k(k-1)}{2d}\right) \max_{T\in \Sigma_{k,k}\setminus S_k} \tr(T \int_{\psi}d\psi \ketbra{\psi}^{\otimes k}).
\end{split}
\end{equation}
From Theorem~10 in Ref.~\cite{Bittel2026operational}, we have that
\begin{equation}
\tr(T \int_{\psi}d\psi \ketbra{\psi}^{\otimes k}) \leq \tr(T_4 \int_{\psi}d\psi \ketbra{\psi}^{\otimes 4}) = \frac{4}{d+3},
\end{equation}
where $T_4 = \frac{1}{d}\sum_{P\in \mathbb{P}_\xi}P^{\otimes 4}$ is the non-permutation Clifford commutant element on 4 copies.

Thus, for sufficiently large $\xi$, $d = 2^{\xi} \gg k(k-1)$, $\exp\left(\frac{k(k-1)}{2d}\right) < 1.01$, and $4/(d+3)<4.01d^{-1}$. We obtain that
\begin{equation}
\Vert B^T\Vert_{\infty\rightarrow\infty}\leq 4.1d^{-1}.
\end{equation}

Combining the preceding estimates gives
\begin{equation}
\begin{split}
\Vert A\Vert_{\infty\rightarrow\infty}&\leq \Vert \Gamma_C^{-1}\Vert_{\infty\rightarrow\infty}\Vert \Gamma_Q^{-1}\Vert_{\infty\rightarrow\infty}\Vert \Gamma_H^{-1}\Vert_{\infty\rightarrow\infty}\Vert B\Vert_{\infty\rightarrow\infty}\Vert B^T\Vert_{\infty\rightarrow\infty}\\
&\leq 1.38^2\times 1.01\times 4.1 \times 2^{k}d^{-2}\\
&\leq 2^{k+3}d^{-2}.
\end{split}
\end{equation}
Therefore,
\begin{equation}
\sum_{T_1,T_2\in S_k}d^{k}\abs{\Wg_C(T_1,T_2)-\Wg(T_1,T_2)} \leq k!2^{k+3}d^{-2}.
\end{equation}
This proves Lemma~\ref{lem:permutation-block-wg-discrepancy}.

To prove Lemma~\ref{lem:CliffordWeingartenSum}, we notice that
\begin{equation}
d^k\Wg_C = (\id + A_C)^{-1} = \sum_{r=0}^{\infty} (-A_C)^r,
\end{equation}
and
\begin{equation}
\begin{split}
d^k \abs{\Wg_C(T, T)}&= \abs{\sum_{r=0}^{\infty}(-A_C)^r_{T,T}}\\
&\leq \sum_{r=0}^{\infty} \abs{(-A_C)^r_{T,T}}\\
&\leq \sum_{r=0}^{\infty} \Vert A_C^r\Vert_{\infty\rightarrow\infty}\\
&\leq \frac{1}{1-\Vert A_C\Vert_{\infty\rightarrow\infty}}.
\end{split}
\end{equation}
Similarly, given $T$,
\begin{equation}
\begin{split}
\sum_{T'\neq T} d^k\abs{\Wg_C(T, T')} &= \sum_{T'\neq T} \abs{\sum_{r=1}^{\infty} (-A_C)^r_{T,T'}}\\
&\leq \sum_{r=1}^{\infty} \sum_{T'\neq T} \abs{(-A_C)^r_{T,T'}}\\
&\leq \sum_{r=1}^{\infty} \Vert A_C^r\Vert_{\infty\rightarrow\infty}\\
&\leq \frac{\Vert A_C\Vert_{\infty\rightarrow\infty}}{1-\Vert A_C\Vert_{\infty\rightarrow\infty}}.
\end{split}
\end{equation}
Substituting $\Vert A_C\Vert_{\infty\rightarrow\infty}\leq 2^{k-\xi}$ finishes the proof.
\end{proof}

\begin{proof}[Proof of Lemma~\ref{lem:twirledCliffordCommutant}]
First, we clarify a notation. In this proof, we set system dimension $d=2^{\ell}$ and require that $d\geq k^8$.

The quantity $f(T) = \Vert \Phi^{(k)}_H(T) \Vert_{\infty} = \Vert \sum_{\sigma, \pi}\Wg(\sigma, \pi) \cketbra{\sigma}{\pi}\cket{T} \Vert_{\infty}$ is invariant under multiplying permutations. For any permutation $\pi$,
\begin{equation}
f(TT_{\pi}) = \Vert \Phi^{(k)}_H(TT_{\pi}) \Vert_{\infty} = \Vert \Phi^{(k)}_H(T) T_{\pi} \Vert_{\infty} = \Vert \Phi^{(k)}_H(T) \Vert_{\infty} = f(T).
\end{equation}
This invariance allows us to define an equivalence relation on Clifford commutant elements: two commutant elements $T_1$ and $T_2$ are equivalent if and only if they differ by a right permutation, i.e.,
\begin{equation}
T_1 \sim T_2 \iff \exists \pi \in S_k \text{ such that } T_1 = T_2 T_{\pi}.
\end{equation}
By permutation covariance, we can select a canonical representative $\Bar{T}$ for each equivalence class $[T]$, so that $\Bar{T}$ minimizes its distance to the identity element in the equivalence class,
\begin{equation}
\delta(\Bar{T}) \leq \delta(\Bar{T} T_\pi), \quad \forall \pi \in S_k.
\end{equation}
Recall that $\delta(T) = \delta(T, T_e)$ is the distance between $T$ and identity $T_e$. In the following, we restrict our focus to evaluating $f(\Bar{T})$ on these canonical representatives.

We set
\begin{equation}
\delta(\Bar{T}) = m.
\end{equation}
Since we only consider non-permutation commutant, we have that $1\leq m \leq k-1$. For $\pi\in S_k$, let
\begin{equation}
r(\pi)=\delta(T_\pi)=k-c(\pi).
\end{equation}
By Lemma 44 of Ref.~\cite{bittel2025completetheorycliffordcommutant},
\begin{equation}
\delta(\Bar{T}T_\pi)=m+r(\pi)-2\beta(\pi),
\end{equation}
where $\beta(\pi)$ is the number of shared linear dependencies between $\Bar{T}$ and $T_\pi$. We have that $\beta(\pi)\leq \min\{m, r(\pi)\}$. Since $\Bar{T}$ is the representative, we have that
\begin{equation}
\delta(\Bar{T}T_\pi)\geq \max\{m,r(\pi)-m\}.
\end{equation}

The Weingarten estimate gives
\begin{equation}
\begin{split}
f(\Bar{T}) &= \Vert \sum_{\sigma, \pi} \Wg(\sigma, \pi)\cketbra{\sigma}{\pi}\cket{\Bar{T}} \Vert_{\infty}\\
&\leq \sum_{\sigma, \pi}\abs{\Wg(\sigma, \pi)} \tr(T_{\pi}^{\dagger}\Bar{T})\\
&= \frac{1}{d(d-1)\cdots(d-k+1)}\sum_{\pi\in S_k}\tr(\Bar{T}T_{\pi}),
\end{split}
\end{equation}
where we use $\forall \pi, \sum_{\sigma}\abs{\Wg(\sigma, \pi)} = (d-k)!/d!$.

Using
\begin{equation}
\tr(\Bar{T}T_{\pi})=d^{k-\delta(\Bar{T}T_\pi)},
\end{equation}
we get
\begin{equation}
f(\Bar{T})
\le
A_{k,d}
\sum_{\pi\in S_k}d^{-\delta(\Bar{T}T_\pi)},
\end{equation}
where
\begin{equation}
A_{k,d}
=
\frac{d^k}{d(d-1)\cdots(d-k+1)}
=
\prod_{j=1}^{k-1}\left(1-\frac{j}{d}\right)^{-1}.
\end{equation}
Since $d\ge k^8$ and $j\le k-1$,
\begin{equation}
0\le \frac{j}{d}\le \frac{1}{k^7}\le \frac12.
\end{equation}
For $0\le x\le 1/2$,
\begin{equation}
\log(1-x)^{-1}\le 2x.
\end{equation}
Therefore
\begin{equation}
\log A_{k,d}
=
\sum_{j=1}^{k-1}
\log\left(1-\frac{j}{d}\right)^{-1}
\le
\frac{2}{d}
\sum_{j=1}^{k-1}j
=
\frac{k(k-1)}{d}
\le
\frac{1}{k^2}.
\end{equation}
Hence
\begin{equation}
A_{k,d}\le \exp\left(\frac{1}{k^2}\right).
\end{equation}

Now we bound $\sum_{\pi\in S_k}d^{-\delta(\Bar{T}T_\pi)}$. We first group permutations by $r=r(\pi)$. The number of permutations with $r(\pi)=r$ is the unsigned Stirling number
\begin{equation}
\#\{\pi\in S_k:r(\pi)=r\}
=
\begin{bmatrix}
k\\ k-r
\end{bmatrix}.
\end{equation}
Thus
\begin{equation}
\sum_{\pi\in S_k}d^{-\delta(\Bar{T}T_\pi)}
\le
\sum_{r=0}^{k-1}
\begin{bmatrix}
k\\ k-r
\end{bmatrix}
d^{-\max\{m,r-m\}}.
\end{equation}
Splitting at $r=2m$ gives
\begin{equation}
\sum_{\pi\in S_k}d^{-\delta(\Bar{T}T_\pi)}
\le
\Sigma_{-}+\Sigma_{+},
\end{equation}
where
\begin{equation}
\Sigma_{-} = d^{-m}\sum_{r=0}^{2m}
\begin{bmatrix}
k\\ k-r
\end{bmatrix},
\end{equation}
and
\begin{equation}
\Sigma_{+} = d^m\sum_{r=2m+1}^{k-1}
\begin{bmatrix}
k\\ k-r
\end{bmatrix}
d^{-r}.
\end{equation}
We have a bound for the unsigned Stirling number,
\begin{equation}
\begin{bmatrix}
k\\ k-r
\end{bmatrix}
\le
\frac{K^r}{r!},\quad K=\frac{k(k-1)}{2}
\end{equation}
Then,
\begin{equation}
\Sigma_{-}\leq d^{-m}\sum_{r=0}^{2m}\frac{K^r}{r!},
\end{equation}
and
\begin{equation}
\Sigma_{+}\leq d^{m}\sum_{r=2m+1}^{k-1}\frac{1}{r!}(\frac{K}{d})^r.
\end{equation}
We first bound $\Sigma_{-}$. For $0\le r\le 2m$, the ratio of consecutive terms satisfies
\begin{equation}
\frac{K^{r-1}/(r-1)!}{K^r/r!}=\frac{r}{K}\leq \frac{2m}{K}\leq\frac{4}{k}.
\end{equation}
Hence
\begin{equation}
\Sigma_{-}\leq d^{-m}\frac{K^{2m}}{(2m)!}\frac{1}{1-\frac{4}{k}}
\leq
\left(1-\frac{4}{k}\right)^{-1}\frac{1}{4^m(2m)!}\left(\frac{k^4}{d}\right)^m.
\end{equation}
We now bound $\Sigma_{+}$. Since $d\geq k^8$,
\begin{equation}
\frac{K}{d}\le \frac{k^2}{2d}\le \frac{1}{2k^6}\leq \frac{1}{2}.
\end{equation}
Therefore,
\begin{equation}
\Sigma_{+}\leq d^m
\frac{(K/d)^{2m+1}}{(2m+1)!}\sum_{r=0}^{+\infty}(\frac{1}{2})^r
\leq 2d^m\frac{(K/d)^{2m+1}}{(2m+1)!}.
\end{equation}
Rewriting,
\begin{equation}
\Sigma_{+}
\leq
\frac{2K}{d(2m+1)}\frac{K^{2m}}{d^m(2m)!}\leq \frac{1}{3k^2}
\frac{1}{4^m(2m)!}
\left(\frac{k^4}{d}\right)^m.
\end{equation}
Combining the bounds for $\Sigma_{-}$ and $\Sigma_{+}$ gives
\begin{equation}
\sum_{\pi\in S_k}d^{-\delta(\Bar{T}T_\pi)}
\le
\left[
\left(1-\frac{4}{k}\right)^{-1}
+
\frac{1}{3k^2}
\right]
\frac{1}{4^m(2m)!}
\left(\frac{k^4}{d}\right)^m.
\end{equation}
Multiplying the factor $A_{k,d}$, we get
\begin{equation}
\begin{split}
f(\Bar{T})&\leq \exp\left(\frac{1}{k^2}\right)\left[\left(1-\frac{4}{k}\right)^{-1}+\frac{1}{3k^2}\right]\frac{1}{4^m(2m)!}\left(\frac{k^4}{d}\right)^m\\
&\leq \exp\left(\frac{1}{k^2}\right)\left[\left(1-\frac{4}{k}\right)^{-1}+\frac{1}{3k^2}\right]\frac{1}{8}\left(\frac{k^4}{d}\right)^m.
\end{split}
\end{equation}
For $k\geq 5$, $\exp\left(\frac{1}{k^2}\right)\left[\left(1-\frac{4}{k}\right)^{-1}+\frac{1}{3k^2}\right] < 8$, and thus we obtain
\begin{equation}
f(\Bar{T})\le
\left(\frac{k^4}{d}\right)^m.
\end{equation}
Hence, the final target is upper-bounded by
\begin{equation}
\begin{split}
\sum_{T\in \Sigma_{k,k}\setminus S_k} f(T)^{2\xi/\ell}&\leq k! \sum_{\Bar{T}\in \Sigma_{k,k}/S_k \setminus \{e\}} f(\Bar{T})^{2\xi/\ell}\\
&\leq k! \sum_{T\in \Sigma_{k,k}/S_k \setminus \{e\}} \left[(\frac{k^4}{2^\ell})^{\delta(\Bar{T})}\right]^{2\xi/\ell}\\
&\leq k! \sum_{T\in \Sigma_{k,k}\setminus S_k} [(\frac{k^4}{2^\ell})^{2\xi/\ell}]^{\delta(T)}\\
&=k!\left[ \prod_{i=0}^{k-2}(1+2^i(\frac{k^4}{2^\ell})^{2\xi/\ell}) - \prod_{j=1}^{k-1}(1+j(\frac{k^4}{2^\ell})^{2\xi/\ell}) \right]\\
&\leq k!\left[ \prod_{i=0}^{k-2}(1+2^i(\frac{k^4}{2^\ell})^{2\xi/\ell}) - 1 \right]\\
&\leq k! 2^{k} (\frac{k^4}{2^\ell})^{2\xi/\ell}.
\end{split}
\end{equation}
Here, $\Sigma_{k,k}/S_k$ represents the quotient set under permutation transformation symmetry. The last line can be proved by induction and is also utilized in the proof of Lemma~\ref{lem:permutation-block-wg-discrepancy}. The key point is that when $d = 2^{\ell}\geq k^8$, we have that
\begin{equation}
\sum_{i=0}^{k-2}2^i(\frac{k^4}{2^\ell})^{2\xi/\ell}\leq 2^{k-1}(\frac{k^4}{2^\ell})^{2\xi/\ell} \leq\frac{1}{2}.
\end{equation}
Then,
\begin{equation}
\prod_{i=0}^{k-2}(1+2^i(\frac{k^4}{2^\ell})^{2\xi/\ell})\leq 1 + 2^k(\frac{k^4}{2^\ell})^{2\xi/\ell}.
\end{equation}
\end{proof}

\section{Constant depth circuits with constant spectral gap}\label{sc:spectral-gap}
In this section, we prove Theorem~\ref{thm:CPZPCdesign}. We first restate and present a more detailed version.

\begin{theorem}[Constant-depth CPZPC generators and design depth]\label{thm:main-cpzpc}
Let $\ell$ be the number of qubits and $k$ be the design order. Assume
\begin{equation}
k\leq c 2^{\ell/6.1}
\end{equation}
for a sufficiently small constant $c>0$. There is a probability distribution $\widetilde\nu_{\CPZPC,\ell}$ on $\ell$-qubit Clifford and permutation circuits such that the following hold.
\begin{enumerate}
\item Each sample from $\widetilde\nu_{\CPZPC,\ell}$ is realizable by a nearest-neighbor one-dimensional circuit of depth $\cO(1)$.
\item The $k$-th moment operator has constant spectral gap:
\begin{equation}
g(\widetilde\nu_{\CPZPC,\ell},k)\leq 1-\Delta_{\widetilde\nu_{\CPZPC}}
\end{equation}
for a constant $\Delta_{\widetilde\nu_{\CPZPC}}>0$ independent of $\ell$ and $k$.
\end{enumerate}
Consequently, $L$ independent samples from $\widetilde\nu_{\CPZPC,\ell}$ form an $\varepsilon''$-relative-error approximate unitary $k$-design whenever
\begin{equation}
L=\cO\left(\ell k+\log \frac{1}{\varepsilon''}\right).
\end{equation}
Since one sample has depth $\cO(1)$, the total circuit depth is
\begin{equation}
D_{1\mathrm{D}}=\cO\left(\ell k+\log \frac{1}{\varepsilon''}\right).
\end{equation}
\end{theorem}

\begin{proof}[Proof of Theorem~\ref{thm:main-cpzpc}]
As demonstrated in Fig.~\ref{fig:circuitCPZPC}, the probability distribution $\widetilde\nu_{\CPZPC,\ell}$ is constructed by
\begin{equation}\label{eq:implemented-cpzpc}
\widetilde{\nu}_{\CPZPC,\ell}=\nu_{\Cl,\ell}^{*r_C}*\nu_{P,\ell}^{*r_P}*\delta_Z*\nu_{P,\ell}^{*r_P}*\nu_{\Cl,\ell}^{*r_C},
\end{equation}
where $\nu_{P,\ell}$ is a probability distribution over the 1D-local Kassabov generators and $\nu_{\Cl,\ell}$ is a probability distribution over the 1D-local Clifford generators. The integers $r_C$ and $r_P$ are fixed constants chosen later, which make sure that the essential norm $g(\widetilde{\nu}_{\CPZPC,\ell}, k) < 1$, and that the spectral gap $\Delta(\widetilde{\nu}_{\CPZPC,\ell}, k)$ is a constant between $0$ and $1$. To achieve this goal, we prove the existence of 1D-local generators for the permutation group and the Clifford group with constant gaps. We formulate the results as two lemmas below. Recall that the $k$-th moment operator of the permutation group can be simulated by the alternating group~\cite{chen2024incompressibilityspectralgapsrandom}.
\begin{lemma}[Lemma~4.4 in Ref.~\cite{chen2024incompressibilityspectralgapsrandom}]
For all $k \le 2^{\ell} - 2$ we have
\begin{equation}
\E_{\pi \sim \mu(\Alt_U(2^{\ell}))} P(\pi)^{\otimes k} = \E_{\pi \sim \mu(\Sym_U(2^{\ell}))} P(\pi)^{\otimes k},
\end{equation}
where $P(\pi)\ket{x} = \ket{\pi(x)}$ for all $x \in [2^{\ell}]$ and $\pi \in \Sym_U(2^{\ell})$.
\end{lemma}
We only need to show the efficient realization of Kassabov's generators for the alternating group. In fact, it is implicitly proved in Section 5.2 of Ref.~\cite{chen2024incompressibilityspectralgapsrandom} that all the Kassabov's generators for the alternating group can be realized in a 1D system with a periodic boundary condition. Here, we use a folded lay-out to show the realization is also available in a 1D system with an open boundary condition. The result is summarized below.
\begin{lemma}[Improved from Theorem B.1 in Ref.~\cite{chen2024incompressibilityspectralgapsrandom}]\label{lem:1d-kassabov}
For any integers $\ell\geq 1$ and $k\leq \cO(2^{\ell/6.1})$, there exists an explicit set $S$ such that each element in $S$ is a product of $\mathrm{NOT}$, $\mathrm{CNOT}$, and Toffoli gates with circuit depth $\cO(1)$ in a one-dimensional nearest-neighbor system, and
\begin{equation}
g(\mu(S), \tau, \Alt(2^{\ell})) = 1-\Omega(1),
\end{equation}
where $\tau: \pi\mapsto P^{\otimes k}(\pi)$ and $P(\pi)\ket{z} = \ket{\pi(z)}$ for any $z\in \{0, 1\}^{\ell}$ and $\pi\in\Sym(2^{\ell})$.
\end{lemma}

We leave the proof of Lemma~\ref{lem:1d-kassabov} to the end of this section. Besides the alternating group, we  must also construct constant depth 1D-local generators for the Clifford group. This construction is given in the following Lemma:
\begin{lemma}[Constant-depth 1D Clifford generators]\label{lem:1d-clifford}
For any integer $\ell\geq 1$, there exists a constant-size symmetric generating set $S_{\Cl,\ell} \subset \Cl_{\ell}$ such that each element in $S_{\Cl,\ell}$ is a product of single-qubit Pauli, Hadamard, phase, $\mathrm{CZ}$, and $\mathrm{CNOT}$ gates with circuit depth $\cO(1)$ in a one-dimensional nearest-neighbor system, and
\begin{equation}
    g(\mu(S_{\Cl,\ell}), \rho, \Cl_{\ell}) = 1-\Omega(1),
\end{equation}
for any finite-dimensional unitary representation $\rho$ of $\Cl_{\ell}$ with no trivial subrepresentation.
\end{lemma}

The proof of Lemma~\ref{lem:1d-clifford} is presented in next section (section \ref{sc:clif}), where we start from constructing the generating set for the binary symplectic group and extend it to the Clifford group.\\

We now apply Lemma~\ref{lem:1d-kassabov} and Lemma~\ref{lem:1d-clifford} to prove Theorem~\ref{thm:main-cpzpc}. First, we notice that the representation $P^{\otimes k}(\pi)\otimes \overline{P}^{\otimes k}(\pi) = P^{\otimes 2k}(\pi)$ since this representation is real. Thus, from the spectral gap guarantee in Lemma~\ref{lem:1d-kassabov}, as long as $2k\leq \cO(2^{\ell/6.1})$, then there exists a positive constant $\Delta_P = \Omega(1)$, such that
\begin{equation}
\Vert M_k(\nu_{P}^{*r}) - M_k(\mu(\Alt(2^{\ell})))\Vert_{\infty} \leq (1-\Delta_P)^r.
\end{equation}
Similarly, from Lemma~\ref{lem:1d-clifford}, there exists a positive constant $\Delta_{\Cl} = \Omega(1)$ such that
\begin{equation}
\Vert M_k(\nu_{\Cl}^{*r}) - M_k(\mu(\Cl_{\ell}))\Vert_{\infty} \leq (1-\Delta_{\Cl})^r.
\end{equation}
Thus, for $k \leq \cO(2^{\ell/6.1})$, by triangle inequality, we bound the difference between the implemented generating set $\widetilde{\nu}_{\CPZPC,\ell}=\nu_{\Cl,\ell}^{*r_C}*\nu_{P,\ell}^{*r_P}*\delta_Z*\nu_{P,\ell}^{*r_P}*\nu_{\Cl,\ell}^{*r_C}$ and the ideal CPZPC ensemble:
\begin{equation}
    \Vert M_k(\widetilde\nu_{\CPZPC,\ell}) - M_k(\nu_{\CPZPC})\Vert_{\infty} 
    \leq 2(1-\Delta_{\Cl})^{r_C} + 2(1-\Delta_P)^{r_P}.
\end{equation}
Because $\Delta_{\Cl}$ and $\Delta_P$ are strictly positive constants, we can choose fixed integers $r_C = \cO(1)$ and $r_P = \cO(1)$, independent of $\ell$, $k$, and the design precision $\varepsilon$, sufficiently large such that:
\begin{equation}
    2(1-\Delta_{\Cl})^{r_C} + 2(1-\Delta_P)^{r_P} \leq \frac{1}{4}.
\end{equation}
For the ideal CPZPC ensemble, it is guaranteed that the essential norm is bounded by $\cO(k/2^{\ell/2})$~\cite{chen2024incompressibilityspectralgapsrandom}. For sufficiently large $\ell$, this difference is strictly away from 1:
\begin{equation}
g(\nu_{\CPZPC},k)=\Vert M_k(\nu_{\CPZPC})-P_H\Vert_\infty\leq \frac{1}{4}.
\end{equation}
Applying the triangle inequality, the essential norm for $\widetilde{\nu}_{\CPZPC,\ell}$ is bounded by:
\begin{equation}
\begin{split}
g(\widetilde\nu_{\CPZPC,\ell},k) &= 
\Vert M_k(\widetilde\nu_{\CPZPC,\ell})-P_H\Vert_\infty\\
&\leq \Vert M_k(\nu_{\CPZPC})-P_H\Vert_\infty + \Vert M_k(\widetilde\nu_{\CPZPC,\ell}) - M_k(\nu_{\CPZPC})\Vert_\infty\\
&\leq \frac{1}{4} + \frac{1}{4} < \frac{1}{2}.
\end{split}
\end{equation}
Thus, the spectral gap of the implemented CPZPC generator circuit satisfies $\Delta_{\widetilde\nu_{\CPZPC}} \geq 1 - 1/2 = 1/2$. We have proved property $2$ of Theorem~\ref{thm:main-cpzpc}.\\

We now evaluate the physical circuit complexity. Each sample drawn from $\widetilde\nu_{\CPZPC,\ell}$ consists of $2 r_C$ Clifford generators, $2 r_P$ Kassabov permutation generators, and one Pauli $Z$ gate. As proved in Lemmas~\ref{lem:1d-kassabov} and~\ref{lem:1d-clifford}, every individual generator admits a 1D nearest-neighbor circuit of depth $\cO(1)$. Because $r_C$ and $r_P$ are positive constants, the entire circuit requires a 1D circuit depth of:
\begin{equation}
2 r_C \cdot \cO(1) + 2 r_P \cdot \cO(1) + 1 = \cO(1).
\end{equation}
This establishes property 1 of Theorem~\ref{thm:main-cpzpc}. Following a standard spectral-gap-to-relative-design reduction or Lemma~\ref{lemma:gap-to-design}, we get an $\ell$-qubit $\varepsilon''$-relative-error approximate unitary $k$-design at depth $\cO\left(\ell k+\log \frac{1}{\varepsilon''}\right)$.
\end{proof}
We end this section by providing a proof of Lemma~\ref{lem:1d-kassabov}:

\begin{proof}[Proof of Lemma~\ref{lem:1d-kassabov}]
We first discuss the case that $\ell \bmod 18 = 0$ and set $\ell = 18s$. In this case, the authors of Ref.~\cite{chen2024incompressibilityspectralgapsrandom} have already constructed a generating set $S$ with a constant spectral gap, i.e. $g(\mu(S), \tau, \Alt(2^{\ell})) = 1-\Omega(1)$. This set is given as
\begin{equation}
S = \{xyz: x,z\in P_X, y\in S_0\},
\end{equation}
where $P_X = \{I, X\}^{\otimes \ell}$ is the set of all bit-flip operators, and $S_0$ is the set of Kassabov's generators of alternating group $\Alt(K_s)$ with $K_s = (\mathbb F_2^{3s}\setminus\{0^{3s}\})^6 \subseteq (\mathbb F_2^{3s})^6$. Note that $P_X$ only contains local operators, so any operator within this set can be realized in depth $1$. It remains to show all of Kassabov's generators for the alternating group in $S_0$ can be realized in constant depth in 1D. Below, we give a review of Kassabov's generators of $\Alt(K_s)$. We start with defining the qubit layout to realize these generators. The total $\ell$ qubits have already been divided into $18$ patches with each patch having $s$ qubits. The computational basis is hence a bit string belonging to $\F_2^{18s}$. We write an input string as six registers
\begin{equation}
(z_1,z_2,z_3,z_4,z_5,z_6),
\quad
z_r\in\mathbb F_2^{3s}.
\end{equation}
We further split each $3s$-bit register into three $s$-bit blocks:
\begin{equation}\label{eq:kassbovbitlabel}
z_r=
(z_{r,1,1},\ldots,z_{r,1,s},
 z_{r,2,1},\ldots,z_{r,2,s},
 z_{r,3,1},\ldots,z_{r,3,s}).
\end{equation}
Here $r\in\{1,\ldots,6\}$ labels the register, $\alpha\in\{1,2,3\}$ labels the block inside the register, and $a\in\{1,\ldots,s\}$ labels the coordinate inside an $s$-bit block. Next, we group all qubits with the same coordinate and define the coordinate cell
\begin{equation}
C_a=\{z_{r,\alpha,a}:1\leq r\leq 6,\ 1\leq \alpha\leq 3\}.
\end{equation}
Each cell contains exactly $6\times 3=18$ qubits, which is a constant. Next we lightly deviate from \cite{chen2024incompressibilityspectralgapsrandom}. We place the cells on a one-dimensional line in a \emph{folded order}:
\begin{equation}
\label{eq:kassabov-folded-order}
C_1,\ C_s,\ C_2,\ C_{s-1},\ C_3,\ C_{s-2},\ldots .
\end{equation}
This ordering has the following useful property: every cyclic-neighbor pair $C_a,C_{a+1\pmod s}$ is separated by at most one cell and interacts with at most 54 qubits. Therefore, a gate whose inputs are contained in one cell, or in two cyclic-neighbor cells, can be routed in constant nearest-neighbor depth in a one-dimensional system. We summarize and generalize this observation below and depict the diagram of the folded order in Fig.~\ref{fig:kassabov-folded-cells}. We now make the following observation:

\begin{observation}[Constant-size local routing]\label{observation:constant-routing}
Suppose a layer of constant-local gates, including CNOT and Toffoli gates, has the following property: every gate is supported either inside one cell $C_a$ or inside a constant number of neighboring cells, and each cell participates in only $\cO(1)$ gates. Then the whole layer has one-dimensional nearest-neighbor depth $\cO(1)$.
\end{observation}

\begin{proof}
A constant number of neighboring cells contains only $\cO(1)$ bits, because each cell contains $18$ bits.  Inside such a constant-size window, we can use a constant number of SWAP gates to bring the relevant two or three bits next to each other, apply the CNOT or Toffoli gate, and then undo the SWAPs. The interaction graph between cells has bounded degree, so the gates can be colored into $\cO(1)$ non-overlapping batches. Gates in the same batch act on disjoint constant-size windows and can be run in parallel. Hence the depth is $\cO(1)$.
\end{proof}

\begin{figure}[!t]
\centering
\begin{tikzpicture}[
    x=1.0cm,y=0.8cm,
    cell/.style={draw,rounded corners,minimum width=0.9cm,minimum height=0.55cm,fill=gray!10},
    arr/.style={-{Latex},thick,blue}
]
\foreach \i/\lab in {0/$C_1$,1/$C_s$,2/$C_2$,3/$C_{s-1}$,4/$C_{3}$,5/$C_{s-2}$,6/$C_{4}$,7/$\cdots$}{
  \node[cell] (c\i) at (\i,0) {\lab};
}
\draw[decorate,decoration={brace,amplitude=4pt}] (5.55,0.55) -- (6.45,0.55)
node[midway,above=4pt] {$18$ bits};
\draw[arr,bend left=35] (c0.north) to node[above] {$C_1\sim C_2$} (c2.north);
\draw[arr,bend left=35] (c2.north) to node[above] {$C_2\sim C_3$} (c4.north);
\draw[arr,bend left=35] (c1.south) to node[below] {$C_s\sim C_1$} (c0.south);
\draw[arr,bend left=35] (c5.south) to node[below] {$C_{s-2}\sim C_{s-1}$} (c3.south);
\end{tikzpicture}
\caption{Folded cell order.  The cyclic shift matrix used by Kassabov couples $C_a$ only to $C_{a+1\pmod s}$.  In this ordering, those cyclic neighbors are constant-distance neighbors on an open one-dimensional line.}
\label{fig:kassabov-folded-cells}
\end{figure}
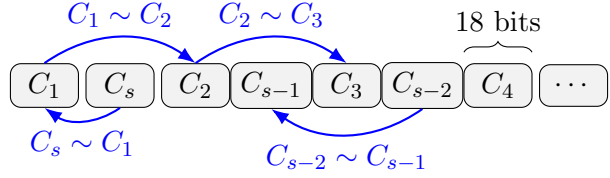

Now we introduce Kassabov's generators for the alternating group and discuss their realization. We refer the reader to Section 5.1 in Ref.~\cite{chen2024incompressibilityspectralgapsrandom} for the mathematical details of the construction. The alternating group $\Alt(K_s)$ is generated by the union set~\cite{Kassabov2007symmetric},
\begin{equation}\label{eq:alternating_union}
\bigcup_{i=1}^{6}\pi_i(\Gamma_0).
\end{equation}
Here, $\Gamma_0$ is the generating set of $\Gamma = \SL(3s; \F_2)^{\times K^5}$ with $K = 2^{3s}-1$, and $\pi_i$ is the embedding map from the (direct product) special linear group $\Gamma$ to the alternating group $\Alt(K_s)$. $\SL(m;\F_2)$ represents the special linear group on $m\times m$ matrices with field $\F_2$. Recall that the input is written as $(z_1,z_2,z_3,z_4,z_5,z_6)\in (\F_2^{3s})^{\times 6}$. We denote $z_{\Bar{i}} = (z_1,\cdots z_{i-1}, z_{i+1}, \cdots z_6)$. For any element $g = (g_w)_{w\in (\F_2^{3s}\setminus\{0^{3s}\})^{\times 5}}\in \Gamma$, $g_w\in \SL(3s; \F_2)$, the map $\pi_i$ satisfies
\begin{equation}
\pi_i(g)(z_1,z_2,z_3,z_4,z_5,z_6) = (z_1,\cdots z_{i-1}, g_{z_{\Bar{i}}}z_i, z_{i+1}, \cdots z_6).
\end{equation}
Note that $\SL(3s;\mathbb F_2)$ naturally permutes $\F_2^{3s}\setminus\{0^{3s}\}$, and the induced permutation is even. Thus, $\pi_i(g)$ is a valid alternating group element in $\Alt(K_s)$. Thanks to Eq.~\eqref{eq:alternating_union}, we only need to focus on the generating set $\Gamma_0$ for $\Gamma = \SL(3s; \F_2)^{\times K^5}$ and its mapping $\pi_1(\Gamma_0)$. Note that the $\pi_i$ are symmetric for different $i$, so we only need to focus on the realization of $\pi_1(\Gamma_0)$.\\

In Ref.~\cite{chen2024incompressibilityspectralgapsrandom}, the authors actually discuss the generating set for a larger group $\SL(3s; \F_2)^{\times 2^{15s}}$ instead of $\SL(3s; \F_2)^{\times K^5}$, which still covers the target group. We follow this approach here. To construct the generating set for $\SL(3s; \F_2)^{\times 2^{15s}}$, the key point is to utilize the following isomorphisms:
\begin{equation}
\SL(3s; \F_2)^{\times 2^{15s}} = \EL(3s; \F_2)^{\times 2^{15s}}\cong \EL(3s; \F_2^{\times 2^{15s}}),
\end{equation}
and the following surjective group homomorphisms,
\begin{equation}\label{eq:homomorphism}
\begin{split}
&\EL(3;\mathbb{Z}\langle a,b,x_1, \cdots, x_{15}\rangle)\\
\xrightarrow{\ \varphi_1\ }&\EL(3;\Mat(s;\F_2)[x_{1,1},\ldots,x_{15,s}]) = \EL(3s;\F_2[x_{1,1},\ldots,x_{15,s}])\\
\xrightarrow{\ \varphi_2\ }&\EL(3s; \F_2^{\times 2^{15s}}).
\end{split}
\end{equation}
Here, $\EL(m;R)$ is the multiplicative matrix group generated by all elementary matrices $E_{i,j}(r)=I_m + re_{i,j}$, with $1\leq i\neq j\leq m$, and $r$ taken from a unital ring $R$. It satisfies $\EL(m;R)^{\times k}\cong \EL(m;R^{\times k})$. When $R$ is a field, we have that $\EL(m;R) = \SL(m;R)$. The element $x_c = \diag(x_{c,1}, \cdots, x_{c,s})$ contains $s$ variables. The equation in the second line is from Lemma~5.10 in Ref.~\cite{chen2024incompressibilityspectralgapsrandom}, which is listed below.
\begin{lemma}[Lemma 5.10 in Ref.~\cite{chen2024incompressibilityspectralgapsrandom}]
Let $R$ be any unital ring that may be noncommutative. Then, for any $n,m\geq 2$ we have
\begin{equation}
\EL(n;\Mat(m;R)) = \EL(nm;R).
\end{equation}
\end{lemma}
The homomorphism $\varphi_1$ satisfies
\begin{align}
\varphi_1(1) &= I_s = \begin{pmatrix}
1 & 0 & \cdots & 0\\
0 & 1 & \cdots & 0\\
\vdots & \vdots & \ddots & \vdots\\
0 & 0 & \cdots & 1
\end{pmatrix}, \tag{Identity matrix} \\
\varphi_1(a) &= A = \begin{pmatrix}
0 & 0 & 0 &\cdots & 1\\
1 & 0 & 0 & \cdots & 0\\
0 & 1 & 0 & \cdots & 0\\
\vdots & \vdots & \vdots & \ddots & \vdots\\
0 & 0 & 0 & \cdots & 0
\end{pmatrix}, \tag{Cyclic shift matrix, $A e_a = e_{a+1\pmod s}$} \\
\varphi_1(b) &= B = E_{11} = \begin{pmatrix}
1 & 0 & \cdots & 0\\
0 & 0 & \cdots & 0\\
\vdots & \vdots & \ddots & \vdots\\
0 & 0 & \cdots & 0
\end{pmatrix}, \tag{Single non-zero entry at top-left} \\
\varphi_1(x_c) &=  \operatorname{diag}(x_{c,1}, \ldots, x_{c,s}) = \begin{pmatrix}
x_{c,1} & 0 & \cdots & 0\\
0 & x_{c,2} & \cdots & 0\\
\vdots & \vdots & \ddots & \vdots\\
0 & 0 & \cdots & x_{c,s}
\end{pmatrix}. \tag{Diagonal matrix controlled by $x_c$}
\end{align}
and the homomorphism $\varphi_2$ satisfies:
\begin{equation}
\begin{split}
\varphi_2: \F_2[x_{1,1},\ldots,x_{15,s}]
\rightarrow
\prod_{x\in\F_2^{15s}}\F_2
\cong
\F_2^{\times 2^{15s}},\quad
\ker \varphi_2 = \left(x_{c,a}^2-x_{c,a}\right)_{\substack{1\le c\le 15\\1\le a\le s}}.
\end{split}
\end{equation}
Note that $A$ and $B$ generate $\Mat(s; \F_2)$, including all elementary matrices $E_{ij}$. Then, with multiplication in the form of $E_{ii}\varphi_1(x_c)E_{ii}$, one can get any element $x_{c,i}$. This is why $\varphi_1$ is a surjective group homomorphism from $\mathbb{Z}\langle a,b,x_1, \cdots, x_{15}\rangle)$ to $\EL(3;\Mat(s;\F_2)[x_{1,1},\ldots,x_{15,s}])$. The map $\varphi_2$ is a canonical ring homomorphism from $\F_2[x_{1,1},\ldots,x_{15,s}]$ onto a quotient ring. Its function is to transform the variable $x_{c,a}$ into a bit (i.e. an element of $\F_2$) by evaluation. By definition, $\varphi_2$ acts trivially on $A$ and $B$. That is, $\varphi_2(A) = A$, and $\varphi_2(B) = B$.\\

Meanwhile, we have the following lemma to generate $\EL(3; R)$ where $R$ is a free ring.
\begin{lemma}[Lemma 5.11 in Ref.~\cite{chen2024incompressibilityspectralgapsrandom}]\label{lem:ELgenerator}
An elementary matrix group over a free ring with $l$ generators, $\EL(3;\mathbb{Z}\langle w_1,\ldots,w_l\rangle)$, is generated by the following $4l+6$ elements and their inverses:
\begin{equation}
E_{\alpha,\beta}(1)\quad (\alpha\neq \beta\in \{1, 2, 3\}),\quad
E_{1,2}(w_j),\ E_{2,3}(w_j),\ E_{2,1}(w_j),\ E_{3,2}(w_j),\quad 1\leq j\leq l.
\end{equation}
Note that $E_{\alpha,\beta}(r)=I_s + re_{\alpha,\beta}$.
\end{lemma}
This means the generating set of $\SL(3s; \F_2)^{\times 2^{15s}}$ is given by
\begin{equation}
\begin{split}
\Gamma_0 = &\bigg\{(\varphi_2\circ \varphi_1)(E_{\alpha,\beta}(1)), (\varphi_2\circ \varphi_1)(E_{\alpha',\beta'}(a)), (\varphi_2\circ \varphi_1)(E_{\alpha',\beta'}(b)), (\varphi_2\circ \varphi_1)(E_{\alpha',\beta'}(x_c))\;\|\\ 
&\hspace{10em}\alpha\neq\beta\in \{1,2,3\}, (\alpha',\beta')\in \{(1,2), (2,3), (2,1), (3,2)\} \bigg\}.
\end{split}
\end{equation}
Below, we discuss the circuit realization of all generators. This can be obtained by studying the action of $\pi_1(\Gamma_0)$ on the computational basis. The results have already been given in Section 5.2 in Ref.~\cite{chen2024incompressibilityspectralgapsrandom}. Below, we list the results and show that they can be realized in constant depth in a 1D circuit with open boundary conditions. To find the circuit corresponding to $(\pi_1\circ\varphi_2\circ\varphi_1)(E_{\alpha, \beta}(r))$, we look at its action on $z_1 = (z_{1,1,1},\ldots,z_{1,1,s},
 z_{1,2,1},\ldots,z_{1,2,s},
 z_{1,3,1},\ldots,z_{1,3,s})$. We further denote $y_{\alpha,a}=z_{1, \alpha, a}$, and define
\begin{equation}
y_{\alpha} = (y_{\alpha, 1}, y_{\alpha, 2}, \cdots, y_{\alpha, s})^T.
\end{equation}
For embedding map $\pi_1$, the first register acts as the target register, and the remaining five registers act as controls. They provide $5 \times 3s = 15s$ control bits, corresponding to
\begin{equation}
x_{c,a}, \quad c\in\{1,\ldots,15\}, \quad a\in\{1,\ldots,s\},
\end{equation}
which appear in Eq.~\eqref{eq:homomorphism}. In fact, $x_{c,a}$ is $z_{r_c, \alpha, a}$ defined in Eq.~\eqref{eq:kassbovbitlabel} with $r_c\neq 1$ and $1\leq \alpha\leq 3$. Crucially, for any fixed coordinate $a$, the target bits $y_{\alpha,a}$ and the control bits $x_{c,a}$ are all localized within the same physical cell $C_a$. On the input bits $(y_1, y_2, y_3)$, the action of $R = (\varphi_2\circ\varphi_1)(E_{\alpha, \beta}(r))$ is 
\begin{equation}
y_{\alpha}\mapsto y_{\alpha} \oplus (\varphi_2\circ\varphi_1)(r) y_{\beta}.
\end{equation}
This corresponds to a CNOT circuit between qubits $y_{\alpha}$ and $y_{\beta}$. We now analyze its 1D realization for $r$ chosen from each of the four types $\{1, a, b, x_c\}$. We demonstrate that every generator type satisfies the conditions of Observation~\ref{observation:constant-routing}, thus achieving $\cO(1)$ depth.
\begin{enumerate}
\item When $r=1$, the circuit is composed of bitwise CNOTs between blocks:
\begin{equation}\label{eq:slcircuit1}
y_{\alpha,a} \mapsto y_{\alpha,a} \oplus y_{\beta,a}, \quad a=1,\ldots,s.
\end{equation}
This is a parallel layer of CNOT gates $\mathrm{CNOT}_{y_{\beta,a}\to y_{\alpha,a}}$. Since both the control $y_{\beta,a}$ and the target $y_{\alpha,a}$ share the exact same coordinate $a$, the CNOT gate is completely contained inside the single cell $C_a$. Different values of $a$ use disjoint cells, allowing full parallelization. Thus, the layer has a 1D depth of $\cO(1)$.

\item When $r=a$, the transformation offsets the target coordinate by $1$:
\begin{equation}\label{eq:slcircuita}
y_{\alpha,a+1} \mapsto y_{\alpha,a+1} \oplus y_{\beta,a}, \quad a=1,\ldots,s,
\end{equation}
with the coordinate $a+1$ interpreted modulo $s$. Circuit-wise, this applies a CNOT from a bit in cell $C_a$ to a bit in cell $C_{a+1\pmod s}$. These are the only non-intracell gates in Kassabov's generators. However, by the folded ordering defined in Eq.~\eqref{eq:kassabov-folded-order}, the cyclic-neighbor cells $C_a$ and $C_{a+1}$ are at a physical distance of at most three cells (interacting maximum 54 qubits). Since each cell participates in only $\cO(1)$ such bounded-range interactions, Observation~\ref{observation:constant-routing} guarantees the whole layer has a 1D depth of $\cO(1)$.

\item When $r=b$, the generator applies only a single CNOT gate at the first coordinate:
\begin{equation}\label{eq:slcircuitb}
y_{\alpha,1} \mapsto y_{\alpha,1} \oplus y_{\beta,1}.
\end{equation}
Both bits lie strictly inside $C_1$, and no other cells are affected. This trivially has depth $\cO(1)$.

\item When $r=x_c$, the transformation acts as a diagonal matrix of the external control bits $x_{c,a}$. The operation is:
\begin{equation}
y_{\alpha,a} \mapsto y_{\alpha,a} \oplus (x_{c,a} \cdot y_{\beta,a}), \quad a=1,\ldots,s.
\end{equation}
This translates to a parallel layer of Toffoli gates $\mathrm{Toffoli}_{x_{c,a},y_{\beta,a}\to y_{\alpha,a}}$. The critical observation is that the two controls ($x_{c,a}, y_{\beta,a}$) and the target ($y_{\alpha,a}$) all share the same coordinate $a$. Therefore, the entire Toffoli operation is contained strictly inside the single cell $C_a$. Because the operations for different $a$ are disjoint, the full Toffoli layer has a 1D depth of $\cO(1)$.
\end{enumerate}
Consequently, every generator $y \in S_0$ possesses a 1D nearest-neighbor depth of $\cO(1)$. Recall that $P_X = \{I,X\}^{\otimes \ell}$, and the full generating set in the case where $\ell \bmod 18 = 0$ is
\begin{equation}
S = \{xyz : x,z\in P_X, \ y\in S_0\}.
\end{equation}
Every composite generator $xyz \in S$ has a total 1D depth of $1 + \cO(1) + 1 = \cO(1)$. When $\ell$ is not a multiple of $18$, we write $\ell = 18s + \ell_0$, where $0 < \ell_0 \leq 17$. If $\ell < 18$, the system size is constant, and any gate has $\cO(1)$ depth. For $s \geq 1$, we partition the $\ell$ qubits into three contiguous intervals:
\begin{equation}
A \sqcup B \sqcup C, \quad |A|=|C|=\ell_0, \quad |B|=\ell-2\ell_0.
\end{equation}
This partition defines two overlapping subsets $A \sqcup B$ and $B \sqcup C$, both of size $18s$. Let $\nu_{AB}$ be the generator distribution from the $18s$-bit construction applied exclusively to $A \sqcup B$, which acts trivially on $C$, and let $\nu_{BC}$ be the equivalent distribution applied to $B \sqcup C$. Then alternating these localized distributions yields a constant spectral gap on the global alternating group on $\ell$ qubits. This has already been proved with the permutation overlap theorem introduced in Ref.~\cite{chen2024incompressibilityspectralgapsrandom}, which is stated below.

\begin{lemma}[Theorem A.1 in Ref.~\cite{chen2024incompressibilityspectralgapsrandom}]\label{lem:overlap}
For finite sets $A$, $B$, and $C$, let $\mathbb{C}^{|A|} \otimes \mathbb{C}^{|B|} \otimes \mathbb{C}^{|C|}$ be the corresponding tensor product Hilbert space. Let $P(\pi)$ denote the permutation matrix associated with $\pi$. Then, the distance between the product of local permutations and the global permutation is bounded by:
\begin{equation}
\begin{split}
&\left\| \E_{\pi_{AB} \sim \mu(\Sym(A \times B))} P(\pi_{AB})^{\otimes k} \E_{\pi_{BC} \sim \mu(\Sym(B \times C))} P(\pi_{BC})^{\otimes k} - \E_{\pi_{ABC} \sim \mu(\Sym(A \times B \times C))} P(\pi_{ABC})^{\otimes k} \right\|_{\infty}\\
\leq &\cO\left( \frac{(k \log k + \log |B|)^3}{\sqrt{|B|}} \right),
\end{split}
\end{equation}
where $P(\pi_{AB}) \equiv P(\pi_{AB}) \otimes I_{C}$ and $P(\pi_{BC}) \equiv I_{A} \otimes P(\pi_{BC})$. Furthermore, because the $k$-th moments of the alternating group and the symmetric group match exactly for all $k \leq \min\{|A||B|-2, |B||C|-2\}$, the identical bound holds when replacing $\Sym$ with $\Alt$ in this order regime.
\end{lemma}

In our context, the dimensions of the three sets are: $|A| = |C| = 2^{\ell_0}$ and $|B| = 2^{\ell-2\ell_0}$. Substituting these dimensions into Lemma~\ref{lem:overlap}, provided with $k\leq 2^{\ell-\ell_0}-2$, the overlap error is bounded by:
\begin{equation}
\cO\left( \frac{(k \log k + \ell - 2\ell_0)^3}{\sqrt{2^{\ell-2\ell_0}}} \right).
\end{equation}
Under the assumption $k \leq \cO(2^{\ell/6.1})$, the error term is exponentially small in $\ell$. In particular, it is strictly smaller than any constant for sufficiently large $\ell$. Because this overlap error is small, we can first implement a constant number $q_0/2 = \cO(1)$ generators of $\pi_{AB}$ to achieve a small essential norm to form $\Alt_{AB}$. Then we implement a constant number $q_0/2 = \cO(1)$ generators of $\pi_{BC}$ to achieve a small essential norm to form $\Alt_{BC}$. Then combine these two parts to achieve an essential norm with an arbitrarily small constant. Consequently, a constant number $q_0 = \cO(1)$ of generators suffices. Since each sample from $\nu_{AB}$ or $\nu_{BC}$ utilizes the folded cell ordering on its respective contiguous interval, it naturally retains a 1D depth of $\cO(1)$. The product of $q_0$ such localized generators thus yields a total 1D circuit depth of $q_0 \times \cO(1) = \cO(1)$. Combining this constant-depth analysis with the spectral gap guarantees concludes the proof of the Lemma. 
\end{proof}

\section{Constant-size Constant-1D-depth generators with constant spectral gap for the Clifford group}\label{sc:clif}
In this section, we explicitly construct a set of generators for the Clifford group and prove Lemma~\ref{lem:1d-clifford}. The standard definition for the Clifford group is the group generated by Hadamard, phase, and CNOT gates, $\Cl_{\ell} = \langle H, S, \mathrm{CNOT} \rangle$. This group contains a nontrivial center (i.e. a global phase). If we do not consider any ancillary system or controlled Clifford operations, the global phase will not influence the quantum state. Meanwhile, for the purpose of generating unitary designs, where we care about the unitary representation $U^{\otimes k}\otimes \overline{U}^{\otimes k}$, the global phase of group elements cancels. In this case, the projective Clifford group $\overline{\Cl}_{\ell}$ matters, which differs the full Clifford group $\Cl_{\ell}$ by a discrete phase group~\cite{Selinger2015Clifford}:
\begin{equation}
1 \longrightarrow \langle e^{i\pi/4} I \rangle \longrightarrow \Cl_{\ell} \longrightarrow \overline{\Cl}_{\ell} \longrightarrow 1.
\end{equation}
In the following, we will first discuss a generating set for the projective Clifford group $\overline{\Cl}_{\ell}$, and then extend it to the full Clifford group $\Cl_{\ell}$.

\subsection{Generators for the projective Clifford group}
The projective Clifford group has a normal subgroup, the projective Pauli group $\overline{\mathrm{P}}_{\ell}$, with the quotient isomorphic to the binary symplectic group $\Sp(2\ell; \F_2)$:
\begin{equation}\label{eq:clifford_exact_2}
1 \longrightarrow \overline{\mathrm{P}}_{\ell} \longrightarrow \overline{\Cl}_{\ell} \longrightarrow \Sp(2\ell; \F_2) \longrightarrow 1,
\end{equation}
where $\overline{\mathrm{P}}_{\ell} \cong \F_2^{2\ell}$ is an abelian group. This structure allows us to first investigate the generators of the symplectic group $\Sp(2\ell; \F_2)$ and then lift these through the extension above.

For the symplectic group (and other finite classical groups), constant-size generating sets with bounded Kazhdan constant are known~\cite{Kassabov2006expanders,nikolov2005productdecompositionclassicalquasisimple}. Our contribution consists primarily of translating this work to quantum circuits (this requires a careful retracing of these arguments and is far from trivial).
We will follow Ref.~\cite{nikolov2005productdecompositionclassicalquasisimple} to find the generators for the symplectic group. According to Theorem~1 in Ref.~\cite{nikolov2005productdecompositionclassicalquasisimple}, each symplectic group element is the product of a constant number of elements in the special linear group and conjugates of this group inside the symplectic group. We have the following decomposition:
\begin{lemma}[Symplectic group decomposition]\label{lem:symplecticdecompose}
There is a constant $M\leq 200$ such that there exist $M$ conjugates $S_i = g_i Gg_i^{-1}$ with $G\cong \SL(\ell; \F_2)$ and $g_i\in \Sp(2\ell; \F_2)$ satisfying
\begin{equation}
\Sp(2\ell; \F_2) = \prod_{i=1}^M S_i = \{\prod_{i=1}^M s_i, s_i\in S_i\}.
\end{equation}
\end{lemma}
For the symplectic group, the upper bound on the constant $M$ can be further tightened to $130$, which we will detail later. By Lemma~\ref{lem:symplecticdecompose}, we can compose the generators for the symplectic group from the special linear group. It has been shown before in Ref.~\cite{kassabov2005universallatticesunboundedrank} that the special linear group has a bounded-size set of generators $S_{\SL, \ell}$ with an $\Omega(1)$ Kazhdan constant. Nonetheless, the exact size of the generating set for \emph{arbitrary} integer $\ell$ is not explicit. In the following lemma, we show a more explicit and quantified result, following the idea of Ref.~\cite{kassabov2005universallatticesunboundedrank}.
\begin{lemma}[Constant-size generating set for the special linear group]\label{lem:sl_generators}
For every integer $\ell\geq 6$, there exists a symmetric generating set $S_{\SL,\ell}$ for $\SL(\ell; \F_2)$ such that:
\begin{enumerate}
\item The size of the generating set is bounded, $|S_{\SL,\ell}|\leq 28$.

\item $\SL(\ell; \F_2)$ has a Kazhdan constant $\kappa_{\SL} = \Omega(1)$ with respect to $S_{\SL,\ell}$.
\end{enumerate}
\end{lemma}
\begin{proof}[Proof of Lemma~\ref{lem:sl_generators}]
We first consider the case $\ell=3s$ where $s$ is a positive integer. Recall that $\SL(3s; \F_2) = \EL(3s; \F_2)$, and we have the following surjective group homomorphism $\varphi_1$, which has appeared in the proof of Lemma~\ref{lem:1d-kassabov}:
\begin{equation}
\begin{split}
\EL(3;\mathbb{Z}\langle a,b\rangle) \xrightarrow{\ \varphi_1\ } \EL(3;\Mat(s;\F_2)) = \EL(3s;\F_2).
\end{split}
\end{equation}
Since here our target group is a single copy of $\EL(3s;\F_2)$, we only consider the action of the homomorphism $\varphi_1$ on the ring $\mathbb{Z}\langle a,b\rangle$. \footnote{As a remark, since here we do not have the control bit $x_{c,a}$, the group homomorphism $\varphi_2$ that appears in the proof of Lemma~\ref{lem:1d-kassabov} is not used.} Recall the action of $\varphi_1$ satisfies
\begin{align}
\varphi_1(1) &= I_s;\\
\varphi_1(a) &= A, \quad A e_a = e_{a+1\pmod s};\\
\varphi_1(b) &= B = E_{11}.
\end{align}
Following Lemma~\ref{lem:ELgenerator}, the group $\EL(3s;\F_2)$ has the following generating set:
\begin{equation}\label{eq:gamma3s}
\begin{split}
S_{\SL,3s} = &\big\{\varphi_1(E_{\alpha,\beta}(1)), \varphi_1(E_{\alpha',\beta'}(a)), \varphi_1(E_{\alpha',\beta'}(b))\hspace{1em}\bigm|\\ 
&\hspace{2em}\alpha\neq\beta\in \{1,2,3\},\, (\alpha',\beta')\in \{(1,2), (2,3), (2,1), (3,2)\} \big\}.
\end{split}
\end{equation}
This generating set corresponds to the set $S_{3,2}$ defined in the ``Notation" section of Ref.~\cite{kassabov2005universallatticesunboundedrank}. In Ref.~\cite{kassabov2005universallatticesunboundedrank}, the author defines $G_{d,k} := \EL(d;R_k)$ where $R_k = \mathbb{Z}\langle x_1, \cdots, x_k\rangle$ is a ring with $k$ generators. The generating set of $G_{d,k}$ is denoted as $S_{d,k}$, composed of two subsets, $S_{d,k} = F_1\cup F_2$. Here, $F_1$ is the set of $2(d^2-d)$ elementary matrices with $\pm 1$ off the diagonal, and $F_2$ is the set of $4k(d-1)$ elementary matrices, $I_d \pm x_l e_{ij}$ with $|i-j|=1$ and $ 1\leq l\leq k$. Note that over $\F_2$, $-1$ equals $1$. Thus in this case, $F_1$ contains $(d^2-d)$ elementary matrices, and $F_2$ contains $2k(d-1)$ elementary matrices. A direct comparison shows that when $d = 3, k = 2$, $F_1$ is the set $\{E_{\alpha,\beta}(1)\}$, and $F_2$ is the set $\{E_{\alpha',\beta'}(a), E_{\alpha',\beta'}(b)\}$. The whole generating set contains $14$ generators in total. It was then shown in Theorem~9 in Ref.~\cite{kassabov2005universallatticesunboundedrank} that after taking the map $\varphi_1$ (denoted as $\pi_{l,s}$ in Ref.~\cite{kassabov2005universallatticesunboundedrank}), the generating set $S_{\SL,3s}$ (denoted as $S'_{3l}$, $S'_{3l,0}$, or $\Sigma'_{3l}$ in Ref.~\cite{kassabov2005universallatticesunboundedrank}) has a Kazhdan constant lower bounded by $1/400$. We summarize the result in our notation below.
\begin{lemma}[Theorem~9 in Ref.~\cite{kassabov2005universallatticesunboundedrank}]
For all $s\geq 1$, there exists a generating set defined by $S_{\SL,3s}$ in Eq.~\eqref{eq:gamma3s} with 14 elements such that the Kazhdan constant
\begin{equation}
\cK(\SL(3s; \F_2); S_{\SL,3s}) > 1/400.
\end{equation}
\end{lemma}
Note that the original Theorem~9 is established for $\F_p$ with a generic prime $p$ instead of $\F_2$. In this case, the generating set $S_{3,2}$ contains $28$ elements. When  working over $\F_2$, $14$ elements suffice. The above argument shows that $S_{\SL,3s}$ generates $\SL(3s; \F_2)$ with a Kazhdan constant $\kappa_{\SL} = \Omega(1)$. Now we discuss the case for a generic integer $\ell$. In this case, $\SL(\ell; \F_2)$ is a product of a bounded number of copies of $\SL(3s; \F_2)$ for $s = \lfloor \ell/3 \rfloor$~\cite{kassabov2005universallatticesunboundedrank,Kassabov2006expanders}. Nonetheless, Refs.~\cite{kassabov2005universallatticesunboundedrank,Kassabov2006expanders} did not explicitly give a specific number of copies. Here, we find an explicit upper bound for this number, which is $3$, as shown in Lemma~\ref{lem:sl-overlap} below. As this proof is composed of standard linear algebra arguments (and is rather tedious), readers primarily interested in the main trajectory may safely skip it. The proof strategy closely follows Proposition 2.6 in Ref.~\cite{alavi2014triplefactorisationsgenerallinear}. Although Theorem~1.1 of Ref.~\cite{alavi2014triplefactorisationsgenerallinear} also provides a three-term decomposition of $\mathrm{SL}(\ell; \mathbb{F}_2)$, its factor subgroups are not necessarily smaller special linear groups and thus do not fit our context.

\begin{lemma}[Decomposition of the special linear group]\label{lem:sl-overlap}
For an $\ell$-dimensional space, $A\oplus B\oplus C$, where $\dim A = a$, $\dim B = b$, $\dim C = c$, and $a+b+c=\ell$, any group element $M_{ABC}$ of $\SL(\ell; \F_2)$ can be decomposed as
\begin{equation}
M_{ABC} = (P_{AB}\oplus I_C)(I_A\oplus R_{BC})(Q_{AB}\oplus I_C),
\end{equation}
provided with $b\geq 2c$. Here, $P_{AB}, Q_{AB}\in \SL(a+b; \F_2)$ are special linear group elements on system $A\oplus B$, and $R_{BC}\in \SL(b+c; \F_2)$ is a special linear group element on system $B\oplus C$.
\end{lemma}
\begin{proof}[Proof of Lemma~\ref{lem:sl-overlap}]
Denote the standard basis of matrix as $e_i$. Then $A = \Span\{e_i, 1\leq i\leq a\}$, $B = \Span\{e_i, a+1\leq i\leq a+b\}$, and $C = \Span\{e_i, a+b+1\leq i\leq a+b+c\}$. We first write $M_{ABC}$ with respect to the decomposition $A\oplus B\oplus C$,
\begin{equation}
M_{ABC} = \begin{pmatrix}
M_{AA} & M_{AB} & M_{AC} \\
M_{BA} & M_{BB} & M_{BC} \\
M_{CA} & M_{CB} & M_{CC} \\
\end{pmatrix}.
\end{equation}
Now we try to find basis transformations on $A\oplus B$, such that $M_{ABC}$ can be block diagonal on $A$ and $BC$, which is inspired by Proposition 2.6(a) in Ref.~\cite{alavi2014triplefactorisationsgenerallinear}. We summarize this proposition in our language below.
\begin{lemma}[Pairs of subspaces with trivial intersection, Proposition 2.6(a) in Ref.~\cite{alavi2014triplefactorisationsgenerallinear}]\label{lem:subspace-pair}
Let $V$ be an $n$-dimensional vector space over $\F_2$ and
$U_1,U_2,S_1,S_2\subseteq V$ be subspaces satisfying
\begin{equation}
\dim U_1=\dim U_2,\quad \dim S_1=\dim S_2,\quad \dim(U_1\cap S_1) = \dim(U_2\cap S_2).
\end{equation}
Then there exists a transformation $G\in \mathrm{GL}(V,\F_2)$ such that,
\begin{equation}
G(U_1)=U_2, \quad G(S_1)=S_2.
\end{equation}
Moreover, subspaces $U,S\subseteq V$ satisfying
\begin{equation}
\dim U=r,\quad \dim S=s,\quad \dim(U\cap S)=j
\end{equation}
exist if and only if $\max\{0,r+s-n\}\leq j \leq \min\{r,s\}$.
\end{lemma}
Here, $\mathrm{GL}(V;\F_2)$ represents the general linear group on space $V$ over $\F_2$. Note that for the field $\F_2$, $\mathrm{GL}(V;\F_2) = \SL(V;\F_2)$, which is the special linear group on space $V$ over $\F_2$. We will apply Lemma~\ref{lem:subspace-pair} to construct a basis transformation. First, we define
\begin{equation}
K = \ker
\begin{pmatrix}
M_{CA} & M_{CB}
\end{pmatrix}
\subseteq A\oplus B.
\end{equation}
Since the operator 
$\begin{pmatrix}
M_{CA} & M_{CB}
\end{pmatrix}
:A\oplus B\rightarrow C$
has rank at most $c$, we have that
\begin{equation}\label{eq:dimK}
\dim K\geq a+b-c.
\end{equation}
Since the overall transformation $M_{ABC}$ is invertible, its restriction to $K$ must be injective. For any $x\in K$, $M_{ABC}\begin{pmatrix}
x\\ 0^{\times c}
\end{pmatrix} = 0^{\times \ell}$ implies $x = 0^{\times (a+b)}$. Here, $0^{\times k}$ means a length-$k$ all-zero bit string. Hence,
\begin{equation}
\dim M_{ABC}(K)=\dim K\geq a+b-c.
\end{equation}
Moreover, by the definition of $K$, every vector in $M_{ABC}(K)$ does not have support on $C$. Therefore,
\begin{equation}
W := M_{ABC}(K)\subseteq A\oplus B.
\end{equation}
We also define
\begin{equation}
S=\operatorname{Im}
\begin{pmatrix}
M_{AC}\\
M_{BC}
\end{pmatrix}
\subseteq A\oplus B.
\end{equation}
Here, for a matrix
$X=\begin{pmatrix}x_1&\cdots&x_k\end{pmatrix}$, we use
$\operatorname{Im}(X)=\Span\{x_1,\ldots,x_k\}$ to denote its column
space. Then $S$ is generated by the $A\oplus B$ components of the
last $c$ columns of $M_{ABC}$. In particular, $\dim S\leq c$. Consequently,
\begin{equation}
\dim(W\cap S)\leq c.
\end{equation}
Under the assumption $b\geq 2c$, since we have $a+\dim(W\cap S)-\dim W\leq a+c-(a+b-c)\leq 0$, applying Lemma~\ref{lem:subspace-pair} we get that there exists a subspace $U\subseteq W$ with $\dim U = a$ such that
\begin{equation}\label{eq:USzero}
U\cap (W\cap S)=\{0\}.
\end{equation}
Since $U\subseteq W$, Eq.~\eqref{eq:USzero} further implies
\begin{equation}\label{eq:USzero2}
U\cap S=\{0\}.
\end{equation}
We now use the subspace $U$ to choose $G_{R,1}$. Since
$M_{ABC}|_K:K\rightarrow W$ is injective and hence an isomorphism onto its image, the inverse image
\begin{equation}
\widetilde{A} := \left(M_{ABC}|_K\right)^{-1}(U)\subseteq K\subseteq A\oplus B,
\end{equation}
Since $\dim\widetilde{A} = \dim U = a =\dim A$, there exists an invertible transformation $G_{R,1}$ on $A\oplus B$ such that
\begin{equation}\label{eq:GR1A}
G_{R,1}(A)=\widetilde A.
\end{equation}
which follows directly from the transitivity statement in Lemma~\ref{lem:subspace-pair}. Now we apply the basis transformation $G_{R,1}\oplus I_C$ to get
\begin{equation}
M'_{ABC} = M_{ABC}(G_{R,1}\oplus I_C) =
\begin{pmatrix}
M'_{AA} & M'_{AB} & M_{AC}\\
M'_{BA} & M'_{BB} & M_{BC}\\
M'_{CA} & M'_{CB} & M_{CC}
\end{pmatrix}.
\end{equation}
For every $x\in A$, Eq.~\eqref{eq:GR1A} implies
$G_{R,1}x\in\widetilde A\subseteq K$. Hence the $C$ component of
$M_{ABC}G_{R,1}x$ vanishes by the definition of $K$. Then $M'_{CA} = 0$. Meanwhile,
\begin{equation}
M_{ABC}G_{R,1}(A) = M_{ABC}(\widetilde A) = U.
\end{equation}
Since $U\subseteq W\subseteq A\oplus B$, we have that
\begin{equation}\label{eq:imageU}
\operatorname{Im}
\begin{pmatrix}
M'_{AA}\\
M'_{BA}
\end{pmatrix} = U,\quad \dim \operatorname{Im}
\begin{pmatrix}
M'_{AA}\\
M'_{BA}
\end{pmatrix} = a.
\end{equation}
Note that right multiplication by $G_{R,1}\oplus I_C$ does not change the last $c$ columns. Since $U\cap S = \{0\}$ as shown in Eq.~\eqref{eq:USzero2}, we have that
\begin{equation}\label{eq:Imintersection}
\operatorname{Im}
\begin{pmatrix}
M'_{AA}\\
M'_{BA}
\end{pmatrix}
\cap
\operatorname{Im}
\begin{pmatrix}
M_{AC}\\
M_{BC}
\end{pmatrix}
=
\{0\}.
\end{equation}
As a summary, after right multiplication by $G_{R,1}\oplus I_C$, we get
\begin{equation}
M'_{ABC} = M_{ABC}(G_{R,1}\oplus I_C) =
\begin{pmatrix}
M'_{AA} & M'_{AB} & M_{AC}\\
M'_{BA} & M'_{BB} & M_{BC}\\
0 & M'_{CB} & M_{CC}
\end{pmatrix},
\end{equation}
with Eq.~\eqref{eq:Imintersection}. We next construct the left basis transformation $G_{L,1}$. Let
\begin{equation}
t = \dim S = \rank
\begin{pmatrix}
M_{AC}\\
M_{BC}
\end{pmatrix}
\leq c.
\end{equation}
Since $b\geq 2c\geq t$, we can choose a $t$-dimensional subspace $B_t\subseteq B$. By
Eqs.~\eqref{eq:imageU} and~\eqref{eq:Imintersection}, the two ordered
pairs of subspaces $(U,S)$ and $(A,B_t)$ have the same respective dimensions $a$ and $t$, and both have
trivial intersection, $\dim(U\cap S) = \dim(A\cap B_t) = 0$. By Lemma~\ref{lem:subspace-pair}, there exists an invertible
transformation $\widetilde G_{L,1}$ on $A\oplus B$ such that
\begin{equation}
\widetilde G_{L,1}(U)=A,
\quad
\widetilde G_{L,1}(S)=B_t\subseteq B.
\end{equation}
Since
$\begin{pmatrix}M'_{AA}\\M'_{BA}\end{pmatrix}$ has rank $a$ and $\operatorname{Im}
\begin{pmatrix}
M'_{AA}\\
M'_{BA}
\end{pmatrix} = U$, there exists an invertible matrix $H_A$ on $A$ such
that
\begin{equation}
\widetilde G_{L,1}
\begin{pmatrix}
M'_{AA}\\
M'_{BA}
\end{pmatrix}
=
\begin{pmatrix}
H_A\\
0
\end{pmatrix}.
\end{equation}
Similarly, since $\widetilde G_{L,1}(S)\subseteq B$, there exists a
matrix $M'_{BC}$ such that
\begin{equation}
\widetilde G_{L,1}
\begin{pmatrix}
M_{AC}\\
M_{BC}
\end{pmatrix}
=
\begin{pmatrix}
0\\
M'_{BC}
\end{pmatrix}.
\end{equation}
Define
\begin{equation}
G_{L,1} = (H_A^{-1}\oplus I_B)\widetilde G_{L,1}.
\end{equation}
Then,
\begin{equation}
G_{L,1}
\begin{pmatrix}
M'_{AA}\\
M'_{BA}
\end{pmatrix}
=
\begin{pmatrix}
I_A\\
0
\end{pmatrix},
\quad
G_{L,1}
\begin{pmatrix}
M_{AC}\\
M_{BC}
\end{pmatrix}
=
\begin{pmatrix}
0\\
M'_{BC}
\end{pmatrix}.
\end{equation}
Thus, after basis transformation $G_{L,1}$ and $G_{R,1}$, we have that
\begin{equation}\label{eq:afterGL1}
(G_{L,1}\oplus I_C)
M_{ABC}
(G_{R,1}\oplus I_C)
=
\begin{pmatrix}
I_A & M''_{AB} & 0\\
0 & M''_{BB} & M'_{BC}\\
0 & M'_{CB} & M_{CC}
\end{pmatrix},
\end{equation}
with some matrices $M''_{AB}$ and $M''_{BB}$. Now the only part left is to eliminate the block $M''_{AB}$. Clearly, by column Gaussian elimination multiplying $M^{''}_{AB}$ with $I_A$ and adding to the $B$ subspace, one can eliminate $M^{''}_{AB}$. Thus, there exists an invertible matrix $G_{R,2}$ on $A\oplus B$, such that
\begin{equation}
(G_{L,1}\oplus I_C)M_{ABC}(G_{R,1}G_{R,2}\oplus I_C) = \begin{pmatrix}
I_A & 0 & 0 \\
0 & M^{''}_{BB} & M'_{BC} \\
0 & M'_{CB} & M_{CC} \\
\end{pmatrix}.
\end{equation}
Set $P_{AB} = G_{L,1}^{-1}$, $Q_{AB} = (G_{R,1}G_{R,2})^{-1}$, and $R_{BC} = \begin{pmatrix}
M^{''}_{BB} & M'_{BC} \\
M'_{CB} & M_{CC} \\
\end{pmatrix}$, we obtain that
\begin{equation}
M_{ABC} = (P_{AB}\oplus I_C)(I_A\oplus R_{BC})(Q_{AB}\oplus I_C),
\end{equation}
which completes the proof
\end{proof}
We can now finish proving the main result. We define $r = \ell \bmod 3\in \{1, 2\}$ and set $\dim A = \dim C = r$, $\dim B = \ell- 2r$. As long as $\ell\geq 4r$, we have $b\geq 2c$ and can apply Lemma~\ref{lem:sl-overlap}. Except for $\ell \in\{1, 2, 5\}$, Lemma~\ref{lem:sl-overlap} works, and we can generate $\SL(\ell; \F_2)$ with three copies of $\SL(\ell-r; \F_2)$. Among these three copies, there are two copies sharing the same set of generators. This gives $2 \times 14 = 28$ generators in total. Applying Lemma~\ref{lem:bounded-product}, we still get that the Kazhdan constant $\kappa_{\SL} = \Omega(1)$. 
\end{proof}
Note that if $\ell < 6$, then the system size is a small constant. There is no need to discuss the generators in this case, since we only consider the large-$\ell$ limit. Hence, we omit this requirement in the result.\\

Building on Lemmas~\ref{lem:bounded-product},~\ref{lem:symplecticdecompose}, and~\ref{lem:sl_generators}, finding all the required conjugation elements $g_i$ in the product decomposition leads to a generating set for the symplectic group with the desired properties. Note that though a symplectic group element is generally a product of $130$ elements of the special linear group and its conjugates, the size of the generating set can be found to be only $12$ times larger since many of the conjugated generators coincide. All the mathematical results related to these conjugation elements and how they are constructed can be found in Ref.~\cite{nikolov2005productdecompositionclassicalquasisimple}. Below, we translate these results into the language of quantum information.
\begin{lemma}[Constant-size generating set for the symplectic group]\label{lem:symplectic_generators}
For any integer $\ell \geq 1$, there exists an explicit symmetric generating set $S_{\Sp,\ell}$ for $\Sp(2\ell; \F_2)$ satisfying the following properties:
\begin{enumerate}
\item The size of the generating set is bounded by $|S_{\Sp,\ell}| \leq 12 |S_{\SL,\ell}| \leq 12\times 28 = 336 = \cO(1)$.
\item $\Sp(2\ell; \F_2)$ has a Kazhdan constant $\kappa_{\Sp} = \Omega(1)$ with respect to $S_{\Sp,\ell}$.
\end{enumerate}
\end{lemma}
\begin{proof}[Proof of Lemma~\ref{lem:symplectic_generators}]
By Lemma~\ref{lem:bounded-product} and Lemma~\ref{lem:symplecticdecompose}, the Kazhdan constant $\kappa_{\Sp} \geq \frac{\kappa_{\SL}}{130\sqrt{2}} = \Omega(1)$. Below, we show the explicit form of the generators. By Lemma~\ref{lem:symplecticdecompose}, the generating set is composed of generators in $S_{\SL, \ell}$ and their conjugation, where $S_{\SL, \ell}$ is given in Lemma~\ref{lem:sl_generators}. Note that the group $\SL(\ell; \F_2)$ and $\Sp(2\ell; \F_2)$ act on different spaces, however we can naturally embed $\SL(\ell; \F_2)$ into $\Sp(2\ell; \F_2)$ via the following mapping\footnote{This is akin to considering the Pauli action of the circuit representation of the special linear group elements}:
\begin{equation}
\label{eq:sl-embed}
    \phi(g) = 
    \begin{pmatrix}
g & \mathbf{0}\\
\mathbf{0} & (g^{-1})^T
\end{pmatrix},\;\; g\in \SL(\ell; \F_2).
\end{equation}
Correspondingly, the generating set $S_{\SL, \ell}$ is mapped to
\begin{equation}
S_{\Sp, \ell, 0} = \{\phi(g), g\in S_{\SL, \ell}\}.
\end{equation}
Note that $S_{\Sp, \ell, 0}$ is the generating set of group $G$ defined in Lemma~\ref{lem:symplecticdecompose}. The whole generating set $S_{\Sp, \ell}$ is then constructed by
\begin{equation}
S_{\Sp, \ell} = \bigcup_i g_i S_{\Sp, \ell, 0} g_i^{-1} = \bigcup_i \big\{g_i \phi(g) g_i^{-1}, g\in S_{\SL, \ell}\big\}.
\end{equation}
Below we show the explicit form of each conjugation element $g_i$, which corresponds to a specific symplectic matrix. First, based on Theorem~$D$ in Ref.~\cite{liebeck2001finitelineargroups} or Theorem~$2$ in Ref.~\cite{nikolov2005productdecompositionclassicalquasisimple}, the whole symplectic group has the following decomposition:
\begin{equation}
\Sp(2\ell; \F_2) = (U^+U^-)^6U^+,
\end{equation}
where
\begin{equation}
U^- = JU^+ J^{-1},\quad J = \begin{pmatrix}
\mathbf{0} & I_{\ell}\\
I_{\ell} & \mathbf{0}
\end{pmatrix}.
\end{equation}
Thus, we only need to focus on the conjugation elements for $U^+$. The conjugation elements for $U^-$ can be obtained by considering an additional conjugation with $J$ (the symplectic form, and itself an element of the symplectic group). Using Lemma~$1$ and Proposition~$1$ in Ref.~\cite{nikolov2005productdecompositionclassicalquasisimple}, $U^+$ has the following decomposition,
\begin{equation}
U^+ \subseteq  X_{r_0}G W \cdot\prod_{j=1}^4 (GG^{w_j}),
\end{equation}
where $G = \phi(\SL(\ell; \F_2))$. For $\Sp(2\ell; \F_2)$, by direct calculation, one can find that $w_4 = 1$, $X_{r_0} \subseteq GG^{s_1}$ and $W\subseteq G^{s_2}$ for some conjugation elements $s_1$ and $s_2$. Thus,
\begin{equation}
U^+ \subseteq\prod_{i=1}^2 (GG^{s_i}) \cdot\prod_{j=1}^3 (GG^{w_j}).
\end{equation}
The conjugation elements $s_i$ and $w_j$ are labeled by root elements in Ref.~\cite{nikolov2005productdecompositionclassicalquasisimple}. Here, we translate them to explicit symplectic matrices shown below.
\begin{align}
&s_1 = 
\begin{pmatrix}
\mathbf{0} & \mathbf{0} & 1 & \mathbf{0}\\
\mathbf{0} & I_{\ell-1} & \mathbf{0} & \mathbf{0}\\
1 & \mathbf{0} & \mathbf{0} & \mathbf{0}\\
\mathbf{0} & \mathbf{0} & \mathbf{0} & I_{\ell-1}
\end{pmatrix}, \\
&s_2 = 
\begin{pmatrix}
I_{\ell} & \mathbf{0}\\
E_{22} & I_{\ell}
\end{pmatrix}, \\
&w_1 = \prod_{i=1}^{\lfloor \ell / 2 \rfloor} \begin{pmatrix} I & \mathbf{0} \\ E_{2i-1,2i} + E_{2i,2i-1} & I \end{pmatrix}, \\
&w_2 = \prod_{i=1}^{\lfloor (\ell-1) / 2 \rfloor} \begin{pmatrix} I & \mathbf{0} \\ E_{2i,2i+1} + E_{2i+1,2i} & I \end{pmatrix}, \\
&w_3 = \prod_{i=1}^{\ell} \begin{pmatrix} I & \mathbf{0} \\ E_{i,i} & I \end{pmatrix}.
\end{align}
One can count that $U^+$ is composed of 10 elements from the conjugation group of $G = \phi(\SL(\ell; \F_2))$. Thus, $\Sp(2\ell; \F_2)$ is composed of at most 130 elements from the conjugation groups of $\phi(\SL(\ell; \F_2))$. At the level of the generating set, one can show straightforwardly that $U^+$ contributes $6$ new conjugation elements, $\{1, s_1, s_2, w_1, w_2, w_3\}$, and $U^{-}$ also contributes $6$ new elements. Thus, the size of the generating set of the symplectic group is $12$ times larger than that of the special linear group, which is upper bounded by $12\times 28 = 336$.
\end{proof}

Now we extend the symplectic group to the projective Clifford group and discuss the quantum realization of the generators. Note that we know the symplectic form of the conjugation elements. We can map them to quantum gates via their natural action on the projective Pauli group (which is just the binary symplectic vector space). Because the projective Pauli group $\overline{\mathrm{P}}_{\ell}$ is a finite abelian group, we can directly employ Hadad's abelian-extension theorem shown in Lemma~\ref{lem:hadad-abelian} to bootstrap the Kazhdan constant from the quotient to the extension group. Note that we need to first transform the Kazhdan constant to the average Kazhdan constant to use Lemma~\ref{lem:hadad-abelian}, and then do the reverse transformation. The properties of the resulting generating set are summarized in the following lemma.

\begin{lemma}[Constant-depth 1D generators for the projective Clifford group]
\label{lem:projective_generators}
Suppose $\ell \geq 1$. Let $\widetilde{S}_{\Sp,\ell} \subset \overline{\Cl}_{\ell}$ be the Clifford lift of the symplectic generating set $S_{\Sp,\ell}$, and let $X_1 \in \overline{\mathrm{P}}_{\ell}$ be the single-qubit Pauli $X$ gate acting on the first qubit. The generating set
\begin{equation}
    S_{\overline{\Cl},\ell} = \widetilde S_{\Sp,\ell} \cup \{X_1\}
\end{equation}
satisfies the following properties:
\begin{enumerate}
\item The size is bounded by $|S_{\overline{\Cl},\ell}| \leq |\widetilde S_{\Sp,\ell}| + 1 \leq 337 = \cO(1)$.
\item Each generator can be implemented as a one-dimensional, nearest-neighbor quantum circuit of depth $\cO(1)$.
\item $\overline{\Cl}_{\ell}$ has a Kazhdan constant $\kappa_{\overline{\Cl}} = \Omega(1)$ with respect to $S_{\overline{\Cl},\ell}$.
\end{enumerate}
\end{lemma}
\begin{proof}[Proof of Lemma~\ref{lem:projective_generators}]
We first discuss how to lift a binary symplectic matrix to a Clifford circuit (this is standard, we go over it in the interest of being self-contained). For any vector $x = (x_1,\cdots, x_\ell)^T, z = (z_1,\cdots, z_\ell)^T\in \F_2^{\ell}$, one can define an associated Hermitian Pauli operator. Set $a = (x, z)^T$ and define
\begin{equation}\label{eq:hermitian-Weyl}
W_a = W_{x,z} := i^{x^{T}z}X^xZ^z, \quad X^x:=\bigotimes_{i=1}^{\ell}X_i^{x_i}, Z^z:=\bigotimes_{i=1}^{\ell}Z_i^{z_i},
\end{equation}
where $X_i$ and $Z_i$ are Pauli $X$ and $Z$ operators on qubit $i$. The Pauli group $\mathrm{P}_{\ell}$ is generated by $\{W_{a}, a\in \F_2^{2\ell}\}$, and the projective Pauli group $\overline{\mathrm{P}}_{\ell} = \mathrm{P}_{\ell}/\langle iI\rangle$. The vector $(x, z)^T$ belongs to space $\F_2^{\ell}\oplus \F_2^{\ell}$ equipped with symplectic inner product. For
$a=(x,z)^{T}$ and $b=(x',z')^{T}$, $\langle a,b\rangle_{\mathrm{sp}} = x^{T}z' + z^{T}x'$, and $W_aW_b = (-1)^{\langle a,b\rangle_{\mathrm{sp}}} W_bW_a$.

After modding out the phase, multiplication of Pauli operators corresponds to addition in $\F_2^{2\ell}$. Thus, there is a group isomorphism between the projective Pauli group $\overline{\mathrm{P}}_{\ell}$ and $\F_2^{2\ell}$,
\begin{equation}
\F_2^{2\ell}\longrightarrow \overline{\mathrm{P}}_{\ell}, a\longmapsto [W_a],
\end{equation}
where $[W_a]$ is a representative element for the set $\langle iI\rangle \times W_a$.

For every Clifford gate $U\in \Cl_{\ell}$, we denote its representative element after modding out the phase as $[U]$, which belongs to the projective Clifford group $\overline{\Cl}_{\ell}$. Recall that a Clifford gate $U$ will transform a Pauli operator $W_a$ into another Pauli operator $e^{i\theta}W_b$ with some phase factor $e^{i\theta}$. After removing the phase, we denote it as $[U][W_a][U]^{-1} = [W_b]$. Under this notation, there is a quotient map $\iota: \overline{\Cl}_{\ell}\rightarrow \Sp(2\ell; \F_2)$ such that
\begin{equation}
[U][W_a][U]^{-1} = [W_{\iota([U])a}].
\end{equation}
The lifting map from $\Sp(2\ell; \F_2)$ to $\overline{\Cl}_{\ell}$ can be chosen as any map $f: \Sp(2\ell; \F_2)\rightarrow\overline{\Cl}_{\ell}$ such that for all $g\in \Sp(2\ell; \F_2)$ we have:
\begin{equation}
(\iota\circ f)(g) = g.
\end{equation}
To fix a lifting map $f$, we resort to stabilizer formalism~\cite{aaronson2004improved}. A projective Clifford operator $[U]$ can be fixed by considering its action on all single-qubit Pauli $X$ and $Z$ operators. Let $e_1,\ldots,e_{2\ell}$ be the standard basis
of $\F_2^{2\ell}$ where $e_i$ has $1$ for the $i$-th bit but $0$ for the others. Clearly, $W_{e_i} = X_i$ if $1\leq i\leq \ell$, and $W_{e_i} = Z_i$ if $\ell+1\leq i\leq 2\ell$. For each
$g\in\Sp(2\ell,\mathbb F_2)$, we define its projective Clifford lift $f(g)$ as
\begin{equation}
[f(g)] [W_{e_i}] [f(g)]^{-1} = [W_{ge_i}],\quad 1\leq i\leq 2\ell.
\end{equation}
Because $g$ is a symplectic matrix and preserves the symplectic inner product, the operators $W_{ge_i}$ maintain the pairwise commutation relations as $W_{e_i}$ and hence $[f(g)]$ is a valid projective Clifford operation. The circuit representation of $f(g)$ can be obtained by a standard Clifford synthesis algorithm~\cite{aaronson2004improved}.

After fixing $f$, we get a lift of the generating set $S_{\Sp, \ell}$ for the symplectic group $\Sp(2\ell; \F_2)$ to a set in the projective Clifford group,
\begin{equation}
\widetilde{S}_{\Sp,\ell} = f(S_{\Sp, \ell}) = \bigcup_i \{f(g_i \phi(g) g_i^{-1}), g\in S_{\SL, \ell}\} = \bigcup_i \{f(g_i) (f\circ \phi)(g) f(g_i)^{-1}, g\in S_{\SL, \ell}\}.
\end{equation}
Note that the last equality does not mean that $f$ is a group homomorphism. It just means that $f(g_i) (f\circ \phi)(g) f(g_i)^{-1})$ is a Clifford lift of $g_i \phi(g) g_i^{-1}$. This follows from the fact that $\iota$ is a group homomorphism. We have that $\iota(f(g_i) (f\circ \phi)(g) f(g_i)^{-1})) = (\iota\circ f)(g_i) (\iota\circ f)(\phi(g)) (\iota\circ f)(g_i^{-1}) = g_i \phi(g) g_i^{-1}$.

Below, we list the circuit representation $(f\circ \phi)(g)$ for $g$ from the special linear group $\SL(\ell; \F_2)$, and the circuit representation $f(g_i)$ for conjugation elements $g_i$ from the symplectic group $\Sp(2\ell; \F_2)$.

Recall that the generators of $\SL(\ell; \F_2)$ are composed of generators of $\SL(3s; \F_2)$ with $s = \lfloor \ell/3\rfloor$, where the generating set $S_{\SL, 3s}$ is given by Eq.~\eqref{eq:gamma3s}. Now we define the embedding map $\pi$ from $\SL(3s; \F_2)$ to $\Alt(2^{3s}-1)$. For $g\in S_{\SL, 3s}$ and a bit string $z\in \F_2^{3s}$, we have that
\begin{equation}
\pi(g) z = gz.
\end{equation}
This embedding map is actually the one defined in Eq.~\eqref{eq:alternating_union}. Nonetheless, there is only one register and no control bit in this case. One can find that for all $ g\in S_{\SL, 3s}$,
\begin{equation}\label{eq:interchangingmap}
\pi(g) = (f\circ \phi)(g).
\end{equation}
This is best illustrated by an example.  If we take $g$ as $\varphi_1(E_{\alpha,\beta}(1)) = E_{\alpha,\beta}(I_s)$. Then for an input bit string $(y_1, y_2, y_3)^T$, where $y_{\alpha} = (y_{\alpha,1}, y_{\alpha,2}, \cdots, y_{\alpha,s})^T$, we have that $\pi(g)$ has the following map (as shown in the proof of Lemma~\ref{lem:1d-kassabov}):
\begin{equation}
y_{\alpha,a} \mapsto y_{\alpha,a} \oplus y_{\beta,a}, \quad a=1,\ldots,s,
\end{equation}
which is a parallel layer of CNOT gates $\mathrm{CNOT}_{y_{\beta,a}\to y_{\alpha,a}}$. On the other hand,
\begin{equation}
\phi(g) = \begin{pmatrix}
E_{\alpha,\beta}(I_s) & \mathbf{0}\\
\mathbf{0} & E_{\beta, \alpha}(I_s)
\end{pmatrix}.
\end{equation}
Thus, $\phi(g)$ will map the Pauli string $(x,z)$ by updating $x_{\alpha, a} \mapsto x_{\alpha, a} + x_{\beta, a}$ and $z_{\beta, a} \mapsto z_{\beta, a} + z_{\alpha, a}$ for all $1\leq a\leq s$. We use $(\alpha, a)$ to label the index of qubits. This exactly corresponds to the parallel layer of CNOT gates $\mathrm{CNOT}_{y_{\beta,a}\to y_{\alpha,a}}$ we noted earlier. \\

Thanks to Eq.~\eqref{eq:interchangingmap}, the quantum circuit lift for $S_{\SL, 3s}$ is
\begin{equation}
\begin{split}
\pi(S_{\SL, 3s}) = &\{\pi(g), g\in S_{\SL, 3s}\} = \bigg\{(\pi\circ\varphi_1)(E_{\alpha,\beta}(1)), (\pi\circ\varphi_1)(E_{\alpha',\beta'}(a)), (\pi\circ\varphi_1)(E_{\alpha',\beta'}(b)),\\ 
&\hspace{10.5em}\alpha\neq\beta\in \{1,2,3\}, (\alpha',\beta')\in \{(1,2), (2,3), (2,1), (3,2)\} \bigg\}.
\end{split}
\end{equation}
The circuit realization of all elements above have already been discussed in the proof of Lemma~\ref{lem:1d-kassabov}, which are listed in Eqs.~\eqref{eq:slcircuit1},~\eqref{eq:slcircuita}, and~\eqref{eq:slcircuitb}. All these circuits are realizable in a constant depth in 1D. Finally we must give quantum circuit  realizations of $f(g_i)$ for all conjugation elements $g_i\in \{J, s_1, s_2, w_1, w_2, w_3\}$ that fill out the rest of the symplectic group. The associated circuit is found by computing the action of $g_i$ on Pauli strings directly. We have
\begin{align}
&J = \begin{pmatrix}
\mathbf{0} & I_{\ell}\\
I_{\ell} & \mathbf{0}
\end{pmatrix}, \tag{Parallel Hadamard gate, $f(J) = H^{\otimes \ell}$} \\
&s_1 = 
\begin{pmatrix}
\mathbf{0} & \mathbf{0} & 1 & \mathbf{0}\\
\mathbf{0} & I_{\ell-1} & \mathbf{0} & \mathbf{0}\\
1 & \mathbf{0} & \mathbf{0} & \mathbf{0}\\
\mathbf{0} & \mathbf{0} & \mathbf{0} & I_{\ell-1}
\end{pmatrix}, \tag{Hadamard gate on the first qubit, $f(s_1) = H_1$} \\
&s_2 = 
\begin{pmatrix}
I_{\ell} & \mathbf{0}\\
E_{22} & I_{\ell}
\end{pmatrix}, \tag{Phase gate on the second qubit, $f(s_2) = S_2$} \\
&w_1 = \prod_{i=1}^{\lfloor \ell / 2 \rfloor} \begin{pmatrix} I & \mathbf{0} \\ E_{2i-1,2i} + E_{2i,2i-1} & I \end{pmatrix}, \tag{Parallel CZ gate, $f(w_1) = \bigotimes_{i=1}^{\lfloor \ell / 2 \rfloor}\mathrm{CZ}_{2i-1,2i}$} \\
&w_2 = \prod_{i=1}^{\lfloor (\ell-1) / 2 \rfloor} \begin{pmatrix} I & \mathbf{0} \\ E_{2i,2i+1} + E_{2i+1,2i} & I \end{pmatrix}, \tag{Parallel CZ gate, $f(w_2) = \bigotimes_{i=1}^{\lfloor (\ell-1) / 2 \rfloor}\mathrm{CZ}_{2i,2i+1}$} \\
&w_3 = \prod_{i=1}^{\ell} \begin{pmatrix} I & \mathbf{0} \\ E_{i,i} & I \end{pmatrix}. \tag{Parallel phase gate, $f(w_3) = \bigotimes_{i=1}^{\ell} S_i$}
\end{align}
It is clear that all these conjugation elements can be realized in depth $1$ on a line.

Next, we apply Lemma~\ref{lem:hadad-abelian} to extend the Clifford lift of the symplectic group to the projective Clifford group. We identify the sequence components in Lemma~\ref{lem:hadad-abelian} as $\Gamma = \overline{\Cl}_{\ell}$, the quotient $H = \Sp(2\ell; \F_2)$, and the abelian kernel $A = \overline{\mathrm{P}}_{\ell}$. The quotient generating set is $S = S_{\Sp,\ell}$. The set $B$ is chosen as the Pauli $X$ operator (without phase) on the first qubit, $B = \{X_1\} \subseteq \overline{\mathrm{P}}_{\ell}$, where $|B| = 1$. Accordingly, its full $H$-orbit, $\widetilde B = \{X_1^h : h \in \Sp(2\ell; \F_2)\} = \overline{\mathrm{P}}_{\ell} \setminus \{I\}$.

We now explicitly calculate the average Kazhdan constant $\cK_{\av,A}$ of the abelian kernel $\overline{\mathrm{P}}_{\ell}$ with respect to $\widetilde B$. Every irreducible representation of $\overline{\mathrm{P}}_{\ell}$ is a character of the form $\chi_y(a) = (-1)^{\langle y,a\rangle_{\mathrm{sp}}}$ for some vector $y \in \overline{\mathrm{P}}_{\ell}$. For any non-trivial representation ($y \neq 0$), the symplectic inner product $\langle y, a \rangle_{\mathrm{sp}}$ evaluates to $1$ for exactly half of the elements in $\overline{\mathrm{P}}_{\ell}$, and $0$ for the other half. Among the $|\widetilde B| = 2^{2\ell} - 1$ non-zero elements, exactly $2^{2\ell-1}$ elements yield $\chi_y(a) = -1$. By definition, the average Kazhdan constant is calculated as
\begin{equation}
    \frac{1}{|\widetilde B|}\sum_{a \in \widetilde B}|\chi_y(a)-1|^2 = \frac{1}{2^{2\ell}-1} \sum_{\langle y,a \rangle_{\mathrm{sp}} = 1} |-2|^2 = \frac{4 \cdot 2^{2\ell-1}}{2^{2\ell}-1} > \frac{2^{2\ell+1}}{2^{2\ell}} = 2.
\end{equation}
This demonstrates that the abelian kernel $\overline{\mathrm{P}}_{\ell}$ possesses an average Kazhdan constant $\cK_{\av,A} \geq 2$ with respect to $\widetilde B$. From Lemma~\ref{lem:symplectic_generators}, the symplectic group has a Kazhdan constant $\kappa_{\Sp} = \Omega(1)$ with a bounded set size $|\widetilde{S}_{\Sp,\ell}| = |S_{\Sp, \ell}| = \cO(1)$. Thus, the average Kazhdan constant $\cK_{\av,H}\geq \kappa_{\Sp}^2/|S_{\Sp, \ell}| = \Omega(1)$. Applying the bound from Lemma~\ref{lem:hadad-abelian}, the projective Clifford group $\overline{\Cl}_{\ell}$ achieves an average Kazhdan constant $\cK_{\av,\overline{\Cl}}$ with respect to the combined generating set $S_{\overline{\Cl},\ell} = \widetilde{S}_{\Sp,\ell} \cup \{X_1\}$ bounded by
\begin{equation}
\cK_{\av,\overline{\Cl}} \geq \frac{\cK_{\av,H} \cdot 2}{512\left(1+|\widetilde{S}_{\Sp,\ell}|/1+1/|\widetilde{S}_{\Sp,\ell}|\right)}.
\end{equation}
Because $\cK_{\av,H}$ and $|\widetilde{S}_{\Sp,\ell}|$ are positive constants, the resulting average Kazhdan constant $\cK_{\av,\overline{\Cl}}$ is a constant. The Kazhdan constant $\kappa_{\overline{\Cl}}\geq \sqrt{\cK_{\av,\overline{\Cl}}}$ is also a positive constant.

Finally, we evaluate the physical circuit realization of the generators. We previously proved that all operations in the symplectic lift $\widetilde S_{\Sp,\ell}$ possess a 1D nearest-neighbor circuit depth of $\cO(1)$. The newly adjoined element $X_1$ is a single-qubit Pauli gate natively executed in depth $1$. Therefore, every generator in the union $S_{\overline{\Cl},\ell}$ possesses a one-dimensional circuit depth of $\cO(1)$, which completes the proof.
\end{proof}

It is important to note that the generating set $S_{\overline{\Cl},\ell}$ constructed in the proof of Lemma~\ref{lem:projective_generators} is not symmetric. There are two non-involutory elements composed of phase gates: $S_2$ and $\bigotimes_{i=1}^{\ell} S_i$. To make the generating set symmetric, we also need to add conjugate elements $S_2^{-1}$ and $\bigotimes_{i=1}^{\ell} S_i^{-1}$ as well as their multiplication with $J$ and special linear group generators. This adds an additional $2\times 2\times \abs{S_{\SL,\ell}} = 112$ generators. Thus, one can construct a symmetric generating set $S'_{\overline{\Cl}, \ell}$ for the projective Clifford group $\overline{\Cl}_{\ell}$ with size bounded by $337+112 = 449 = \cO(1)$. The extra inverse elements $S_2^{-1}$ and $\bigotimes_{i=1}^{\ell} S_i^{-1}$ are also implementable at a constant depth in 1D. Note that by Definition~\ref{def:kazhdanconstant}, the expansion of the generating set does not reduce Kazhdan constant, so with respect to this symmetric generating set the Kazhdan constant still maintains $\Omega(1)$.

\subsection{Lifting the projective Clifford group to the full Clifford group}
Recall that the projective Clifford group and the full Clifford group only differ by a phase. Denote finite phase group $Z_{\ph} = \langle e^{i\pi/4} I \rangle \cong \mathbb{Z}_8$, then:
\begin{equation}
\label{eq:central-clifford-extension}
    1 \longrightarrow Z_{\ph} \longrightarrow \Cl_{\ell} \longrightarrow \overline{\Cl}_{\ell} \longrightarrow 1.
\end{equation}
By introducing two new generators $\{e^{i\pi/4} I, e^{-i\pi/4}I\}$ to $S'_{\overline{\Cl},\ell}$, we can get the generators for the full Clifford group. The Kazhdan constant can be proved to be $\Omega(1)$ by direct calculation, as shown in the following lemma.
\begin{lemma}[Constant-depth 1D generators for the full Clifford group]
\label{lem:fullClifford_generators}
For every integer $\ell\geq 1$, there exists a symmetric generating set $S_{\Cl,\ell}$ for the full Clifford group satisfying the following properties:
\begin{enumerate}
    \item The size is bounded by $|S_{\Cl,\ell}|\leq |S'_{\overline{\Cl},\ell}| + 2 \leq 451 = \cO(1)$.
    \item Each generator can be implemented as a one-dimensional, nearest-neighbor quantum circuit of depth $\cO(1)$.
    \item $\Cl_{\ell}$ has a Kazhdan constant $\kappa_{\Cl} = \Omega(1)$ with respect to $S_{\Cl,\ell}$.
\end{enumerate}
\end{lemma}
\begin{proof}[Proof of Lemma~\ref{lem:fullClifford_generators}]
We already have the symmetric generating set $S = S'_{\overline{\Cl},\ell}$ for the projective Clifford group, which possesses an average Kazhdan constant $\cK_{\av,\overline{\Cl}} = \Omega(1)$. We define $\widehat{S}_{\overline{\Cl},\ell}$ as a lift of this set into $\Cl_{\ell}$. To maintain the symmetric condition of the generating set, we define $\Omega = \{e^{i\pi/4} I, e^{-i\pi/4}I\}$, and the generating set for the full Clifford group is constructed as
\begin{equation}
\label{eq:full-clifford-generators}
S_{\Cl,\ell} = \widehat S_{\overline{\Cl},\ell} \cup \Omega.
\end{equation}
Obviously, each generator in $S_{\Cl,\ell}$ has a one-dimensional depth $\cO(1)$ and $\abs{S_{\Cl,\ell}}\leq 449 + 2 = 451 = \cO(1)$. It remains to lower bound the average Kazhdan constant for all non-trivial irreducible representations. Let
\begin{equation}
\rho:\Cl_{\ell}\longrightarrow\U(\mathcal{H}_{\rho})
\end{equation}
be a non-trivial irreducible representation and let $v\in\mathcal H_{\rho}$ be a unit vector. Since $Z_{\ph}$ is contained in the center of $\Cl_{\ell}$, Schur's lemma implies that the center acts on $\mathcal{H}_{\rho}$ as a number. In particular, there exists $m\in\{0,\ldots,7\}$ such that
\begin{equation}
\rho(e^{i\pi/4} I) = e^{im\pi/4} I_{\mathcal{H}_{\rho}}.
\end{equation}
We distinguish two cases. First, suppose $m=0$. Then $\rho$ is trivial on $Z_{\ph}$ and hence factors through the quotient:
\begin{equation}
\rho=\overline{\rho}\circ q
\end{equation}
for a non-trivial irreducible representation $\overline{\rho}$ of $\overline{\Cl}_{\ell}$. Here $q: \Cl_{\ell}\rightarrow \overline{\Cl}_{\ell}$ is the quotient map mapping from the full Clifford group to the projective Clifford group. Then by definition of the average Kazhdan constant we have
\begin{equation}
\begin{split}
\cK_{\av,\Cl} &= \inf_{v}\frac{1}{|S_{\Cl,\ell}|} \sum_{g\in S_{\Cl,\ell}} \Vert\rho(g)v-v\Vert^2\\
&\geq \inf_{v}\frac{1}{|S_{\Cl,\ell}|} \sum_{\hat{s}\in \widehat{S}_{\overline{\Cl},\ell}} \Vert\rho(\hat{s})v-v\Vert^2\\
&= \frac{|S'_{\overline{\Cl},\ell}|}{|S_{\Cl,\ell}|}\inf_{v}\frac{1}{|S'_{\overline{\Cl},\ell}|} \sum_{s\in S'_{\overline{\Cl},\ell}} \Vert\overline{\rho}(s)v-v\Vert^2\\
&\geq \frac{|S'_{\overline{\Cl},\ell}|}{|S_{\Cl,\ell}|}\cK_{\av,\overline{\Cl}}.
\end{split}
\end{equation}
Thus, in this case the average Kazhdan constant is $\Omega(1)$.\\

On the other hand, when $m\neq 0$, the two central generators $e^{\pm i\pi/4}I$ alone already provide a constant Kazhdan constant:
\begin{equation}
\begin{split}
\cK_{\av,\Cl}&=\inf_{v}\frac{1}{|S_{\Cl,\ell}|}\sum_{g\in S_{\Cl,\ell}}\Vert\rho(g)v-v\Vert^2\\
&\geq\inf_{v}\frac{1}{|S_{\Cl,\ell}|}\left(\Vert\rho(e^{i\pi/4} I)v-v\Vert^2+\Vert\rho(e^{-i\pi/4}I)v-v\Vert^2\right)\\
&=\frac{2|e^{im\pi/4}-1|^2}{|S_{\Cl,\ell}|}\\
&\geq\frac{2|e^{i\pi/4}-1|^2}{451}=\Omega(1).
\end{split}
\end{equation}
Combining all the cases, the average Kazhdan constant of $\Cl_{\ell}$ is a positive constant $\cK_{\av,\Cl}$. Using the inequality $\kappa_{\Cl}\geq \sqrt{\cK_{\av,\Cl}} = \Omega(1)$, we get that the Kazhdan constant of $\Cl_{\ell}$ is a positive constant.
\end{proof}

Applying Lemma~\ref{lem:fullClifford_generators} and using the Kazhdan-to-spectral-gap inequality shown in Lemma~\ref{lem:kazhdan-constant-to-spectral-gap}, the spectral gap of the generating set $S_{\Cl,\ell}$ for $\Cl_{\ell}$ is a positive constant. For any representation $\rho$ containing no trivial subrepresentation, the essential norm satisfies:
\begin{equation}
g(\mu(S_{\Cl,\ell}), \rho, \Cl_{\ell})\leq 1 - \frac{\kappa_{\Cl}^2}{2|S_{\Cl,\ell}|} = 1 - \Omega(1).
\end{equation}
This proves Lemma~\ref{lem:1d-clifford}.

\bibliographystyle{unsrtnat}
\bibliography{bibdesign}

@misc{bittel2025completetheorycliffordcommutant,
      title={A complete theory of the Clifford commutant}, 
      author={Lennart Bittel and Jens Eisert and Lorenzo Leone and Antonio A. Mele and Salvatore F. E. Oliviero},
      year={2025},
      eprint={2504.12263},
      archivePrefix={arXiv},
      primaryClass={quant-ph},
      url={https://arxiv.org/abs/2504.12263}, 
}

@article{haferkamp2022random,
  title={Random quantum circuits are approximate unitary $t$-designs in depth $O(nt^{5+o(1)})$},
  author={Haferkamp, Jonas},
  journal={Quantum},
  volume={6},
  pages={795},
  year={2022},
  publisher={Verein zur F{\"o}rderung des Open Access Publizierens in den Quantenwissenschaften}
}

@article{helsen2022general,
  title={General framework for randomized benchmarking},
  author={Helsen, Jonas and Roth, Ingo and Onorati, Emilio and Werner, Albert H and Eisert, Jens},
  journal={PRX quantum},
  volume={3},
  number={2},
  pages={020357},
  year={2022},
  publisher={APS}
}

@article{anshu2026depth,
  title={Depth-1 quantum (tensor product) expanders and applications},
  author={Anurag Anshu and Shankar Balasubramanian and Jonas Haferkamp and Aram W. Harrow and
Xinyu Tan},
  journal={In preparation},
  year={2026},
}

@article{harrow2023permutation,
  title={Approximate orthogonality of permutation operators, with application to quantum information},
  author={Harrow, Aram W},
  journal={Letters in Mathematical Physics},
  volume={114},
  number={1},
  pages={1},
  year={2023},
  publisher={Springer}
}

@inproceedings{metger2024simple,
  title={Simple constructions of linear-depth t-designs and pseudorandom unitaries},
  author={Metger, Tony and Poremba, Alexander and Sinha, Makrand and Yuen, Henry},
  booktitle={IEEE 65th Annual Symposium on Foundations of Computer Science (FOCS)},
  pages={485--492},
  year={2024},
  organization={IEEE}
}

@article{harrow2023approximate,
  title={Approximate unitary t-designs by short random quantum circuits using nearest-neighbor and long-range gates},
  author={Harrow, Aram W and Mehraban, Saeed},
  journal={Communications in Mathematical Physics},
  volume={401},
  number={2},
  pages={1531--1626},
  year={2023},
  publisher={Springer}
}

@article{haah2025efficient,
  title={Efficient approximate unitary designs from random pauli rotations},
  author={Haah, Jeongwan and Liu, Yunchao and Tan, Xinyu},
  journal={Communications in Mathematical Physics},
  volume={406},
  number={12},
  pages={1--24},
  year={2025},
  publisher={Springer}
}

@article{laracuente2026approximate,
  title={Approximate unitary k-designs from shallow, low-communication circuits},
  author={LaRacuente, Nicholas and Leditzky, Felix},
  journal={Communications in Mathematical Physics},
  volume={407},
  number={3},
  pages={51},
  year={2026},
  publisher={Springer}
}

@article{haferkamp2023efficient,
  title={Efficient Unitary Designs with a System-Size Independent Number of Non-Clifford Gates: J. Haferkamp, F. Montealegre-Mora, M. Heinrich, J. Eisert, D. Gross, I. Roth},
  author={Haferkamp, Jonas and Montealegre-Mora, Felipe and Heinrich, Markus and Eisert, Jens and Gross, David and Roth, Ingo},
  journal={Communications in Mathematical Physics},
  volume={397},
  number={3},
  pages={995--1041},
  year={2023},
  publisher={Springer}
}

@article{hunter2019unitary,
  title={Unitary designs from statistical mechanics in random quantum circuits},
  author={Hunter-Jones, Nicholas},
  journal={arXiv preprint arXiv:1905.12053},
  year={2019}
}

@article{Bittel2026operational,
  doi = {10.22331/q-2026-04-15-2069},
  url = {https://doi.org/10.22331/q-2026-04-15-2069},
  title = {Operational interpretation of the {S}tabilizer {E}ntropy},
  author = {Bittel, Lennart and Leone, Lorenzo},
  journal = {{Quantum}},
  issn = {2521-327X},
  publisher = {{Verein zur F{\"{o}}rderung des Open Access Publizierens in den Quantenwissenschaften}},
  volume = {10},
  pages = {2069},
  month = apr,
  year = {2026}
}

@article{brandao2016local,
  title={Local random quantum circuits are approximate polynomial-designs},
  author={Brandao, Fernando GSL and Harrow, Aram W and Horodecki, Micha{\l}},
  journal={Communications in Mathematical Physics},
  volume={346},
  number={2},
  pages={397--434},
  year={2016},
  publisher={Springer}
}

@misc{chen2024incompressibilityspectralgapsrandom,
      title={Incompressibility and spectral gaps of random circuits}, 
      author={Chi-Fang Chen and Jeongwan Haah and Jonas Haferkamp and Yunchao Liu and Tony Metger and Xinyu Tan},
      year={2024},
      eprint={2406.07478},
      archivePrefix={arXiv},
      primaryClass={quant-ph},
      url={https://arxiv.org/abs/2406.07478}, 
}

@article{hadad2010kazhdanconstantsgroupextensions,
author = {HADAD, UZY},
title = {KAZHDAN CONSTANTS OF GROUP EXTENSIONS},
journal = {International Journal of Algebra and Computation},
volume = {20},
number = {05},
pages = {671-688},
year = {2010},
doi = {10.1142/S0218196710005832},
URL = {https://doi.org/10.1142/S0218196710005832
},
eprint = {https://doi.org/10.1142/S0218196710005832
}
}

@misc{nikolov2005productdecompositionclassicalquasisimple,
      title={A product decomposition for the classical quasisimple groups}, 
      author={Nikolay Nikolov},
      year={2005},
      eprint={math/0510173},
      archivePrefix={arXiv},
      primaryClass={math.GR},
      url={https://arxiv.org/abs/math/0510173}, 
}

@article{Selinger2015Clifford,
   title={Generators and relations for n-qubit Clifford operators},
   volume={Volume 11, Issue 2},
   ISSN={1860-5974},
   url={http://dx.doi.org/10.2168/LMCS-11(2:10)2015},
   DOI={10.2168/lmcs-11(2:10)2015},
   journal={Logical Methods in Computer Science},
   publisher={Centre pour la Communication Scientifique Directe (CCSD)},
   author={Selinger, Peter},
   year={2015},
   month={June} }

@article{Kassabov2007symmetric,
	author = {Kassabov, Martin},
	date = {2007/11/01},
	doi = {10.1007/s00222-007-0065-y},
	id = {Kassabov2007},
	isbn = {1432-1297},
	journal = {Inventiones mathematicae},
	number = {2},
	pages = {327--354},
	title = {Symmetric groups and expander graphs},
	url = {https://doi.org/10.1007/s00222-007-0065-y},
	volume = {170},
	year = {2007}
    }

@article{liebeck2001finitelineargroups,
 ISSN = {0003486X},
 URL = {http://www.jstor.org/stable/3062101},
 author = {Martin W. Liebeck and Aner Shalev},
 journal = {Annals of Mathematics},
 number = {2},
 pages = {383--406},
 publisher = {Annals of Mathematics},
 title = {Diameters of Finite Simple Groups: Sharp Bounds and Applications},
 urldate = {2026-06-28},
 volume = {154},
 year = {2001}
}

@article{Schuster2025Gluing,
author = {Thomas Schuster  and Jonas Haferkamp  and Hsin-Yuan Huang },
title = {Random unitaries in extremely low depth},
journal = {Science},
volume = {389},
number = {6755},
pages = {92-96},
year = {2025},
doi = {10.1126/science.adv8590},
URL = {https://www.science.org/doi/abs/10.1126/science.adv8590},
eprint = {https://www.science.org/doi/pdf/10.1126/science.adv8590}
}

@article{Kassabov2006expanders,
author = {Martin Kassabov  and Alexander Lubotzky  and Nikolay Nikolov },
title = {Finite simple groups as expanders},
journal = {Proceedings of the National Academy of Sciences},
volume = {103},
number = {16},
pages = {6116-6119},
year = {2006},
doi = {10.1073/pnas.0510337103},
URL = {https://www.pnas.org/doi/abs/10.1073/pnas.0510337103},
eprint = {https://www.pnas.org/doi/pdf/10.1073/pnas.0510337103}
}

@misc{kassabov2005universallatticesunboundedrank,
      title={Universal lattices and unbounded rank expanders}, 
      author={Martin Kassabov},
      year={2005},
      eprint={math/0502237},
      archivePrefix={arXiv},
      primaryClass={math.GR},
      url={https://arxiv.org/abs/math/0502237}, 
}

@article{Gross2021Clifford,
	author = {Gross, David and Nezami, Sepehr and Walter, Michael},
	date = {2021/08/01},
	doi = {10.1007/s00220-021-04118-7},
	id = {Gross2021},
	isbn = {1432-0916},
	journal = {Communications in Mathematical Physics},
	number = {3},
	pages = {1325--1393},
	title = {Schur--Weyl Duality for the Clifford Group with Applications: Property Testing, a Robust Hudson Theorem, and de Finetti Representations},
	url = {https://doi.org/10.1007/s00220-021-04118-7},
	volume = {385},
	year = {2021}
    }

@misc{cui2025unitarydesignsnearlyoptimal,
      title={Unitary designs in nearly optimal depth}, 
      author={Laura Cui and Thomas Schuster and Fernando Brandao and Hsin-Yuan Huang},
      year={2025},
      eprint={2507.06216},
      archivePrefix={arXiv},
      primaryClass={quant-ph},
      url={https://arxiv.org/abs/2507.06216}, 
}

@misc{baer2026randomunitarycircuitsconstant,
      title={Random unitary circuits with constant spectral gap}, 
      author={Tim Baer and Jeongwan Haah},
      year={2026},
      eprint={2607.20919},
      archivePrefix={arXiv},
      primaryClass={quant-ph},
      url={https://arxiv.org/abs/2607.20919}, 
}

@article{Zhang2026magic,
  title = {Designs from Magic-Augmented Clifford Circuits},
  author = {Zhang, Yuzhen and Vijay, Sagar and Gu, Yingfei and Bao, Yimu},
  journal = {PRX Quantum},
  volume = {7},
  issue = {1},
  pages = {010344},
  numpages = {37},
  year = {2026},
  month = {Mar},
  publisher = {American Physical Society},
  doi = {10.1103/myrb-nyhf},
  url = {https://link.aps.org/doi/10.1103/myrb-nyhf}
}

@article{Leone2026NonClifford,
   title={Non-Clifford Cost of Random Unitaries},
   volume={7},
   ISSN={2691-3399},
   url={http://dx.doi.org/10.1103/25v1-my1x},
   DOI={10.1103/25v1-my1x},
   number={2},
   journal={PRX Quantum},
   publisher={American Physical Society (APS)},
   author={Leone, Lorenzo and Oliviero, Salvatore F.E. and Hamma, Alioscia and Eisert, Jens and Bittel, Lennart},
   year={2026},
   month={May}}

@article{Haferkamp2022randomquantum,
  doi = {10.22331/q-2022-09-08-795},
  url = {https://doi.org/10.22331/q-2022-09-08-795},
  title = {Random quantum circuits are approximate unitary {$t$}-designs in depth {$O\left(nt^{5+o(1)}\right)$}},
  author = {Haferkamp, Jonas},
  journal = {{Quantum}},
  issn = {2521-327X},
  publisher = {{Verein zur F{\"{o}}rderung des Open Access Publizierens in den Quantenwissenschaften}},
  volume = {6},
  pages = {795},
  month = sep,
  year = {2022}
}

@article{Haferkamp2023Designs,
	author = {Haferkamp, J. and Montealegre-Mora, F. and Heinrich, M. and Eisert, J. and Gross, D. and Roth, I.},
	date = {2023/02/01},
	doi = {10.1007/s00220-022-04507-6},
	id = {Haferkamp2023},
	isbn = {1432-0916},
	journal = {Communications in Mathematical Physics},
	number = {3},
	pages = {995--1041},
	title = {Efficient Unitary Designs with a System-Size Independent Number of Non-Clifford Gates},
	url = {https://doi.org/10.1007/s00220-022-04507-6},
	volume = {397},
	year = {2023}
    }

@article{Patrick2007Blackhole,
doi = {10.1088/1126-6708/2007/09/120},
url = {https://dx.doi.org/10.1088/1126-6708/2007/09/120},
year = {2007},
month = {sep},
publisher = {},
volume = {2007},
number = {09},
pages = {120},
author = {Patrick Hayden and  John Preskill},
title = {Black holes as mirrors: quantum information in random subsystems},
journal = {Journal of High Energy Physics}
}

@article{Almheiri2015QEC,
author = {Almheiri, Ahmed and Dong, Xi and Harlow, Daniel},
date = {2015/04/29},
doi = {10.1007/JHEP04(2015)163},
id = {Almheiri2015},
isbn = {1029-8479},
journal = {Journal of High Energy Physics},
number = {4},
pages = {163},
title = {Bulk locality and quantum error correction in AdS/CFT},
url = {https://doi.org/10.1007/JHEP04(2015)163},
volume = {2015},
year = {2015}
}

@article{Pastawski2015Holographic,
author = {Pastawski, Fernando and Yoshida, Beni and Harlow, Daniel and Preskill, John},
date = {2015/06/23},
doi = {10.1007/JHEP06(2015)149},
id = {Pastawski2015},
isbn = {1029-8479},
journal = {Journal of High Energy Physics},
number = {6},
pages = {149},
title = {Holographic quantum error-correcting codes: toy models for the bulk/boundary correspondence},
url = {https://doi.org/10.1007/JHEP06(2015)149},
volume = {2015},
year = {2015}}

@article{Zhou2026Metrology,
  title = {Randomized Measurements for Multiparameter Quantum Metrology},
  author = {Zhou, Sisi and Chen, Senrui},
  journal = {PRX Quantum},
  volume = {7},
  issue = {1},
  pages = {010314},
  numpages = {34},
  year = {2026},
  month = {Jan},
  publisher = {American Physical Society},
  doi = {10.1103/s27y-gbrp},
  url = {https://link.aps.org/doi/10.1103/s27y-gbrp}
}

@Article{Elben2023toolbox,
author={Elben, Andreas
and Flammia, Steven T.
and Huang, Hsin-Yuan
and Kueng, Richard
and Preskill, John
and Vermersch, Beno{\^i}t
and Zoller, Peter},
title={The randomized measurement toolbox},
journal={Nature Reviews Physics},
year={2023},
month={Jan},
day={01},
volume={5},
number={1},
pages={9-24},
issn={2522-5820},
url={https://doi.org/10.1038/s42254-022-00535-2}
}

@Article{huang2020shadow,
author={Huang, Hsin-Yuan
and Kueng, Richard
and Preskill, John},
title={Predicting many properties of a quantum system from very few measurements},
journal={Nat. Phys.},
year={2020},
month={Oct},
day={01},
volume={16},
number={10},
pages={1050-1057},
issn={1745-2481},
doi={10.1038/s41567-020-0932-7},
url={https://doi.org/10.1038/s41567-020-0932-7}
}

@article{Nahum2017Random,
title = {Quantum Entanglement Growth under Random Unitary Dynamics},
author = {Nahum, Adam and Ruhman, Jonathan and Vijay, Sagar and Haah, Jeongwan},
journal = {Phys. Rev. X},
volume = {7},
issue = {3},
pages = {031016},
numpages = {30},
year = {2017},
month = {Jul},
publisher = {American Physical Society},
doi = {10.1103/PhysRevX.7.031016},
url = {https://link.aps.org/doi/10.1103/PhysRevX.7.031016}
}

@article{Nahum2018randomness,
title = {Dynamics of entanglement and transport in one-dimensional systems with quenched randomness},
author = {Nahum, Adam and Ruhman, Jonathan and Huse, David A.},
journal = {Phys. Rev. B},
volume = {98},
issue = {3},
pages = {035118},
numpages = {16},
year = {2018},
month = {Jul},
publisher = {American Physical Society},
doi = {10.1103/PhysRevB.98.035118},
url = {https://link.aps.org/doi/10.1103/PhysRevB.98.035118}
}

@Article{arute2019supremacy,
author={Arute, Frank
and Arya, Kunal
and Babbush, Ryan
and Bacon, Dave
and Bardin, Joseph C.
and Barends, Rami
and Biswas, Rupak
and Boixo, Sergio
and Brandao, Fernando G. S. L.
and Buell, David A.
and Burkett, Brian
and Chen, Yu
and Chen, Zijun
and Chiaro, Ben
and Collins, Roberto
and Courtney, William
and Dunsworth, Andrew
and Farhi, Edward
and Foxen, Brooks
and Fowler, Austin
and Gidney, Craig
and Giustina, Marissa
and Graff, Rob
and Guerin, Keith
and Habegger, Steve
and Harrigan, Matthew P.
and Hartmann, Michael J.
and Ho, Alan
and Hoffmann, Markus
and Huang, Trent
and Humble, Travis S.
and Isakov, Sergei V.
and Jeffrey, Evan
and Jiang, Zhang
and Kafri, Dvir
and Kechedzhi, Kostyantyn
and Kelly, Julian
and Klimov, Paul V.
and Knysh, Sergey
and Korotkov, Alexander
and Kostritsa, Fedor
and Landhuis, David
and Lindmark, Mike
and Lucero, Erik
and Lyakh, Dmitry
and Mandr{\`a}, Salvatore
and McClean, Jarrod R.
and McEwen, Matthew
and Megrant, Anthony
and Mi, Xiao
and Michielsen, Kristel
and Mohseni, Masoud
and Mutus, Josh
and Naaman, Ofer
and Neeley, Matthew
and Neill, Charles
and Niu, Murphy Yuezhen
and Ostby, Eric
and Petukhov, Andre
and Platt, John C.
and Quintana, Chris
and Rieffel, Eleanor G.
and Roushan, Pedram
and Rubin, Nicholas C.
and Sank, Daniel
and Satzinger, Kevin J.
and Smelyanskiy, Vadim
and Sung, Kevin J.
and Trevithick, Matthew D.
and Vainsencher, Amit
and Villalonga, Benjamin
and White, Theodore
and Yao, Z. Jamie
and Yeh, Ping
and Zalcman, Adam
and Neven, Hartmut
and Martinis, John M.},
title={Quantum supremacy using a programmable superconducting processor},
journal={Nature},
year={2019},
month={Oct},
day={01},
volume={574},
number={7779},
pages={505-510},
issn={1476-4687},
doi={10.1038/s41586-019-1666-5},
url={https://doi.org/10.1038/s41586-019-1666-5}
}

@inproceedings{Ji2018Pseudorandom,
author = {Ji, Zhengfeng and Liu, Yi-Kai and Song, Fang},
title = {Pseudorandom Quantum States},
year = {2018},
isbn = {978-3-319-96877-3},
publisher = {Springer-Verlag},
address = {Berlin, Heidelberg},
url = {https://doi.org/10.1007/978-3-319-96878-0_5},
doi = {10.1007/978-3-319-96878-0_5},
booktitle = {Advances in Cryptology – CRYPTO 2018: 38th Annual International Cryptology Conference, Santa Barbara, CA, USA, August 19–23, 2018, Proceedings, Part III},
pages = {126–152},
numpages = {27},
location = {Santa Barbara, CA, USA}
}

@misc{cleve2016nearlinearconstructionsexactunitary,
      title={Near-linear constructions of exact unitary 2-designs}, 
      author={Richard Cleve and Debbie Leung and Li Liu and Chunhao Wang},
      year={2016},
      eprint={1501.04592},
      archivePrefix={arXiv},
      primaryClass={quant-ph},
      url={https://arxiv.org/abs/1501.04592}, 
}

@misc{bourin2019russodyetheorempositivelinear,
      title={On the Russo-Dye Theorem for positive linear maps}, 
      author={Jean-Christophe Bourin and Eun-Young Lee},
      year={2019},
      eprint={1911.10573},
      archivePrefix={arXiv},
      primaryClass={math.FA},
      url={https://arxiv.org/abs/1911.10573}, 
}

@article{russo1966note,
title={A note on unitary operators in $C\backslash\sp\{\ast\}$-algebras},
author={Russo, Bernard and Dye, HA193530},
journal={Duke mathematical journal},
volume={33},
number={2},
pages={413--416},
year={1966}
}

@article{aaronson2004improved,
  title = {Improved simulation of stabilizer circuits},
  author = {Aaronson, Scott and Gottesman, Daniel},
  journal = {Phys. Rev. A},
  volume = {70},
  issue = {5},
  pages = {052328},
  numpages = {14},
  year = {2004},
  month = {Nov},
  publisher = {American Physical Society},
  doi = {10.1103/PhysRevA.70.052328},
  url = {https://link.aps.org/doi/10.1103/PhysRevA.70.052328}
}

@misc{alavi2014triplefactorisationsgenerallinear,
      title={Triple factorisations of the general linear group and their associated geometries}, 
      author={Seyed Hassan Alavi and John Bamberg and Cheryl E. Praeger},
      year={2014},
      eprint={1405.5276},
      archivePrefix={arXiv},
      primaryClass={math.GR},
      url={https://arxiv.org/abs/1405.5276}, 
}

@misc{schuster2025strongrandomunitariesfast,
      title={Strong random unitaries and fast scrambling}, 
      author={Thomas Schuster and Fermi Ma and Alex Lombardi and Fernando Brandao and Hsin-Yuan Huang},
      year={2025},
      eprint={2509.26310},
      archivePrefix={arXiv},
      primaryClass={quant-ph},
      url={https://arxiv.org/abs/2509.26310}, 
}

\end{document}